\documentclass[a4paper,notitlepage,fleqn,12pt]{article}
\usepackage[usenames,dvipsnames]{color}

\usepackage[tbtags]{amsmath}
\usepackage{amsfonts}
\usepackage{amsthm}
\usepackage{color}
\usepackage{amsfonts}
\usepackage{graphicx,amssymb}
\usepackage{multirow}
\usepackage{rotating}
\usepackage{mathrsfs}
\usepackage{epstopdf}
\usepackage{afterpage}
\usepackage{natbib}
\usepackage{tabularx}
\usepackage{xr}
\usepackage{subfig}
\usepackage{caption}
\usepackage{indentfirst}

\newtheorem{lemma}{Lemma}
\newtheorem{thm}{Theorem}

\newtheorem{assum}{Assumption}
\newtheorem{remark}{Remark}

\DeclareMathOperator*{\argmin}{argmin}

\DeclareMathOperator*{\sgn}{sgn}
\DeclareMathOperator*{\diag}{diag}
\DeclareMathOperator*{\cov}{cov}
\DeclareMathOperator*{\var}{var}
\newcommand{\bm}[1]{\mbox{\boldmath{$#1$}}}

\numberwithin{equation}{section}

\title{Quasi-maximum Likelihood Inference for Linear Double Autoregressive Models}
\author{Hua Liu, Songhua Tan and Qianqian Zhu\\
\emph{School of Statistics and Management}\\
\emph{Shanghai University of Finance and Economics}
}
\date{}
\begin{document}
\maketitle
\begin{abstract}
This paper investigates the quasi-maximum likelihood inference including estimation, model selection and diagnostic checking for linear double autoregressive (DAR) models, where all asymptotic properties are established under only fractional moment of the observed process. 
We propose a Gaussian quasi-maximum likelihood estimator (G-QMLE) and an exponential quasi-maximum likelihood estimator (E-QMLE) for the linear DAR model, and establish the consistency and asymptotic normality for both estimators. 
Based on the G-QMLE and E-QMLE, two Bayesian information criteria are proposed for model selection, and two mixed portmanteau tests are constructed to check the adequacy of fitted models. 
Moreover, we compare the proposed G-QMLE and E-QMLE with the existing doubly weighted quantile regression estimator in terms of the asymptotic efficiency and numerical performance.  
Simulation studies illustrate the finite-sample performance of the proposed inference tools, and a real example on the Bitcoin return series shows the usefulness of the proposed inference tools. 
\end{abstract}


\newpage

\section{Introduction}\label{Introduction}

Many conditional heteroscedastic models have been proposed to capture the time-varying volatility of financial and economic time series, with the autoregressive conditional heteroscedastic (ARCH) and generalized autoregressive conditional heteroscedastic (GARCH) models proving popular \citep{Engle1982,Bollerslev1986}. 
However, the conditional mean and conditional heteroscedasticity usually appear simultaneously in time series, and ignoring the conditional mean can lead to an inaccurate inference for the volatility \citep*{Li_Ling_McAleer2002}. 
Therefore, it is of vital importance to jointly model the conditional mean and volatility. 
The linear double autoregressive (DAR) model was proposed by \cite*{Zhu2018_LDAR} for this purpose, having the form
\begin{equation}\label{LDAR}
	y_{t}=\sum_{i=1}^{p} \alpha_{i} y_{t-i}+\eta_{t}\left(\omega+\sum_{i=1}^{p}\beta_{i} |y_{t-i}|\right),
\end{equation}
where $\omega>0$, $\beta_i\geq 0$ for $1\leq i\leq p$, and $\{\eta_t\}$ are independent and identically distributed ($i.i.d.$) random variables. Model \eqref{LDAR} assumes a linear structure for the conditional standard deviation, which makes it less sensitive to extreme values, thus leading to a more robust inference than those form models with a linear structure for the conditional variance; see also \cite{Taylor2008} and \cite{Xiao_Koenker2009}. 
Model \eqref{LDAR} has a novel property that it enjoys a larger parameter space than those of the AR and AR-ARCH models \citep*{Zhu2018_LDAR}. For example, with $p=1$, it can be stationary, even if $|\alpha_1|>1$, which is impossible for causal AR and AR-ARCH models. 

A doubly weighted quantile regression estimator (DWQRE) was introduced by \cite*{Zhu2018_LDAR} for model \eqref{LDAR}, with $\omega=1$ for identification. The DWQRE linearly combines the self-weighted quantile regression estimators at multiple quantile levels using weighting matrices, and requires only a fractional moment on the observed process $\{y_t\}$ to establish its asymptotic properties. 
This novel property leads to robust inferences using model \eqref{LDAR}, and thus can be used for heavy-tailed data. 
\cite*{Zhu2018_LDAR} showed that as the total number of quantile levels goes to infinity, the optimal DWQRE can approach the efficiency of the maximum likelihood estimator (MLE) under certain conditions, defined in \eqref{condition of DWQRE} in Section \ref{section2.3-comparison}. However, combining the self-weighted quantile regression estimators at infinite quantile levels is infeasible in practice, and, more importantly, condition \eqref{condition of DWQRE} implies that the DWQRE is, in general, less efficient than the MLE if the conditional mean structure exists. Furthermore, the two-step estimation procedure of the DWQRE is more complex than the one-step MLE or quasi-maximum likelihood estimation (QMLE) methods. 
To the best of our knowledge, no studies have examined the QMLE for model \eqref{LDAR}. This study addresses this gap in the literature. 

Model \eqref{LDAR} was originally motivated by the DAR model proposed by \cite{Ling2004,Ling2007}, which is defined as
\begin{equation}\label{DAR}
	y_{t}=\sum_{i=1}^{p} \alpha_{i} y_{t-i}+\varepsilon_{t}\sqrt{\omega+\sum_{i=1}^{p}\beta_{i} y_{t-i}^2},
\end{equation}
where $\omega>0$, $\beta_i\geq 0$ for $1\leq i\leq p$, and $\{\varepsilon_t\}$ are $i.i.d.$ random variables. Model \eqref{DAR} is a special case of the ARMA-ARCH models in \cite{weiss1986asymptotic}, and has been extended by several studies, including the threshold DAR \citep*{Li_Ling_Zakoian2015MTDAR,Li2016_TDAR}, mixture DAR \citep{Li_Zhu_Liu_Li2017}, linear DAR \citep*{Zhu2018_LDAR}, augmented DAR \citep*{Jiang2020}, and asymmetric linear DAR \citep{Tan_Zhu2022} models.  
In contrast to the linear DAR model \eqref{LDAR}, model \eqref{DAR} assumes a linear structure for the conditional variance, which may make it sensitive to extreme values. 
However, similarly to model \eqref{LDAR}, it enjoys a larger parameter space than those of the causal AR and AR-ARCH models. 
For the estimation of model \eqref{DAR}, \cite{Ling2007} proposed a Gaussian quasi-maximum likelihood estimator (G-QMLE), and established its asymptotic normality under a fractional moment of $y_t$ and $E(\varepsilon_t^4) <\infty$. 
To reduce the moment condition of $\varepsilon_t$ for a more robust estimation, \cite{zhu2013quasi-maximum} proposed an exponential quasi-maximum likelihood estimator (E-QMLE) for model \eqref{DAR}, where its asymptotic normality requires only that $E(\varepsilon_t^2)<\infty$.  
Note that we need only a finite fractional moment of the process $\{y_t\}$ in order to establish the asymptotic normality of the G-QMLE or E-QMLE, yielding a robust estimation for model \eqref{DAR}. 

The G-QMLE and E-QMLE essentially represent the least square estimation and the least absolute deviation estimation, respectively, and are popular for investigating the G-QMLE or E-QMLE for various time series models. For example, \cite{Aue2011} proposed the G-QMLE for random coefficient AR models, \cite{Francq_Zakoian2004} and \cite{Francq2019} studied the G-QMLE for GARCH and ARMA-GARCH models, \cite{Ling2007Self} proposed a self-weighted G-QMLE for ARMA-GARCH models, and \cite{Zhu_Ling2011} investigated the E-QMLE for ARMA-GARCH models. 
Note that the G-QMLE and E-QMLE for the GARCH or ARMA-GARCH models usually require a second or fourth moment condition on the observed process in order to establish the asymptotic normality, whereas we need only a finite fractional moment for the DAR model \eqref{DAR}. 
Therefore, we examine whether the robust property of QMLEs for model \eqref{DAR} can be preserved for the linear DAR model \eqref{LDAR}. Hopefully, the QMLEs for model \eqref{LDAR} are robust and more convenient in terms of computation than the DWQRE. 
Model selection and diagnosis are another two key elements of the classical Box--Jenkins procedure, and hence need to be investigated based on QMLEs for model \eqref{LDAR}. 
This study contributes to the literature in three ways.

First, we propose an E-QMLE for model \eqref{LDAR} in Section \ref{section2.2-E-QMLE}, and establish its asymptotic normality under $E(|y_t|^{\kappa})<\infty$ and $E(\eta_t^2)<\infty$. In particular, to derive the asymptotic normality of the E-QMLE, we adopt the bracketing method of \cite{pollard1985new} to overcome the difficulty of the nondifferentiable and nonconvex objective function; see also \cite{Zhu_Ling2011,zhu2013quasi-maximum}. 
For comparison with the E-QMLE, we introduce the G-QMLE for model \eqref{LDAR} in Section \ref{section2.1-G-QMLE}, and obtain its asymptotic normality under $E(|y_t|^{\kappa})<\infty$, for some $\kappa>0$ and $E(\eta_t^4)<\infty$. 
We also compare the asymptotic efficiency of the E-QMLE with that of the DWQRE in Section \ref{section2.3-comparison}. 
Simulation studies indicate that no estimator dominates in terms of asymptotic efficiency, but the E-QMLE is much more efficient than the G-QMLE, but slightly less efficient than the DWQRE when the data are more heavy tailed.  
Although all three estimators can be used to fit heavy-tailed data with $E(|y_t|^{\kappa})<\infty$, in practice, we suggest choosing the most suitable option according to the fat tailedness of the fitted residuals and the computational complexity. 
Because the proposed E-QMLE offers a good trade-off between robustness and computational complexity, it is preferred for most financial and economic time series with heavy tails.   
The real-data application in Section \ref{section5-realdata} further illustrates this point. 

Second, based on the E-QMLE, we propose a Bayesian information criterion (BIC) for model selection in Section \ref{section2.4-bic}, and show that it enjoys selection consistency without any moment condition on the process $\{y_t\}$. 
As the first stage of the Box--Jenkins procedure, order selection is crucial for fitting time series in practice, and the BIC is widely used to select orders for time series models \citep{Poskitt_Tremayne1983BICinTS,Cryer_Chan2008TSA}.
\cite{Schwarz1978} introduced the BIC for likelihood functions in the exponential family, \cite{machado1990model} extended the BIC to a wide class of likelihood functions, and \cite{machado1993robust} studied the BIC based on objective functions that define M-estimators. 
Note that the E-QMLE belongs to the class of M-estimators, motivating us to propose a BIC that is robust, in line with the estimation procedure. We further show that the proposed BIC is asymptotically consistent in terms of estimating the true order. 
Simulation studies indicate that, even if the sample size is not very large, the proposed BIC exhibits satisfactory performance.

Our third contribution is to construct a robust portmanteau test, in Section \ref{section3-diagnosis}, for checking the adequacy of the fitted models, with no moment conditions imposed on the process $\{y_t\}$. 
It is well known that diagnostic checking is important for time series modeling, and portmanteau tests are commonly used for this purpose; see \cite{Box_Jenkins_Reinsel2008}. For pure conditional mean models, the portmanteau test is usually constructed using the sample autocorrelation functions (ACFs) of the residuals \citep{Ljung_Box1978}, whereas we use the ACFs of the squared or absolute residuals for pure volatility models \citep{Li1994squared,Li2004,Li_Li2005}. 
\cite{Wong2005mixed_portmanteau} proposed a mixed portmanteau statistic for testing the adequacy of fitted time series models using residuals and squared residuals. 
To ensure the robustness of the test based on the E-QMLE, in the sense that its asymptotic properties require only $E(\eta_t^2)<\infty$,
we adopt the ACFs of the absolute rather than the squared residuals to check the adequacy of the volatility part. As a result, we propose a mixed portmanteau test based on the ACFs of the residuals and the absolute residuals to simultaneously detect misspecifications of the conditional mean and volatility in the fitted model. The asymptotic properties of the test are established in Section \ref{section3-diagnosis}.

Section \ref{section4-simulation} evaluates the finite-sample performance of the proposed inference tools by means of simulation, and  a real-data example is given in Section \ref{section5-realdata}. Section \ref{section6-conclusion} concludes the paper. All technical details and additional simulation results are relegated to the Supplementary Material. 
Throughout the paper, $\rightarrow_p$ ($\rightarrow_{\mathcal{L}}$) denotes convergence in probability (distribution), and $o_p(1)$ denotes a sequence of random variables converging to zero in probability.
The data from Section \ref{section5-realdata} and the programs used to analyze them are available from https://github.com/Tansonghua-sufe/Linear-double-autoregression. 

\section{Model estimation}\label{section2}
In this section, we propose an E-QMLE for model \eqref{LDAR} in Section \ref{section2.2-E-QMLE}. Then, we compare the E-QMLE with the G-QMLE and the DWQRE of \cite*{Zhu2018_LDAR} in Sections \ref{section2.1-G-QMLE}--\ref{section2.3-comparison}, respectively.   
Finally, order selection is discussed in Section \ref{section2.4-bic}.

\subsection{Exponential quasi-maximum likelihood estimation}\label{section2.2-E-QMLE}

Let $\bm \theta=(\bm\alpha^\prime,\bm \delta^\prime)^\prime$ be the unknown parameter vector of model \eqref{LDAR}, and let $\bm \theta_{0}=(\bm\alpha^\prime_{0},\bm \delta^\prime_{0})^{\prime}$ be the true parameter vector, where $\bm\alpha=(\alpha_1,\alpha_2,\ldots,\alpha_p)^\prime$ and $\bm \delta=(\omega,\beta_1,\beta_2,\ldots,\beta_p)^\prime$. Denote the parameter space by $\Theta=\Theta_{\alpha}\times\Theta_{\delta}$, where $\Theta_{\alpha} \subset \mathbb{R}^{p}$ and $\Theta_{\delta} \subset \mathbb{R}_{+}^{p+1}$, with $\mathbb{R}_{+} = (0,+\infty)$.
Assume that $\{y_1,\ldots, y_n\}$ are generated by model \eqref{LDAR}, with $\eta_t$ having a zero median and satisfying $E(|\eta_t|)=1$. When $\eta_t$ follows the standard double exponential distribution, the negative conditional log-likelihood function (ignoring a constant) can be written as 
\begin{equation}\label{qmlelog_likelihood}
	L_n^E(\bm \theta) =\dfrac{1}{n-p}\sum_{t=p+1}^{n}\ell_t^E(\bm \theta)\quad \text{and} \quad \ell_t^E(\bm\theta) = \ln h_t(\bm\delta)+\frac{|\varepsilon_t(\bm\alpha)|}{h_t(\bm\delta)},
\end{equation}
where 
\begin{equation}\label{epst_ht}
	\varepsilon_t(\bm\alpha) = y_t-\sum_{i=1}^{p}\alpha_{i}y_{t-i}\quad \text{and} \quad h_t(\bm\delta) = \omega  +\sum_{i=1}^{p}\beta_{i}|y_{t-i}|.
\end{equation} 
Let $\widehat{\bm \theta}_n=\argmin_{\bm \theta \in \Theta} L_n^E(\bm \theta)$. 
Because we do not assume that $\eta_t$ follows the standard double exponential distribution (i.e., the Laplace distribution), $\widehat{\bm \theta}_n$ is called the E-QMLE of $\bm \theta_0$; see also \cite{Zhu_Ling2011,zhu2013quasi-maximum}. 

\begin{assum}\label{assum1-compactness}
	$\Theta$ is compact, with $\underline{\omega} \leq \omega \leq \overline{\omega}$ and $\underline{\beta} \leq \beta_i \leq \overline{\beta}$, for $i=1,\ldots,p$, where $\underline{\omega},\overline{\omega},\underline{\beta},$ and $\overline{\beta}$ are some positive constants, and $\bm \theta_0$ is an interior point in $\Theta$.
\end{assum}
\begin{assum}\label{assum2-stationarity}
	$\{y_t:t=1,2,...\}$ is strictly stationary and ergodic, with $E(|y_t|^{\kappa})< \infty$, for some $0< \kappa \leq 1$. 
\end{assum} 
\begin{assum}\label{assum3EQMLE-moment and density}
	(i) $\eta_t$ has a zero median with $E(|\eta_t|)=1$; 
	(ii) the density function of $\eta_t$ is continuous and positive everywhere on $\mathbb{R}$ satisfying $\sup_{x\in \mathbb{R}} f(x)< \infty$; (iii) $E(\eta_t^2)<\infty$. 
\end{assum}

Assumption \ref{assum1-compactness} is standard in the literature on quasi-maximum likelihood estimation; see also \cite{Ling2007} and \cite{zhu2013quasi-maximum}. 
For general distributions of $\eta_t$, it is difficult to derive a necessary and sufficient condition for the strict stationarity of $y_t$, owing to the nonlinearity of model \eqref{LDAR}. A sufficient condition for Assumption \ref{assum2-stationarity} is given in Theorem 2 of \cite*{Zhu2018_LDAR}, that is, $\sum_{i=1}^{p}\max\{E(|\alpha_i-\beta_i\eta_t|^{\kappa}),E(|\alpha_i+\beta_i\eta_t|^{\kappa})\}<1$, for some $0< \kappa \leq 1$. 
Assumption \ref{assum3EQMLE-moment and density} is a general setup for establishing the consistency and asymptotic normality of the E-QMLE; see also \cite{zhu2013quasi-maximum}. 
\begin{thm}\label{thm1EQMLE-consistency}
	If Assumptions \ref{assum1-compactness}, \ref{assum2-stationarity}, and \ref{assum3EQMLE-moment and density}(i) hold, then $\widehat{\bm\theta}_n \to \bm\theta_0$ almost surely as $n \to \infty$. 	
\end{thm}

Let $\kappa_1=E(\eta_t)$ and $\kappa_2=E(\eta_t^2)-1$. Denote $\bm Y_{1t}=h_t^{-1}(\bm \delta_0)(y_{t-1},\ldots,y_{t-p})^\prime$ and $\bm Y_{2t}=h_t^{-1}(\bm \delta_0)(1,|y_{t-1}|,\ldots,|y_{t-p}|)^\prime$. 
Define the $(2p+1)\times(2p+1)$ matrices 
\begin{equation*}
	\Sigma = \diag \left \{f(0)E(\bm Y_{1t}\bm Y_{1t}^\prime),\frac{1}{2}E(\bm Y_{2t}\bm Y_{2t}^\prime)\right \} \; \text{and} \;
	\Omega = \begin{pmatrix} E(\bm Y_{1t}\bm Y_{1t}^\prime) & \kappa_1 E(\bm Y_{1t}\bm Y_{2t}^\prime) \\ \kappa_1 E(\bm Y_{2t}\bm Y_{1t}^\prime) & \kappa_2 E(\bm Y_{2t}\bm Y_{2t}^\prime) \end{pmatrix}.
\end{equation*}

\begin{thm}\label{thm2EQMLE-normality}
	If Assumptions \ref{assum1-compactness}--\ref{assum3EQMLE-moment and density} hold, then 
	\begin{itemize}
		\item[(i)] $\sqrt{n}(\widehat{\bm\theta}_n-\bm\theta_0)=O_p(1)$; 
		\item[(ii)] $\sqrt{n}(\widehat{\bm\theta}_n-\bm\theta_0) \rightarrow_{\mathcal{L}} N (\bm{0}, \Xi) \; \text{as}\; n\to \infty$, where $\Xi=\Sigma^{-1}\Omega\Sigma^{-1}/4$.
	\end{itemize}
\end{thm}
Theorem \ref{thm2EQMLE-normality} shows that the asymptotic normality of the proposed E-QMLE is established under a fractional moment of $y_t$, with $E(\eta_t^2)<\infty$, for model \eqref{LDAR}. Therefore, the E-QMLE is robust to heavy-tailed data, and can be used for $E(\eta_t^4)=\infty$ and $E(\eta_t^2)<\infty$, where the G-QMLE introduced in Section \ref{section2.1-G-QMLE} is no longer applicable. 
In addition, if $\kappa_1=0$, then the asymptotic covariance in Theorem \ref{thm2EQMLE-normality} reduces to the block diagonal matrix $\Gamma_0=\diag\{4f^2(0)E(\bm Y_{1t}\bm Y_{1t}^\prime), \kappa_2^{-1}E(\bm Y_{2t}\bm Y_{2t}^\prime)\}^{-1}$. 

To estimate the asymptotic covariance of $\widehat{\bm\theta}_n$, define the residuals fitted by the E-QMLE as $\widehat{\eta}_t=\varepsilon_t(\widehat{\bm\alpha}_n)/h_t(\widehat{\bm \delta}_n)$, and then $\widehat{\kappa}_1=(n-p)^{-1}\sum_{t=p+1}^{n}\widehat{\eta}_t$ and $\widehat{\kappa}_2=(n-p)^{-1}\sum_{t=p+1}^{n}\widehat{\eta}_t^2-1$. We can estimate the density function $f(x)$ using the kernel density estimator $\widehat{f}(x)=(nb_n)^{-1}\sum_{t=p+1}^{n}K((x-\widehat{\eta}_{t})/b_n)$, 
where $K(\cdot)$ is the kernel function and $b_n$ is the bandwidth. Finally, we use sample averages to replace the expectations in $\Sigma$ and $\Omega$, $\widehat{\bm\delta}_n$ to replace $\bm \delta_0$, $\widehat{\kappa}_i$ to replace $\kappa_i$, for $i=1,2$, and $\widehat{f}(0)$ to replace $f(0)$. Then, we can obtain estimates of $\Sigma$ and $\Omega$, denoted by $\widehat{\Sigma}_{n}$ and $\widehat{\Omega}_{n}$, respectively. 
Remark \ref{remark1} provides basic conditions on the kernel function $K(\cdot)$ and the bandwidth $b_n$ such that $\widehat{\Sigma}_{n}$ and $\widehat{\Omega}_{n}$ are consistent estimators of $\Sigma$ and $\Omega$, respectively; see the Supplementary Material for the proof. 

\begin{remark}\label{remark1}
	Suppose the conditions in Theorem \ref{thm2EQMLE-normality} hold and $\sup_x|f^\prime(x)|<\infty$. If there exists a positive number $L$ such that $|K(x)-K(y)|\leq L|x-y|$, for any $x,y$, and $b_n\rightarrow 0$ and $nb_n^{4}\rightarrow\infty$ as $n\rightarrow\infty$, then $\widehat{\Sigma}_{n}\rightarrow_p\Sigma$ and $\widehat{\Omega}_{n}\rightarrow_p\Omega$ as $n\rightarrow\infty$; see also Corollary 1 in \cite{zhu2013quasi-maximum}.
\end{remark}
Numerous choices for the kernel function $K(\cdot)$ and bandwidth $b_n$ satisfy the conditions in Remark \ref{remark1}. 
The kernel density estimator is robust to the selection of the kernel functions, but sensitive to that of the bandwidths. Thus, we suggest using the optimal bandwidth related to the selected kernel function in practice. 
Here, we use the Gaussian kernel function and its rule-of-thumb bandwidth $b_n=\text{0.9}n^{-1/5}\min\{s,\widehat{R}/\text{1.34}\}$ for our numerical studies in Sections \ref{section4-simulation}--\ref{section5-realdata}, where $s$ and $\widehat{R}$ are the sample standard deviation and interquartile of the residuals $\{\widehat{\eta}_{t}\}$, respectively; see also \cite*{Zhu2018_LDAR}.

\subsection{Comparison with the Gaussian quasi-maximum likelihood estimation}\label{section2.1-G-QMLE}

A Gaussian quasi-maximum likelihood estimation provides a popular QMLE called the G-QMLE. 
Assuming that $\{y_1,\ldots, y_n\}$ are generated by model \eqref{LDAR}, with $E(\eta_t)=0$ and $\var(\eta_t)=1$, 
the G-QMLE of $\bm \theta_0$ is defined as $\widetilde{\bm \theta}_n=(\widetilde{\bm\alpha}_n^\prime,\widetilde{\bm\delta}_n^{\prime})^\prime=\argmin_{\bm \theta \in \Theta} L_n^G(\bm \theta)$, where 
\[
L_n^G(\bm \theta) =\dfrac{1}{n-p}\sum_{t=p+1}^{n}\ell_t^G(\bm \theta)\quad \text{and} \quad \ell_t^G(\bm\theta) = \ln h_t(\bm\delta)+\frac{\varepsilon_t^2(\bm\alpha)}{2h_t^2(\bm\delta)},
\]
and $\varepsilon_t(\bm\alpha)$ and $h_t(\bm\delta)$ are defined as in \eqref{epst_ht}; see also \cite{Ling2007} and \cite{Tan_Zhu2022}. 
Instead of Assumption \ref{assum3EQMLE-moment and density} for the E-QMLE, we require the following assumption to establish the asymptotic properties of the G-QMLE. 
\begin{assum}\label{assum3G-QMLE-moment and density} 
	(i) $\eta_t$ has a zero mean and unit variance; 
	(ii) $E(\eta_t^{4})<\infty$. 
\end{assum}

Let $\kappa_3=E(\eta_t^{3})$ and $\kappa_4=E(\eta_t^{4})-1$. 
Define the $(2p+1)\times(2p+1)$ matrices 
\[
\Sigma_1 = \diag \left \{E(\bm Y_{1t}\bm Y_{1t}^{\prime}),2E(\bm Y_{2t}\bm Y_{2t}^{\prime})\right \} \; \text{and} \;
\Omega_1 = \begin{pmatrix} E(\bm Y_{1t}\bm Y_{1t}^{\prime}) & \kappa_3 E(\bm Y_{1t}\bm Y_{2t}^{\prime}) \\ \kappa_3 E(\bm Y_{2t}\bm Y_{1t}^{\prime}) & \kappa_4 E(\bm Y_{2t}\bm Y_{2t}^{\prime}) \end{pmatrix}.
\]
Similarly to Theorem 2 of \cite{Tan_Zhu2022}, we can show the consistency and asymptotic normality of the G-QMLE $\widetilde{\bm \theta}_n$ under fractional moments of $y_t$ for model \eqref{LDAR}. 

\begin{thm}\label{thm1G-QMLE-asymptotics}
	Suppose that Assumptions \ref{assum1-compactness}, \ref{assum2-stationarity}, and \ref{assum3G-QMLE-moment and density}(i) hold. Then, 
	\begin{itemize}
		\item[(i)] $\widetilde{\bm \theta}_{n} \rightarrow_{p} \bm \theta_{0}$ as $n\to \infty$; 
		\item[(ii)] furthermore, if Assumption \ref{assum3G-QMLE-moment and density}(ii) holds and the matrix $D=\left(\begin{matrix}
			1 & \kappa_3\\
			\kappa_3& \kappa_4
		\end{matrix}
		\right)$ is positive definite, then $\sqrt{n}(\widetilde{\bm\theta}_{n}-\bm\theta_{0}) \rightarrow_{\mathcal{L}} N(\bm 0, \Xi_1)$ as $n\to \infty$, where $\Xi_1=\Sigma_1^{-1}\Omega_1\Sigma_1^{-1}$.
	\end{itemize}
\end{thm}

We next compare the E-QMLE with the G-QMLE. Because the two estimators need different conditions on the innovation, we assume $\kappa_1=E(\eta_t)=0$, and reparametrize model \eqref{LDAR} under Assumption \ref{assum3EQMLE-moment and density} to ensure that the innovation term has a zero mean and unit variance. Specifically, we consider the following reparametrized model: 
\begin{equation}\label{LDAR-GQMLE}
	y_{t}=\sum_{i=1}^{p} \alpha_{i} y_{t-i}+\eta_{t}^{*}\left(\omega^{*}+\sum_{i=1}^{p}\beta_{i}^{*} |y_{t-i}|\right),
\end{equation}
where $\eta_t^{*}=\eta_t/\sqrt{E(\eta_t^{2})}$ and $\bm\delta^{*}=(\omega^{*},\beta_1^{*},\beta_2^{*},\ldots,\beta_p^{*})^\prime=\bm\delta\sqrt{E(\eta_t^{2})}$. Denote $\bm \theta^{*}=(\bm\alpha^\prime,\bm\delta^{*\prime})^\prime$ as the unknown parameter vector of model \eqref{LDAR-GQMLE} and $\bm \theta_{0}^{*}=(\bm\alpha^\prime_{0},\bm \delta_{0}^{*\prime})^{\prime}$ as the true value. Let $\widetilde{\bm\theta}_n^{*}$ be the G-QMLE of $\bm\theta_0^{*}$, $\kappa_3^*=E(\eta_t^{*3})$, $\kappa_4^*=E(\eta_t^{*4})-1$, and $\bm Y_{it}^*=\bm Y_{it}/\sqrt{E(\eta_t^{2})}$, for $i=1$ and 2. 
If $E(\eta_t^{*4})<\infty$ and Assumptions \ref{assum1-compactness}--\ref{assum2-stationarity} hold, then by Theorem \ref{thm1G-QMLE-asymptotics}, we have $\sqrt{n}(\widetilde{\bm\theta}_n^{*}-\bm\theta_0^{*})\to_{\mathcal{L}} N(\bm 0,\Xi_1^*)$ as $n\to\infty$, where $\Xi_1^*=(\Sigma_1^*)^{-1}\Omega_1^*(\Sigma_1^*)^{-1}$, with
\[
\Sigma_1^* = \diag \left \{E(\bm Y_{1t}^*\bm Y_{1t}^{*\prime}),2E(\bm Y_{2t}^*\bm Y_{2t}^{*\prime})\right \} \; \text{and} \;
\Omega_1^* = \begin{pmatrix} E(\bm Y_{1t}^*\bm Y_{1t}^{*\prime}) & \kappa_3^* E(\bm Y_{1t}^*\bm Y_{2t}^{*\prime}) \\ \kappa_3^* E(\bm Y_{2t}^*\bm Y_{1t}^{*\prime}) & \kappa_4^* E(\bm Y_{2t}^*\bm Y_{2t}^{*\prime}) \end{pmatrix}.
\] 
Note that $\bm\theta_0=R\bm\theta_0^{*}$, where $R=\diag\{I_p, [E(\eta_t^{2})]^{-1/2}I_{p+1}\}$, with $I_m$ being the $m\times m$ identity matrix. Then, we have $\sqrt{n}(\widetilde{\bm\theta}_n-\bm\theta_0)\to_{\mathcal{L}} N(\bm 0,R\Xi_1^*R^{\prime})$ as $n\to\infty$, where $\widetilde{\bm\theta}_n=R\widetilde{\bm\theta}_n^{*}$ is the G-QMLE of $\bm\theta_0$ for model \eqref{LDAR} under Assumption \ref{assum3EQMLE-moment and density}. As a result, it is sufficient to compare the asymptotic covariance $\Xi$ with $R\Xi_1^*R^{\prime}$ for some specific cases: 
\begin{itemize}
	\item[(i)] If $\eta_t$ follows the standard Laplace distribution, such that $\kappa_1=0, \kappa_2=1$, and $f(0)=1/2$, then the E-QMLE reduces to the MLE, with the asymptotic covariance $\Xi=\diag\{[E(\bm Y_{1t}\bm Y_{1t}^\prime)]^{-1}, [E(\bm Y_{2t}\bm Y_{2t}^\prime)]^{-1}\}$ attaining the Cram\'{e}r--Rao lower bound, indicating that the E-QMLE is more efficient than the G-QMLE.
	\item[(ii)] If $\eta_{t}$ is standard normal, such that $\kappa_3=0, \kappa_4=2$, and $\Omega_1=\Sigma_1$, then the G-QMLE reduces to the MLE with the asymptotic covariance $R\Xi_1^*R^{\prime}=\Sigma_1^{-1}$, attaining the Cram\'{e}r--Rao lower bound, and thus the G-QMLE is more efficient than the E-QMLE. 
	\item[(iii)] For other situations of $\eta_{t}$, it is difficult to compare the asymptotic relative efficiency (ARE) of $\widehat{\bm\theta}_n$ to $\widetilde{\bm\theta}_n$. However, given the true parameter vector $\bm\theta_0$ and the density function $f(\cdot)$, we can obtain a theoretical value of $f(0)$, and estimate the matrices in $\Xi$ and $R\Xi_1^*R^{\prime}$ using sample averages and empirical values of $\kappa_i$ or $\kappa_i^*$, based on a large generated sequence. Then, the ARE of $\widehat{\bm\theta}_n$ to $\widetilde{\bm\theta}_n$ can be calculated by $\text{ARE}(\widehat{\bm\theta}_n,\widetilde{\bm\theta}_n)=(|R\Xi_1^*R^{\prime}|/|\Xi|)^{1/(2p+1)}$, where $|\cdot|$ is the determinant of a matrix; see \cite{Serfling2009}.
	The simulation results in Section \ref{section4-simulation} indicate that the E-QMLE is asymptotically more efficient than the G-QMLE when the data are more heavy tailed. 
\end{itemize}

\subsection{Comparison with the DWQRE}\label{section2.3-comparison}

Consider the DWQRE of \cite*{Zhu2018_LDAR}. With the identification condition that $\omega=1$ in model \eqref{LDAR}, they consider the following reparametrized model:
\begin{equation}\label{reparameter}
	y_{t}=\sum_{i=1}^{p} \alpha_{i} y_{t-i}+\varepsilon_{t}\left(1+\sum_{i=1}^{p}\beta^{\star}_{i} |y_{t-i}|\right),
\end{equation}
where $\varepsilon_{t}=\omega\eta_t$ and $\beta_i^{\star}=\beta_i/\omega$. 
Denote $\bm\theta^{\star}=(\bm\alpha^\prime,\bm\beta^{\star\prime})^\prime$ as the unknown parameter vector of model \eqref{reparameter} and $\bm\theta^{\star}_0 = (\bm\alpha_0^{\prime},\bm\beta_0^{\star\prime})^\prime$ as the true value, where $\bm\beta^{\star} = (\beta_{1}^{\star},\ldots,\beta_{p}^{\star})^\prime$. 
The DWQRE of $\bm\theta_0^{\star}$ is defined as $\check{\bm\theta}^{\star}_n=\sum_{k=1}^K \pi_k \check{\bm\theta}_{\tau_k n}^{\star}$, which combines the self-weighted quantile regression estimators $\check{\bm\theta}_{\tau_k n}^{\star}$ at $K$ quantile levels using weighting matrices $\pi_k$, the optimal choices of which $\pi_k^{opt}$ are defined in Theorem 4 of \cite*{Zhu2018_LDAR}. 
The asymptotic properties of the DWQRE are also established under the finite fractional moment of $\{y_t\}$, which makes it possible to handle heavy-tailed data. 
Specifically, \cite*{Zhu2018_LDAR} show that the optimal DWQRE $\check{\bm\theta}^{\star opt}_n=\sum_{k=1}^K \check{\pi}_k^{opt} \check{\bm\theta}_{\tau_k n}^{\star}$ satisfies $\sqrt{n}(\check{\bm\theta}^{\star opt}_n-\bm\theta^{\star}_0)\to_{\mathcal{L}} N(\bm 0,\Xi_2)$ as $n\to\infty$, where $\Xi_2$ is defined in their Theorem 4, and $\check{\pi}_k^{opt}$ are the estimated optimal weighting matrices.  
Moreover, they verify that $\check{\bm\theta}^{\star opt}_n$ approaches the efficiency of the MLE if $K\to\infty$ and the following condition holds:
\begin{equation}\label{condition of DWQRE}
	E(h_t^{-1}\bm{Y}_{1t})=0 \quad\text{and}\quad E(\bm{Y}_{at}^{\prime}\bm{Y}_{1t})=0,
\end{equation}
where $\bm{Y}_{at}=h_t^{-1}(|y_{t-1}|,\ldots,|y_{t-p}|)^\prime$, with $h_t=1  +\sum_{i=1}^{p}\beta_{i0}^{\star}|y_{t-i}|$.   
Condition \eqref{condition of DWQRE} implies that the parameters in the conditional mean (i.e., $\bm\alpha$) and conditional scale (i.e., $\bm\beta^{\star}$) can be estimated separately, without loss of efficiency, which is satisfied when $\eta_t$ is symmetrically distributed about zero and all $\alpha_i$ in model \eqref{reparameter} are zero. 
However, for other general cases, the asymptotic covariance $\Xi_2$ cannot attain the Cram\'{e}r--Rao lower bound.  

Next, we compare the E-QMLE with the DWQRE. 
Note that if $\eta_t$ follows the standard Laplace distribution, then the E-QMLE reduces to the MLE, such that the asymptotic covariance $\Xi$ attains the Cram\'{e}r--Rao lower bound. Therefore, the E-QMLE is asymptotically more efficient than the DWQRE if $\eta_t$ follows the standard Laplace distribution. 
In contrast, the DWQRE is asymptotically more efficient than the E-QMLE only if Condition \eqref{condition of DWQRE} holds, with infinite quantile levels used for estimation, and $\eta_t$ does not follow the standard Laplace distribution. 
For general comparisons, denote the E-QMLE of model \eqref{reparameter} as $\widehat{\bm\theta}_n^{\star}=g(\widehat{\bm\theta}_n)$, where $g: \mathbb{R}^p\times\mathbb{R}_{+}^{p+1}\to \mathbb{R}^p\times\mathbb{R}_{+}^{p}$ is a measurable transformation, such that $g(\bm\theta)=\bm\theta^{\star}$. Then, using the delta method, we have $\sqrt{n}(\widehat{\bm\theta}_n^{\star}-\bm\theta_0^{\star})\to_{\mathcal{L}} N(\bm 0,\dot{g}(\bm\theta_0)\Xi\dot{g}^{\prime}(\bm\theta_0))$ as $n\to\infty$, where $\dot{g}(\bm\theta_0)$ is the first derivative of $g(\bm\theta_0)$, defined as the following $2p\times (2p+1)$ matrix: 
\[\dot{g}(\bm\theta_0)=\begin{pmatrix}
	I_p & 0_{p\times 1} & 0_{p\times p} \\
	0_{p\times p} & -\omega_0^{-2}\bm\beta_0 & \omega_0^{-1}I_p
\end{pmatrix},\]
where $0_{m\times n}$ is an $m\times n$ zero matrix. 
Therefore, we can use the ARE of $\widehat{\bm\theta}_n^{\star}$ to $\check{\bm\theta}^{\star opt}_n$, defined by $\text{ARE}(\widehat{\bm\theta}_n^{\star},\check{\bm\theta}^{\star opt}_n)=[|\Xi_2|/|\dot{g}(\bm\theta_0)\Xi\dot{g}^{\prime}(\bm\theta_0)|]^{1/(2p)}$, to compare the E-QMLE with the optimal DWQRE. Here, $\text{ARE}(\widehat{\bm\theta}_n^{\star},\check{\bm\theta}^{\star opt}_n)$ can be computed by simulation, as for $\text{ARE}(\widehat{\bm\theta}_n,\widetilde{\bm\theta}_n)$ in Section \ref{section2.1-G-QMLE}.  
The simulation results in Section \ref{section4-simulation} indicate that neither model dominates for general situations.  

In addition, the two-step estimation procedure makes the DWQRE more complex in terms of computation than QMLEs obtained using a one-step estimation.  
In particular, given the same bounded maximum number of iterations in the optimization, $O(n)$ operations are required to obtain the E-QMLE, while $O(nK)+O(K^2)$ operations are needed for the DWQRE.  
Clearly, the computational load of the DWQRE can become much larger than that of the E-QMLE as $K\to\infty$, which makes using infinite quantile levels to construct the DWQRE infeasible in practice.    
As a result, the E-QMLE is preferred, because it can be more efficient for moderately heavy-tailed data and the computation is simpler than that of the DWQRE.

\subsection{Model selection}\label{section2.4-bic}

For model \eqref{LDAR} fitted using the E-QMLE, we introduce the following BIC for selection the order $p$:
\begin{equation}\label{BIC-E-QMLE}
	\text{BIC}(p) = 2(n-p_{\max})L_n^E(\widehat{\bm\theta}_{n}^{p}) + (2p+1)\ln(n-p_{\max}),
\end{equation}
where $\widehat{\bm\theta}_{n}^{p}$ is the E-QMLE with the order set to $p$, $L_n^E(\widehat{\bm\theta}_{n}^{p}) =(n-p_{\max})^{-1}\sum_{t=p_{\max}+1}^{n}\ell_t^E(\widehat{\bm\theta}_{n}^{p})$, and $p_{\max}$ is a predetermined positive integer; see also \cite{machado1993robust} and \cite*{Zhu2018_LDAR}. 
Let $\widehat{p}_{n} = \arg\min_{1\leq p\leq p_{\max}}\text{BIC}(p)$. The selection consistency of the BIC is given in the following theorem.
\begin{thm}\label{thm3BIC} 
	Let $p_0$ be the true order and $p_{max}$ be a predetermined positive integer. Under the conditions of Theorem \ref{thm2EQMLE-normality}, if $p_{max}\geq p_0$, then $P(\widehat{p}_{n} = p_0)\rightarrow 1$ as $n\to \infty$.  
\end{thm}

Theorem \ref{thm3BIC} shows that the BIC in \eqref{BIC-E-QMLE} is robust in a similar way to the E-QMLE, that is, its selection consistency requires only $E(|y_t|^{\kappa})<\infty$, for any $\kappa>0$. The results of our simulation studies in Section \ref{section4-simulation} indicate that the BIC performs well in finite samples.

\begin{remark}\label{remark-BIC-of-GQMLE}
	The BIC can also be defined for model \eqref{LDAR} fitted using the G-QMLE. In this case, the BIC is $\text{BIC}^{G}(p) = 2(n-p_{\max})L_n^G(\widetilde{\bm\theta}_{n}^{p}) + (2p+1)\ln(n-p_{\max})$, where $\widetilde{\bm\theta}_{n}^{p}$ is the G-QMLE with the order set to $p$, and $L_n^G(\widetilde{\bm\theta}_{n}^{p}) =(n-p_{\max})^{-1}\sum_{t=p_{\max}+1}^{n}\ell_t^G(\widetilde{\bm\theta}_{n}^{p})$; see also \cite{Tan_Zhu2022}. Let $\widehat{p}_{n}^G = \arg\min_{1\leq p\leq p_{\max}}\text{BIC}^{G}(p)$. 	
	Similarly to Theorem 3 of \cite{Tan_Zhu2022}, we can prove the selection consistency of $\text{BIC}^{G}(p)$ under the conditions of Theorem \ref{thm1G-QMLE-asymptotics}. 
\end{remark}

\section{Model checking}\label{section3-diagnosis}

To check the adequacy of the linear DAR models at \eqref{LDAR} fitted using the E-QMLE, we construct a mixed portmanteau test to jointly detect possible misspecifications in the conditional mean and the conditional standard deviation; see also \cite{Wong2005mixed_portmanteau}. 
We can conduct a diagnostic test of the conditional mean by checking the significance of the sample ACFs of the residuals \citep{Ljung_Box1978}; a similar test of the conditional standard deviation can be done by checking the significance of the sample ACFs of the absolute residuals for robustness \citep{Li_Li2005}. 

The ACFs of $\{\eta_t\}$ and $\{|\eta_t|\}$ at lag $k$ are defined by $\rho_k=\cov(\eta_t, \eta_{t-k})/\var(\eta_t)$ and $\gamma_k=\cov(|\eta_t|, |\eta_{t-k}|)/\var(|\eta_t|)$, respectively.
If the data-generating process is specified correctly by model \eqref{LDAR}, then $\{\eta_t\}$ and $\{|\eta_t|\}$ are independent and identically distributed ($i.i.d.$), such that $\rho_k=0$ and $\gamma_k=0$ hold, for any $k\geq 1$. 
Define the error function as $\eta_{t}(\bm\theta)=\varepsilon_t(\bm \alpha)/h(\bm\delta)$. 
For model \eqref{LDAR} fitted using the E-QMLE, the corresponding residuals are computed as $\widehat{\eta}_{t}=\varepsilon_t(\widehat{\bm \alpha}_n)/h_t(\widehat{\bm \delta}_n)$. 
Accordingly, the residual ACF and absolute residual ACF at lag $k$ are defined as
\[
\widehat{\rho}_{k}=\dfrac{\sum_{t=p+k+1}^{n}(\widehat{\eta}_{t}-\widehat{\eta}_1)(\widehat{\eta}_{t-k}-\widehat{\eta}_1)}{\sum_{t=p+1}^{n}(\widehat{\eta}_{t}-\widehat{\eta}_1)^{2}} \;\text{and}\; \widehat{\gamma}_{k}=\dfrac{\sum_{t=p+k+1}^{n}(|\widehat{\eta}_{t}|-\widehat{\eta}_2)(|\widehat{\eta}_{t-k}|-\widehat{\eta}_2)}{\sum_{t=p+1}^{n}(|\widehat{\eta}_{t}|-\widehat{\eta}_2)^{2}},
\]
respectively, where $\widehat{\eta}_1=(n-p)^{-1}\sum_{t=p+1}^{n}\widehat{\eta}_t$ and $\widehat{\eta}_2=(n-p)^{-1}\sum_{t=p+1}^{n}|\widehat{\eta}_t|$.
Note that $\widehat{\rho}_{k}$ is the sample version of ${\rho}_{k}$, whereas $\widehat{\gamma}_{k}$ is the sample version of ${\gamma}_{k}$. 
If the value of $\widehat{\rho}_{k}$ (or $\widehat{\gamma}_{k}$) deviates far from zero, then the conditional mean (or standard deviation) structure in model \eqref{LDAR} may be misspecified. 

For a prespecified positive integer $M$, denote $\widehat{\bm \rho}=(\widehat{\rho}_{1},\ldots,\widehat{\rho}_{M})^{\prime}$ and $\widehat{\bm \gamma}=(\widehat{\gamma}_{1},\ldots,\widehat{\gamma}_{M})^{\prime}$. 
Let $\sigma_1^{2}=\var(\eta_t)$ and $\sigma_2^2=\var(|\eta_t|)$. 
Define the $M\times(2p+1)$ matrices $U_{\rho}=(\bm U_{\rho 1}^{\prime},\ldots,\bm U^{\prime}_{\rho M})^\prime$ and $U_{\gamma}=(\bm U^{\prime}_{\gamma 1},\ldots,\bm U^{\prime}_{\gamma M})^{\prime}$, where 
$\bm U_{\rho k}=-\left(E\left[ (\eta_{t-k}-\kappa_1) \bm {Y}_{1t}^{\prime}\right], \kappa_1 E\left[(\eta_{t-k}-\kappa_1)\bm Y_{2t}^\prime\right]\right)$ and $\bm U_{\gamma k}=-\left(\bm 0_{p}^{\prime}, E\left[(|\eta_{t-k}|-1) \bm{Y}_{2t}^{\prime}\right]\right)$, for $1\leq k\leq M$, with $\bm 0_{p}$ being a $p$-dimensional zero vector. 
Denote the $2M\times(2M+2p+1)$ matrix
\[
V=\left(\begin{array}{ccc}
	I_{M} & 0_{M\times M} & U_{\rho}/{\sigma}_1^2 \\
	0_{M\times M} & I_{M} & U_{\gamma}/{\sigma}_2^{2}
\end{array}\right).
\] 
Let $\bm G_{t}=(\bm Y_{1t}^\prime[I(\eta_t<0)-I(\eta_t>0)],\bm Y_{2t}^\prime(1-|\eta_t|))^\prime$ and $G=E(\bm v_{t} \bm v_{t}^{\prime})$, where  
\begin{align*}
	\bm v_{t}=&\left[(\eta_t-\kappa_1)(\eta_{t-1}-\kappa_1)/\sigma_1^2,\ldots,(\eta_{t}-\kappa_1)(\eta_{t-M}-\kappa_1)/\sigma_1^2,\right.\\
	&\left.(|\eta_{t}|-1)(|\eta_{t-1}|-1)/{\sigma}_2^{2},\ldots,(|\eta_{t}|-1)(|\eta_{t-M}|-1)/{\sigma}_2^{2},
	-\bm G_{t}^\prime\Sigma_{2}^{-1}/2\right]^\prime.
\end{align*}

\begin{thm}\label{thmACF} 
	Suppose model \eqref{LDAR} is specified correctly. Under the conditions of Theorem \ref{thm2EQMLE-normality}, then  $\sqrt{n}(\widehat{\bm\rho}^{\prime},\widehat{\bm\gamma}^{\prime})^{\prime}\rightarrow_{\mathcal{L}} N\left(\bm{0}, V G V^{\prime}\right)$ as $n\to \infty$. 
\end{thm}

Theorem \ref{thmACF} can be used to check the significance of $\widehat{\rho}_{k}$ or $\widehat{\gamma}_{k}$ individually. 
Consistent estimators of $V$ and $G$, denoted by $\widehat{V}$ and $\widehat{G}$, respectively, can be constructed by replacing the expectations with the sample averages and $\eta_t$ with $\widehat{\eta}_t$. 
Then, we can estimate the asymptotic covariance in Theorem \ref{thmACF}, and construct confidence intervals for $\widehat{\rho}_{k}$ and $\widehat{\gamma}_{k}$ accordingly. 

To check the first $M$ lags jointly, we construct the following portmanteau test statistic:
\begin{equation*}
	Q(M)=n\left(\begin{array}{l}
		\widehat{\bm \rho} \\
		\widehat{\bm \gamma}
	\end{array}\right)^{\prime}\left(\widehat{V} \widehat{G} \widehat{V}^\prime\right)^{-1}\left(\begin{array}{l}
		\widehat{\bm \rho} \\
		\widehat{\bm \gamma}
	\end{array}\right). 
\end{equation*}
Theorem \ref{thmACF} and the continuous mapping theorem imply that $Q(M)\to_{\mathcal{L}}\chi^2_{2M}$ as $n\to\infty$, where $\chi^2_{2M}$ is the chi-squared distribution with $2M$ degrees of freedom. 
Therefore, if $Q(M)$ exceeds the $(1-\tau)$th quantile of the $\chi^2_{2M}$ distribution, we can reject the null hypothesis that $\rho_k$ and $\gamma_k$ ($1 \leq k\leq M$) are jointly nonsignificant at level $\tau$.

\begin{remark}\label{remark-PT-of-GQMLE}	
	The diagnostic tools can also be derived for model \eqref{LDAR}, fitted using the G-QMLE. In this case, the residual ACF and absolute residual ACF at lag $k$ are defined as
	\[
	\widetilde{\rho}_{k}=\dfrac{\sum_{t=p+k+1}^{n}(\widetilde{\eta}_{t}-\widetilde{\eta}_1)(\widetilde{\eta}_{t-k}-\widetilde{\eta}_1)}{\sum_{t=p+1}^{n}(\widetilde{\eta}_{t}-\widetilde{\eta}_1)^{2}} \;\text{and}\; \widetilde{\gamma}_{k}=\dfrac{\sum_{t=p+k+1}^{n}(|\widetilde{\eta}_{t}|-\widetilde{\eta}_2)(|\widetilde{\eta}_{t-k}|-\widetilde{\eta}_2)}{\sum_{t=p+1}^{n}(|\widetilde{\eta}_{t}|-\widetilde{\eta}_2)^{2}},
	\]
	respectively, where $\widetilde{\eta}_{t}=\varepsilon_t(\widetilde{\bm \alpha}_n)/h_t(\widetilde{\bm \delta}_n)$, $\widetilde{\eta}_1=(n-p)^{-1}\sum_{t=p+1}^{n}\widetilde{\eta}_t$, and $\widetilde{\eta}_2=(n-p)^{-1}\sum_{t=p+1}^{n}|\widetilde{\eta}_t|$.
	For a given positive integer $M$, denote $\widetilde{\bm \rho}=(\widetilde{\rho}_{1},\ldots,\widetilde{\rho}_{M})^{\prime}$ and $\widetilde{\bm \gamma}=(\widetilde{\gamma}_{1},\ldots,\widetilde{\gamma}_{M})^{\prime}$.
	Let $\tau_1=E[\sgn(\eta_t)]$ and $\tau_2=E(|\eta_t|)$.  
	Define the $M\times(2p+1)$ matrices $U_{\rho}^G=(\bm U_{\rho 1,G}^{G\prime},\ldots,\bm U^{G\prime}_{\rho M})^\prime$ and $U_{\gamma}^G=(\bm U^{G\prime}_{\gamma 1},\ldots,\bm U^{G\prime}_{\gamma M})^{\prime}$, where $\bm U_{\rho k}^G=-\left(E(\eta_{t-k} \bm {Y}_{1t}^{\prime}), \bm 0_{p+1}^{\prime}\right)$ and $\bm U_{\gamma k}^G=-(\tau_1E[(|\eta_{t-k}|-\tau_2) \bm{Y}_{1t}^{\prime}], \tau_2E[(|\eta_{t-k}|-\tau_2) \bm{Y}_{2t}^{\prime}])$, for $1\leq k\leq M$.	
	Let $\bm G_{1t}=(-\bm Y_{1t}^\prime\eta_t,\bm Y_{2t}^\prime(1-\eta_t^2))^\prime$ and $G_1=E(\bm v_{1t} \bm v_{1t}^{\prime})$, where $\bm v_{1t}=(\eta_t\eta_{t-1},\ldots,\eta_{t}\eta_{t-M},(|\eta_{t}|-\tau_2)(|\eta_{t-1}|-\tau_2)/{\sigma}_2^{2},\ldots,(|\eta_{t}|-\tau_2)(|\eta_{t-M}|-\tau_2)/{\sigma}_2^{2},
	-\bm G_{1t}^\prime\Sigma_1^{-1})^\prime$.
	Similarly to Theorem 7 of \cite{Tan_Zhu2022},  $\sqrt{n}(\widetilde{\bm\rho}^{\prime},\widetilde{\bm\gamma}^{\prime})^{\prime}\rightarrow_{\mathcal{L}} N\left(\bm{0}, V_1 G_1 V_1^{\prime}\right)$ as $n\to \infty$ if model \eqref{LDAR} is specified correctly and the conditions of Theorem \ref{thm1G-QMLE-asymptotics} hold, where $V_1$ is defined as $V$, with $U_{\rho}$ and $U_{\gamma}$ replaced with $U_{\rho}^G$ and $U_{\gamma}^G$, respectively. Then, we can check the significance of $\widetilde{\rho}_{k}$ or $\widetilde{\gamma}_{k}$ individually. 
	In addition, we can use the following portmanteau test statistic to check the first $M$ lags jointly:
	\[
	Q^G(M)=n\left(\begin{array}{l}
		\widetilde{\bm \rho} \\
		\widetilde{\bm \gamma}
	\end{array}\right)^{\prime}\left(\widehat{V}_1 \widehat{G}_1 \widehat{V}_1^\prime\right)^{-1}\left(\begin{array}{l}
		\widetilde{\bm \rho} \\
		\widetilde{\bm \gamma}
	\end{array}\right),  
	\]
	where $\widehat{V}_1$ and $\widehat{G}_1$ are consistent estimators of $V_1$ and $G_1$, respectively.
	It can be shown that $Q^G(M)\to_{\mathcal{L}}\chi^2_{2M}$ as $n\to\infty$ under the null hypothesis that $\rho_k$ and $\gamma_k$ are jointly nonsignificant for $1 \leq k\leq M$ at level $\tau$.
\end{remark}	

\begin{remark}\label{remark-selection-of-M}
	In practice, the choice of $M$ may affect the performance of the portmanteau tests $Q(M)$ and $Q^G(M)$. For other portmanteau tests, the selection of $M$ remains an open issue. 
	Some general rules have been provided to choose $M$ for Box--Pierce and Ljung--Box tests. For example, \cite{box2015_TSA} recommend taking values of $M$ between 10 and 20, and \cite{tsay2005AoFTS} suggests using several choices of $M$ and a general rule of $M\approx \ln(n)$, owing to its satisfactory power performance in simulation studies.
	
	Motivated by \cite{tsay2005AoFTS}, we conduct simulation studies to evaluate the size and power of $Q(M)$ and $Q^G(M)$ with respect to the sample size $n$ and the lag order $M$. We find that the size of a test is insensitive to $n$, and that the power is linearly increasing with respect to $n$. Moreover, the choice of $M>20$ usually makes $Q(M)$ and $Q^G(M)$ undersize, and the logarithmic power is linearly decreasing with respect to $M$; see Section S1 of the Supplementary Material for more details. 
	In general, $M$ should be large enough to capture possible correlations among residuals and absolute residuals, but  not be too large because of the resulting power loss.
	Thus, we suggest using multiple choices of $M$, such as $(M_1,M_2,\ldots,M_{J})=(\lfloor{\ln(n)}\rfloor,2\lfloor{\ln(n)}\rfloor,\ldots,J\lfloor{\ln(n)}\rfloor )$, where $\lfloor{x}\rfloor$ denotes the largest integer not greater than $x$, and $J$ is the maximum value of $j$, such that $M_{j}=j\lfloor{\ln(n)}\rfloor\leq 20$.
\end{remark}	

\section{Simulation experiments}\label{section4-simulation}

\subsection{E-QMLE and G-QMLE}\label{sim-QMLE}
The first experiment examines the finite-sample performance of the E-QMLE $\widehat{\bm \theta}_n$ and the G-QMLE $\widetilde{\bm \theta}_n$ in Sections \ref{section2.2-E-QMLE}--\ref{section2.1-G-QMLE}, respectively, for which the data-generating process is
\[
y_t = 0.5 y_{t-1}+\eta_t(1+0.4|y_{t-1}|),
\]
where $\{\eta_t\}$ are $i.i.d.$ normal, Laplace or Student’s $t_3$ distributed random variables. Here, $\{\eta_t\}$ are standardized with median zero and $E(|\eta_t|)=1$ to evaluate the E-QMLE $\widehat{\bm \theta}_n$, and $\{\eta_t\}$ are standardized with mean zero and $\var(\eta_t)=1$ to evaluate the G-QMLE $\widetilde{\bm \theta}_n$. 
The sample size is set to $n = 500$ or 1000, with 1000 replications for each sample size. 

Table \ref{tableE-QMLE} reports the biases, empirical standard deviations (ESDs), and asymptotic standard deviations (ASDs) of $\widehat{\bm \theta}_n$ and $\widetilde{\bm \theta}_n$ for different innovation distributions and sample sizes. 
We find that as the sample size increases, most of the biases, ESDs, and ASDs of both estimators $\widehat{\bm \theta}_n$ and $\widetilde{\bm \theta}_n$ become smaller, and the ESDs become closer to the corresponding ASDs. 
For the E-QMLE $\widehat{\bm \theta}_n$, the ESDs and ASDs of the scale-type estimators $\widehat{\omega}_n$ and $\widehat{\beta}_n$ increase as the distribution of $\eta_t$ becomes more heavy tailed, while those of the location-type estimator $\widehat{\alpha}_n$ are smallest for the Laplace distribution. 
The mixed performance of the E-QMLE is probably because the heavier tail of $\{\eta_t\}$ makes the E-QMLE less efficient, but it becomes more efficient as it reduces to the MLE, if $\eta_t$ follows the Laplace distribution. For the G-QMLE $\widetilde{\bm \theta}_n$, the ESDs and ASDs increase as the distribution of $\eta_t$ becomes more heavy tailed. This is expected, because the G-QMLE reduces to the MLE if $\eta_t$ follows the normal distribution, and the G-QMLE becomes less efficient as the tail of $\{\eta_t\}$ becomes heavier. 
Note that if $\eta_t$ follows the Student's $t_3$ distribution, then $E(\eta_t^4)=\infty$ and the G-QMLE is not applicable, which results in the inferior performance of the G-QMLE in this case. 
Similar observations can be found for other innovation distributions in the Supplementary Material.

\subsection{Asymptotic efficiency comparison}\label{sim-ARE}

The second experiment compares the ARE of the E-QMLE with those of the G-QMLE and DWQRE. We generate a sequence of sample size $n=10000$ from the following model:
$$
y_t=0.1y_{t-1}+\varepsilon_t(1+0.2|y_{t-1}|),
$$
where $\{\varepsilon_t\}$ are $i.i.d.$ random variables with the mixture distribution and the probability density function (pdf)
$$
f(x)=(1-\delta) \phi(x)+ \delta m(x),
$$
where $\delta\in [0,1]$ is a constant, $\phi(x)$ is the pdf of $N(0,1)$, and $m(x)$ is the pdf of $N(0,6)$, standard Laplace, or $t_3$ distribution. 

Figure \ref{fig_efficiency} plots the $\text{ARE}(\widehat{\bm\theta}_n,\widetilde{\bm\theta}_n)$ and $\text{ARE}(\widehat{\bm\theta}_n^{\star},\check{\bm\theta}^{\star opt}_n)$ defined in Sections \ref{section2.1-G-QMLE} and \ref{section2.3-comparison} for $\delta = k/20$, with $k=0,1,\ldots,20$, and different settings of $m(x)$, where the optimal DWQRE $\check{\bm\theta}^{\star opt}_n$ is obtained using $K=9$ quantile levels.
We have the following findings: 
(1) $\text{ARE}(\widehat{\bm\theta}_n,\widetilde{\bm\theta}_n)$ can be either larger or smaller than one, which suggests that neither the E-QMLE nor the G-QMLE dominate; furthermore the G-QMLE is more efficient than the E-QMLE when $\varepsilon_t$ is closer to normal, whereas the E-QMLE becomes much more efficient as $\varepsilon_t$ becomes more heavy tailed; 
(2) $\text{ARE}(\widehat{\bm\theta}_n^{\star},\check{\bm\theta}^{\star opt}_n)$ can be either greater or less than one, which implies that neither the E-QMLE nor the DWQRE dominate; 
furthermore, the E-QMLE is more efficient than the DWQRE when $\varepsilon_t$ approaches the Laplace distribution, but becomes less efficient than the DWQRE when $\varepsilon_t$ becomes more heavy tailed; 
(3) when $\delta=1$ and $m(x)$ is the pdf of a standard Laplace distribution, then $\text{ARE}(\widehat{\bm\theta}_n,\widetilde{\bm\theta}_n)>1$ and  $\text{ARE}(\widehat{\bm\theta}_n^{\star},\check{\bm\theta}^{\star opt}_n)>1$, indicating that the E-QMLE is the most efficient. This is because the E-QMLE reduces to the MLE when $\varepsilon_t$ follows a standard Laplace distribution, and thus its asymptotic covariance attains the Cram\'{e}r--Rao lower bound;  
and (4) when $\delta=0$, then $\text{ARE}(\widehat{\bm\theta}_n,\widetilde{\bm\theta}_n)<1$, implying that the G-QMLE is the most efficient. This is because the G-QMLE reduces to the MLE when $\varepsilon_t\sim N(0,1)$, such that its asymptotic covariance attains the Cram\'{e}r--Rao lower bound.

\subsection{Model selection}\label{sim-BIC}

In the third experiment, we evaluate the performance of the proposed model selection methods in Section \ref{section2.4-bic}, where the data-generating process is
\[
y_t = 0.1 y_{t-1}+0.2y_{t-2}+\eta_t(1+0.1|y_{t-1}|+ 0.2|y_{t-2}|),
\]
and the innovations $\{\eta_t\}$ are defined as in the first experiment. 
We consider three sample sizes, $n= 300,500$, and 1000, and generate 1000 replications for each sample size. 
The BIC in \eqref{BIC-E-QMLE} and BIC$^G$ in Remark \ref{remark-BIC-of-GQMLE} are employed to select the order $p$ with $p_{\max}=5$. As a result, the underfitted, correctly selected, and overfitted models by BIC (or BIC$^G$) correspond to $\widehat{p}_{n}$ (or $\widehat{p}_{n}^G$) being 1, 2, and greater than 2, respectively. 

Table \ref{tableBIC} provides the percentages of underfitting, correct selection, and overfitting cases by the BIC and BIC$^G$. Both BICs select the correct model in most of the replications when the sample size is as small as $n = 300$, and their performance improves as the sample size increases. 
Moreover, in terms of different distributions for the innovation $\eta_t$, the BIC in \eqref{BIC-E-QMLE} performs best when $\eta_t$ follows the Laplace distribution, whereas the BIC$^G$ performs best when $\eta_t$ follows the normal distribution, especially for small sample sizes. This is expected, because the E-QMLE reduces to the MLE if $\eta_t$ follows the Laplace distribution, whereas the G-QMLE reduces to the MLE if $\eta_t$ is normally distributed. Thus, no model misspecifications appear when deriving the BIC and BIC$^G$, respectively, for these two cases.   
In addition, owing to the inferior performance of the G-QMLE when $\eta_t$ follows the Student's $t_3$ distribution, BIC$^G$ performs worst in this situation. 
We also consider other innovation distributions for both BICs in Section S1 of the Supplementary Material; the findings support those reported above.

\subsection{Portmanteau tests}\label{sim-Portmanteau tests}

In the fourth experiment, we study the proposed mixed portmanteau tests $Q(M)$ and $Q^G(M)$. The data are generated from
\[
y_t = 0.1y_{t-1}+c_1y_{t-2}+\eta_t(1+0.2|y_{t-1}|+c_2|y_{t-2}|),
\]
where all other settings are the same as those in the first experiment. 
We fit a linear DAR model with $p=1$ using the exponential or Gaussian quasi-maximum likelihood estimation. Here, the case of $c_1=c_2=0$ corresponds to the size of the test, $c_1 \neq 0$ corresponds to misspecifications in the conditional mean, and $c_2>0$ corresponds to misspecifications in the conditional standard deviation. Two departure levels, 0.1 and 0.3, are considered for both $c_1$ and $c_2$. 

The rejection rates of $Q(6)$ and $Q^G(6)$ at the $5\%$ significance level are summarized in Table \ref{tablechecking}. We have the following findings: 
First, all sizes are close to the nominal rate, except for those of $Q^G(6)$ in the $t_3$ case, and the power increases as the sample size $n$ or the departure level increases. This is expected, because the G-QMLE is not applicable if $\eta_t$ follows the $t_3$ distribution with $E(\eta_t^4)=\infty$, which makes the size inaccurate.
Second, for the same level of departures, $Q(6)$ and $Q^G(6)$ are more powerful in terms of detecting the misspecification in the conditional mean ($c_1\neq 0, c_2=0$) than they are in doing so in the conditional standard deviation ($c_1=0$, $c_2>0$).
Third, as the innovation distribution becomes more heavy tailed, $Q(6)$ and $Q^G(6)$ perform worse in terms of detecting misspecifications in the conditional standard deviation. 
This seems consistent with the results in the first experiment, where the estimation performance of the scale parameters $\omega_0$ and $\beta_0$ worsens as the innovation distribution becomes more heavy tailed. 
In addition, we also evaluate $Q(M)$ and $Q^G(M)$ with $M=12$ and 18 in Section S1 of the Supplementary Material; once again, the findings are similar.

\section{An empirical example}\label{section5-realdata}

In this section, we apply the proposed inference tools to the weekly closing prices of Bitcoin (BTC) from July 18, 2010, to August 16, 2020, with 527 observations in total. We focus on their log returns (after mean adjustment), denoted by $\{y_t\}$. The time plot of $\{y_t\}$ in Figure \ref{pred} suggests evidence of volatility clustering, and the kurtosis of $\{y_t\}$ is $9.3$, indicating that the tail of $\{y_t\}$ is much heavier than that of a normal distribution. 
Moreover, we can determine the autocorrelation from the sample partial autocorrelation functions (PACFs) of both $\{y_t\}$ and $\{|y_t|\}$. 
Therefore, we fit the data set $\{y_t\}$ using a linear DAR model. 

We first employ the exponential quasi-maximum likelihood estimation method in Section \ref{section2.2-E-QMLE} to fit $\{y_t\}$. Based on $p_{\max}=10$, the proposed BIC in \eqref{BIC-E-QMLE} selects the order $p=3$ for the linear DAR model, and the fitted model is 
\begin{align}\label{E-QMLEestimation}
	&y_t=0.0815_{0.0504}y_{t-1}+0.1401_{0.0487}y_{t-2}+0.0693_{0.0471}y_{t-3}+\widehat{\eta}_t \widehat{\sigma}_t,\nonumber\\
	&\widehat{\sigma}_t=0.0435_{0.0065}+ 0.2192_{0.0664}|y_{t-1}|+0.1895_{0.0645}|y_{t-2}|+0.1616_{0.0624}|y_{t-3}|,
\end{align}
where the subscripts are the standard errors of the estimated coefficients, and $\{\widehat{\sigma}_t\}$ and $\{\widehat{\eta}_t\}$ are the fitted volatilities and residuals, respectively. 
The QQ plots of the fitted residuals $\{\widehat{\eta}_t\}$ against Students' $t_2$, $t_3$, and $t_4$ distributions are shown in Figure \ref{QQplots-E-QMLE}. Here, the residuals are approximately $t_3$ distributed and the tail is heavier than $t_4$, but much lighter than $t_2$, possibly indicating that $E(\eta_t^2)<\infty$ and $E(\eta_t^4)=\infty$. 
We further employ the Kernelized Stein discrepancy (KSD) test proposed by \cite{Luo_Zhu2021} to check whether $\eta_t$ follows the standard Laplace or normal distribution. To calculate the KSD test statistic defined by (2.7) of \cite{Luo_Zhu2021}, we choose $n_0=n=526$ and use the Gaussian kernel $k(x,y)=\exp\{-\|x-y\|^{2}/(2\sigma^{2})\}$, with $\sigma$ being the median of the residual distance. The $p$-value is calculated using the parametric bootstrap, and is $0.37$ for the standard Laplace distribution test, and less than 0.01 for the normal distribution test, suggesting that $\eta_t$ follows the standard Laplace distribution and the E-QMLE may reduce to the MLE for the data $\{y_t\}$.
Moreover, we perform mixed portmanteau tests $Q(M)$ for $M=6,12$, and $18$, as in Section \ref{section3-diagnosis}, and their $p$-values are $0.56, 0.71$, and $0.19$, respectively. 
In addition, Figure \ref{acf-E-QMLE} plots the residual ACFs $\widehat{\rho}_k$ and $\widehat{\gamma}_k$ up to lag $18$, all of which fall within their corresponding 95\% pointwise CIs, except for $\widehat{\gamma}_3$, which is slightly beyond its 95\% CI. 
Clearly, almost all residual ACFs are nonsignificant, both individually and jointly, at the 5\% significance level, and hence the fitted model at \eqref{E-QMLEestimation} is adequate.

For comparison, the Gaussian quasi-maximum likelihood estimation method in Section \ref{section2.1-G-QMLE} is also used to fit $\{y_t\}$. Based on $p_{\max}=10$, the BIC$^G$ in Remark \ref{remark-BIC-of-GQMLE} selects the same order $p=3$ for the linear DAR model, and obtains the following fitted model:
\begin{align}\label{G-QMLEestimation}
	&y_t=0.1098_{0.0579}y_{t-1}+0.1268_{0.0547}y_{t-2}+0.1733_{0.0586}y_{t-3}+\widetilde{\eta}_t \widetilde{\sigma}_t,\nonumber\\
	&\widetilde{\sigma}_t=0.0821_{0.0146}+ 0.2348_{0.1324}|y_{t-1}|+0.1674_{0.1260}|y_{t-2}|+0.2519_{0.1348}|y_{t-3}|,
\end{align}
where the subscripts are the standard errors of the estimated coefficients, and $\{\widetilde{\sigma}_t\}$ and $\{\widetilde{\eta}_t\}$ are the fitted volatilities and residuals, respectively.  
The diagnosis from the residuals $\{\widetilde{\eta}_t\}$ in the Supplementary Material indicates that the fitted model at \eqref{G-QMLEestimation} is adequate. 
Moreover, we estimate the linear DAR model of order three using the DWQRE of \cite*{Zhu2018_LDAR} for comparison. 
To facilitate a comparison between the E-QMLE and DWQRE, the volatility coefficients of the fitted model using the DWQRE are reparametrized to ensure $E(|\eta_t|)=1$. As a result, the fitted model using the DWQRE method based on quantile levels $\tau_k=k/10$, for $k=1,\ldots,9$, is given by
\begin{align}\label{DWQREestimation}
	&y_t=0.1311_{0.0447}y_{t-1}+0.0854_{0.0394}y_{t-2}+0.0627_{0.0357}y_{t-3}+\check{\eta}_t \check{\sigma}_t,\nonumber\\
	&\check{\sigma}_t=0.0384+ 0.3329_{0.0972}|y_{t-1}|+0.2020_{0.0762}|y_{t-2}|+0.1209_{0.0630}|y_{t-3}|,
\end{align}
where the subscripts are the standard errors of the estimated coefficients, $\{\check{\sigma}_t\}$ are the fitted volatilities, and $\{\check{\eta}_t\}$ are the standardized residuals, such that $E(|\eta_t|)=1$. 
To compare the efficiency of the three estimation methods, we approximate the $\text{ARE}(\widehat{\bm\theta}_n,\widetilde{\bm\theta}_n)$ and $\text{ARE}(\widehat{\bm\theta}_n^{\star},\check{\bm\theta}^{\star opt}_n)$ defined in Sections \ref{section2.1-G-QMLE} and \ref{section2.3-comparison}, respectively, using parameter estimates and sample averages. We have $\widehat{\text{ARE}}(\widehat{\bm\theta}_n,\widetilde{\bm\theta}_n)\approx 1.6223$ and $\widehat{\text{ARE}}(\widehat{\bm\theta}_n^{\star},\check{\bm\theta}^{\star opt}_n)\approx 0.8967$, implying that the E-QMLE is slightly less efficient than the DWQRE and that both are more efficient than the G-QMLE for fitting $\{y_t\}$.  
In addition, note that the conditional mean structures of the three fitted models \eqref{E-QMLEestimation}--\eqref{DWQREestimation} are significant at the 5\% significance level, which suggests that the BTC market was not efficient during the examined period \citep{Urquhart_2016}.

For financial time series, an important application of linear DAR models is to forecast risk measures, such as the value-at-risk (VaR). The VaR is actually a tail quantile of the loss series' conditional distribution, and thus the $\tau$th conditional quantile of $y_t$, denoted by $Q_{y_{t}}(\tau\mid \mathcal{F}_{t-1})$, is the negative $\tau$th VaR.  
To examine the forecasting performance of the linear DAR model estimated using the E-QMLE and G-QMLE, we conduct one-step-ahead predictions using a rolling forecasting procedure, with a fixed moving window of size $350$. 
Specifically, we estimate the linear DAR model using the E-QMLE (or G-QMLE) for each moving window, and calculate the one-week-ahead forecast of the $\tau$th conditional quantile of $y_{t+1}$ using $\widehat{Q}_{y_{t+1}}(\tau\mid \mathcal{F}_{t})=\widehat{\mu}_{t+1}+\widehat{\sigma}_{t+1}\widehat{b}_{\tau}$ (or $\widetilde{Q}_{y_{t+1}}(\tau\mid \mathcal{F}_{t})=\widetilde{\mu}_{t+1}+\widetilde{\sigma}_{t+1}\widetilde{b}_{\tau}$), where $\widehat{\mu}_{t+1}$ (or $\widetilde{\mu}_{t+1}$) and $\widehat{\sigma}_{t+1}$ (or $\widetilde{\sigma}_{t+1}$) are the predicted conditional mean and standard deviation, respectively, using the E-QMLE (or G-QMLE), and $\widehat{b}_{\tau}$ (or $\widetilde{b}_{\tau}$) is the $\tau$th sample quantile of the residuals $\{\widehat{\eta}_t\}$ (or $\{\widetilde{\eta}_t\}$).  
For example, the rolling one-week-ahead forecasts at $\tau=5\%$ are displayed in Figure \ref{pred}. Here, the negative VaRs based on the E-QMLE and G-QMLE are close to each other and change according to the volatility of the data, and $y_t$ occasionally falls below its one-week negative VaR forecast.

We next compare the forecasting performance of the proposed E-QMLE and G-QMLE with that of the DWQRE. 
We conduct a rolling forecasting procedure with a fixed moving window of size $350$ for the DWQRE approach, and compute the one-week-ahead forecast of the $\tau$th conditional quantile of $y_{t+1}$ using $\check{Q}_{y_{t+1}}(\tau\mid \mathcal{F}_{t})=\check{\mu}_{t+1}+\check{\sigma}_{t+1}\check{b}_{\tau}$, where $\check{\mu}_{t+1}$ and $\check{\sigma}_{t+1}$ are the predicted conditional mean and conditional standard deviation, respectively, using the DWQRE, and $\check{b}_{\tau}$ is the $\tau$th sample quantile of the corresponding residuals. 
To evaluate the forecasting performance of the three estimation methods, we calculate the empirical coverage rate (ECR), and perform VaR backtests for the forecasts at $\tau=5\%$, 10\%, 90\%, and 95\%. 
Specifically, the ECR is calculated as the proportion of observations that fall below the corresponding conditional quantile forecast for the last $176$ data points. 
We use three VaR backtests, namely, the likelihood ratio tests for correct unconditional coverage (UC) in \cite{Kupiec1995}, correct conditional coverage (CC) in \cite{christoffersen1998evaluating}, and dynamic quantile (DQ) test in \cite{Engle2004}. 
Let $H_t=I(y_t<Q_{y_t}(\tau\mid \mathcal{F}_{t-1}))$ be the hit, where $I(\cdot)$ is an indicator function.
The UC test examines the accuracy of the VaR forecasts using the null hypothesis that $E(H_t)=\tau$. 
The null hypothesis of the CC test is that, conditional on $\mathcal{F}_{t-1}$, $\{H_t\}$ are $i.i.d.$ Bernoulli random variables with success probability $\tau$. For the DQ test, we regress $H_t$ on regressors including a constant, three lagged hits $H_{t-i},$ for $i=1,2,3$, and the VaR forecast at the time point $t$. The null hypothesis of the DQ test is that the intercept is equal to $\tau$ and all regression coefficients are zero. 

Table \ref{forcast} reports the ECRs and $p$-values of the three VaR backtests for one-week-ahead forecasts obtained using the three estimation methods at four quantile levels. 
The proposed E-QMLE performs well at all four quantile levels, with $p$-values not less than $0.1$ for the backtests. In addition, the proposed G-QMLE performs well, except for $\tau=95\%$, with the $p$-value of the DQ test slightly smaller than $0.1$, whereas the DWQRE performs less satisfactorily at three quantile levels in terms of the DQ tests.  
For the ECRs, those of the E-QMLE are closest to the nominal quantile level, except for  $\tau=90\%$. 
Overall, the linear DAR model fitted using the proposed E-QMLE method outperforms that of the DWQRE in forecasting VaRs. 
The G-QMLE method performs worse than the E-QMLE, probably because the G-QMLE is not suitable for the data.

In addition, to compare the forecasting performance of the linear DAR model with that of the DAR model \eqref{DAR} and the AR-GARCH model fitted using QMLEs, we apply the DAR model of order three and an AR($3$)-GARCH($1,1$) model to the data using the E-QMLE and G-QMLE. The results based on the same rolling forecasting procedure are reported in Table \ref{forcast}. Here, both DAR models are comparable in the VaR backtests, and the linear DAR model fitted using the E-QMLE outperforms the DAR model in the ECRs for $\tau=10\%, 90\%$, and 95\%. This demonstrates the forecasting superiority of the linear DAR model for heavy-tailed data, which possibly benefits from its linear structure on the conditional standard deviation instead of on the conditional variance for the DAR model.
Furthermore, the linear DAR model is competitive with the AR-GARCH model in both the ECRs and the backtests. 

In summary, the proposed E-QMLE procedure of a linear DAR model seems to be more reasonable and suitable for the considered BTC data set in terms of fitting and forecasting. 

\section{Conclusion}\label{section6-conclusion}

We have proposed two QMLEs, namely, the E-QMLE and the G-QMLE, for linear DAR models, which are simpler to compute than the DWQRE of \cite*{Zhu2018_LDAR}. 
Under only a finite fractional moment of the process $\{y_t\}$, we establish the consistency and asymptotic normality for both QMLEs. 
We compare the E-QMLE with the G-QMLE and DWQRE in terms of asymptotic efficiency, and provide practical suggestions on choosing a suitable estimator. 
Moreover, we propose two BICs for order selection and two mixed portmanteau tests to check the adequacy of the fitted models based on the two QMLEs, and obtain their asymptotic properties without any moment conditions on the observed process. 
A real-data example confirms the usefulness and superiority of the proposed robust inference tools in terms of data fitting and forecasting. 

The robust inference tools presented here can be extended in two directions. 
First, there is a practical need to consider a linear DAR model with different orders for the conditional location and scale components. The proposed robust inference tools can be adapted to such an extension using a self-weighting approach. 
Second, the linear DAR model can be generalized to a vector form to jointly model multivariate time series. In the framework of a vector linear DAR model, it would be interesting to investigate whether the good properties of robust inference can be preserved. 
We leave these topics for future research. 

 

\renewcommand{\thesection}{A}
\setcounter{equation}{0} 
\section*{Appendix}

This appendix includes additional results for the simulation and the empirical analysis, as well as technical details for Theorems \ref{thm1EQMLE-consistency}--\ref{thmACF} and Remarks \ref{remark1}--\ref{remark-PT-of-GQMLE}.  
To show Theorems \ref{thm1EQMLE-consistency}--\ref{thm1G-QMLE-asymptotics}, Theorem \ref{thmACF} and Remark \ref{remark-PT-of-GQMLE}, Lemmas \ref{lemma1}--\ref{lemma-PT} are also introduced. 
Throughout the appendix, the notation $C$ is a generic constant which may take different values from lines to lines. The norm of a matrix or column vector is defined as $\|A\|=\sqrt{\text{tr}(AA^\prime)}=\sqrt{\sum_{i,j}|a_{ij}|^2}$. 

\subsection{Additional simulation results}

This section reports additional simulation results which are useful but not reported in the manuscript. Section \ref{sim1-tail behavior} evaluates the finite-sample performance of all inference tools for additional innovation distributions. Section \ref{sim2-M-suggestion} explores the size and power of two portmanteau tests $Q(M)$ and $Q^G(M)$ with respect to the lag order $M$ and the sample size $n$ by simulation. Finally, the simulation results for $Q(M)$ and $Q^G(M)$ with $M=12$, or 18, are presented in Section \ref{sim3-M-12-18}. 

\subsubsection{Simulation results for additional innovation distributions} \label{sim1-tail behavior}

In the first, third and fourth experiments in Sections \ref{sim-QMLE}, \ref{sim-BIC} and \ref{sim-Portmanteau tests}, owing to space limitations, we only reported the results with $\{\eta_t\}$ being $i.i.d.$ normal, Laplace, or Student’s $t_3$ distributed random variables.
This subsection provides additional results with $\{\eta_t\}$ being $i.i.d.$ Student’s $t_2$ or $t_5$  distributed random variables, or skewed Student’s $t_3$ distributed random variables with the skew parameter being -1.5 (denoted by $st_{3,-1.5}$). 

Table \ref{table-QMLE-supp} presents the biases, ESDs, and ASDs of $\widehat{\bm \theta}_n$ and $\widetilde{\bm \theta}_n$, when the innovations follow the $t_2$, $t_5$, or $st_{3,-1.5}$ distribution. 
The following findings in Section \ref{sim-QMLE} remain unchanged: 
(1) for the E-QMLE $\widehat{\bm \theta}_n$, the ESDs and ASDs of scale-type estimators $\widehat{\omega}_n$ and $\widehat{\beta}_n$ increase as the distribution of $\eta_t$ becomes more heavy tailed, while that of location-type estimator $\widehat{\alpha}_n$ are decreasing;  
(2) for the G-QMLE $\widetilde{\bm\theta}_n$, the ESDs and ASDs increase as the distribution of $\eta_t$ gets more heavy tailed. If $\eta_t\sim t_2$, then $E(\eta_t^2)=\infty$ such that the G-QMLE is not applicable, which results in the inferior performance of the G-QMLE in this case.
Comparing Table \ref{table-QMLE-supp} with Table 1 in the manuscript, we further have the following findings: 
(1) the E-QMLE performs satisfactorily with small biases and ESDs close to ASDs even for $t_2$ distribution with $E(\eta_t^2)=\infty$, indicating that the E-QMLE is more robust for thick tails than G-QMLE; 
(2) nonzero skewness will increase the bias of the E-QMLE, but it cannot affect its ESD and ASD. While the G-QMLE still has small biases but large ESDs and ASDs in fitting skewed data.

Table \ref{tableBIC-supp} provides the percentages of underfitting, correct selection, and overfitting cases by the BIC and BIC$^G$ when the innovations follow the $t_2$, $t_5$, or $st_{3,-1.5}$ distribution. 
It can be seen that the conclusions in Section \ref{sim-BIC} of the manuscript are unchanged. 
Moreover, we find that nonzero skewness could improve the selection accuracy of BIC but worsen that of BIC$^G$. This is probably because nonzero skewness almost has no effects on ESDs and ASDs of the E-QMLE $\widehat{\bm\theta}_n$, but it increases ESDs and ASDs of the G-QMLE $\widetilde{\bm\theta}_n$. 

Table \ref{tablechecking-supp} summarizes the rejection rates of $Q(6)$ and $Q^G(6)$ at the $5\%$ significance level when the innovations follow the $t_2$, $t_5$, or $st_{3,-1.5}$ distribution. The following findings in Section \ref{sim-Portmanteau tests} remain unchanged: 
(1) all sizes are close to the nominal rate except for the $t_2$ case, and the power increases as the sample size $n$ or the departure level increases. This is expected, because the E-QMLE and G-QMLE are not applicable if $\eta_t$ follows $t_2$ distribution such that $E(\eta_t^2)=\infty$, making the size inaccurate;
(2) for the same level of departures, $Q(6)$ and $Q^G(6)$ are more powerful in terms of detecting the misspecification in the conditional mean ($c_1\neq 0, c_2=0$) than they are in doing so in the conditional standard deviation ($c_1=0$, $c_2>0$);
(3) as the innovation distribution becomes more heavy tailed, $Q(6)$ and $Q^G(6)$ perform worse in detecting misspecifications in the conditional standard deviation.

\subsubsection{Portmanteau tests with respect to $n$ and $M$}\label{sim2-M-suggestion}

In this subsection, we conduct simulation studies to examine the size and power of portmanteau tests $Q(M)$ and $Q^G(M)$ with respect to the sample size $n$ and the lag order $M$.

To evaluate the size with respect to $n$ and $M$, we generate 1000 replications of sample size $n=100,200,\ldots,2000$, from the following model:
$$
y_t=0.1y_{t-1}+\eta_{t}(1+0.2|y_{t-1}|),
$$
where $\{\eta_t\}$ are $i.i.d.$ normal or Laplace distributed random variables. 
Here $\{\eta_t\}$ are standardized with median zero and $E(|\eta_t|)=1$ to evaluate $Q(M)$ using the E-QMLE, and $\{\eta_t\}$ are standardized with mean zero and $\text{var}(\eta_t)=1$ to evaluate $Q^G(M)$ using the G-QMLE.
We fit a linear DAR model with $p = 1$ using the Gaussian or exponential quasi-maximum likelihood estimation. 
Figure \ref{fig_size} plots the empirical size of $Q(M)$ and $Q^G(M)$ at the $5\%$ significance level, for different settings of $M=2, 4,\ldots,50$ and $n=100,200,\ldots,2000$. We have the following findings: 
(1) the size with $M> 10$ is decreasing with respect to the lag $M$, and most of the sizes with $M>20$ are less than 0.05; 
(2) the size is insensitive to the sample size $n$ and the distribution of $\eta_t$. 
Therefore, we suggest $M\leq 20$ to avoid undersize for two portmanteau tests, which is consistent with the guideline given by \cite{box2015_TSA}.

To evaluate the power with respect to $n$ and $M$, we generate 1000 replications of sample size $n=100,200,\ldots,2000$, from the following model:
$$
y_t=0.1y_{t-1}+0.1y_{t-2}+\eta_{t}(1+0.2|y_{t-1}|+0.2|y_{t-2}|),
$$
where all other settings are the same as in the above experiment for size evaluation. Figure \ref{fig_power} plots the empirical power with different settings of $n$ and $M$. We have the following findings:  
(1) the power is linearly increasing with respect to the sample size $n$, i.e. $\text{power}=O(n)$;
(2) the $\ln(\text{power})$ is linearly decreasing with respect to the lag $M$, i.e. $\ln(\text{power})=O(M)$, or $\text{power}=O(e^M)$.
As a result, to avoid power loss for tests with a large lag $M$, sample size $n=O(e^M)$ may be used, or the lag order $M=O(\ln(n))$ may be suggested for a fixed sample size, which is consistent with the suggestion by \cite{tsay2005AoFTS}. 
In practice, we can choose several values of $M$, such as $M_1=\lfloor{\ln(n)}\rfloor$, $M_2=2\lfloor{\ln(n)}\rfloor,\ldots$, where $\lfloor{x}\rfloor$ denotes the largest integer not greater than $x$.

In summary, based on the above simulation studies, we may suggest that the rule-of-thumb for multiple choices of $M$ could be
\begin{equation}\label{selection_M}
	(M_1,M_2,\ldots,M_{K})=(\lfloor{\ln(n)}\rfloor,2\lfloor{\ln(n)}\rfloor,\ldots,K\lfloor{\ln(n)}\rfloor ),
\end{equation} 
where $K$ is the maximum value of $k$, such that $M_{k}=k\lfloor{\ln(n)}\rfloor\leq 20$.

\subsubsection{Portmanteau tests with $M=12$ or 18} \label{sim3-M-12-18}

In the fourth experiment in Section \ref{sim-Portmanteau tests}, owing to space limitations, we only reported the results of $Q(M)$ and $Q^G(M)$ with $M=6$. This subsection provides the results for $M=12$, and 18. 
Table \ref{tablechecking-M} reports the rejection rates of $Q(M)$ and $Q^G(M)$ at the $5\%$ significance level for $M=12$ or 18.
The following findings in Section \ref{sim-Portmanteau tests} in the manuscript remain unchanged: 
(1) the power increases as the sample size $n$ or the departure level increases; 
(2) for the same level of departures, $Q(M)$ and $Q^G(M)$ for $M=12$ or 18 are more powerful in terms of detecting the misspecification in the conditional mean ($c_1\neq 0, c_2=0$) than they are in doing so in the conditional standard deviation ($c_1=0$, $c_2>0$); 
(3) as the innovation distribution becomes more heavy tailed, $Q(M)$ and $Q^G(M)$ perform worse in detecting misspecifications in the conditional standard deviation. 
In addition, most of the sizes of $Q(M)$ and $Q^G(M)$ with $M=12$ or 18 are less than the nominal level, and the power of $Q(M)$ and $Q^G(M)$ decreases slowly as $M$ increases.

\subsection{Proof of Theorem \ref{thm1EQMLE-consistency}}

To show Theorem \ref{thm1EQMLE-consistency}, we introduce the following lemma.
\begin{lemma}\label{lemma1}
	For any $\bm \theta^{*} \in \Theta$, let $B_\eta  (\bm\theta^{*})=\{\bm\theta \in \Theta:  \|\bm\theta - \bm\theta^{*}\| < \eta \}$ be an open neighborhood of $\bm\theta^{*}$ with radius $\eta>0$. If Assumptions \ref{assum1-compactness}, \ref{assum2-stationarity} and \ref{assum3EQMLE-moment and density}(i) hold, then 
	\begin{equation*}
		\begin{aligned}
			&\text{ (i) }E \sup_{ \bm\theta \in \Theta}\left|\ell_{t}^E( \bm\theta)\right|<\infty;\\
			&\text{ (ii) } E\left[\ell_t^E(\bm\theta)\right] \text{has a unique minimum at } \bm\theta_0;\\
			&\text{ (iii) } E \left[ \underset{\bm\theta \in B_{\eta} (\bm\theta^{*}) }{\sup}|\ell_{t}^E(\bm\theta)-\ell_{t}^E(\bm\theta^{*})|\right] \to 0  \text{ as } \eta \to 0.
		\end{aligned}
	\end{equation*}		
\end{lemma}
\begin{proof}
	Recall that $\ell_t^E(\bm\theta) = \ln h_t(\bm\delta)+h_t^{-1}(\bm\delta)|\varepsilon_t(\bm\alpha)|$, where $\varepsilon_t(\bm\alpha) = y_t-\sum_{i=1}^{p}\alpha_{i}y_{t-i}$ and $h_t(\bm\delta) = \omega  +\sum_{i=1}^{p}\beta_{i}|y_{t-i}|$. 
	We first show (i). 
	Denote $c_i=\sup_{\bm \theta \in \Theta}|\alpha_{i0}-\alpha_{i}|$, then by Assumption \ref{assum1-compactness} we have $c_i<\infty$ for $i=1,\ldots,p$. 
	Note that $\varepsilon_t(\bm\alpha)=\sum_{i=1}^p(\alpha_{i0}-\alpha_i)y_{t-i}+\varepsilon_t(\bm\alpha_0)$ and $\varepsilon_t(\bm\alpha_0)=\eta_th_t(\bm\delta_0)=\eta_t\left(\omega_0+\sum_{i=1}^p\beta_{i0}\left|y_{t-i}\right|\right)$. 
	Then by Assumption \ref{assum1-compactness} and $E(|\eta_t|)=1$ by Assumption \ref{assum3EQMLE-moment and density}(i), it follows that
	\begin{align}\label{A.2}
		E&\sup_{\bm \theta \in \Theta}\left[\frac{|\varepsilon_t(\bm\alpha) |}{h_t(\bm\delta)}\right] \nonumber\\
		&\leq E\sup _{\bm \theta \in \Theta}\left[\frac{\sum_{i=1}^{p}|\alpha_{i0}-\alpha_{i}| |y_{t-i}|}{\omega+\sum_{i=1}^p \beta_i\left|y_{t-i}\right|}\right]
		+ E(|\eta_t|)E\sup _{\bm \theta \in \Theta}\left[\dfrac{\omega_0+\sum_{i=1}^{p}\beta_{i0} |y_{t-i}|}{\omega+\sum_{i=1}^{p}\beta_i |y_{t-i}|}\right] \nonumber\\
		&\leq \sum_{i=1}^{p}E\left[\frac{c_i |y_{t-i}|}{\underline{\beta}\left|y_{t-i}\right|}\right]
		+ E \left[\dfrac{\overline{\omega}}{\underline{\omega}+\underline{\beta} \sum_{i=1}^{p}|y_{t-i}|}\right]
		+\sum_{i=1}^{p}E \left[\dfrac{\beta_{i0}|y_{t-i}|}{\underline{\beta}|y_{t-i}|}\right]\nonumber\\
		&\leq \dfrac{1}{\underline{\beta}}\sum_{i=1}^{p}c_i+\dfrac{\overline{\omega}}{\underline{\omega}}+\frac{1}{\underline{\beta}}\sum_{i=1}^{p}\beta_{i0}
		< \infty.
	\end{align}
	By \eqref{finite 1} and \eqref{A.2}, we have
	\[E\sup_{\bm\theta \in \Theta}\left|\ell_{t}^E(\bm\theta)\right|\leq E \sup_{\bm\theta \in \Theta}\left|\ln h_t(\bm\delta)\right| + E \sup_{\bm \theta \in \Theta}\left[\frac{|\varepsilon_t(\bm\alpha)|}{h_t(\bm\delta)}\right] <\infty.\]
	Thus, (i) holds. 
	
	We next show (ii). Note that $\varepsilon_t(\bm\alpha)=\varepsilon_t(\bm\alpha_0)-(\bm\alpha-\bm\alpha_0)^\prime \bm Y_{1t}h_t(\bm \delta_0)$ and $\varepsilon_t(\bm\alpha_0)=\eta_t h_t(\bm \delta_0)$, where $\bm Y_{1t}=h_t^{-1}(\bm \delta_0)(y_{t-1},\ldots,y_{t-p})^\prime$. 
	Denote $\mathcal{F}_t$ as the $\sigma$-field generated by $\{y_s,s\leq t\}$. Recall that $\eta_t$ has zero median and $E(|\eta_t|)=1$, and thus $\min \limits_{a} E|\eta_t-a|=E|\eta_t-\text{median}(\eta_t)|=E|\eta_t|=1$. Then by the law of iterated expectations, we have
	\begin{align*}
		E\left[\ell_t^E(\bm\theta)\right]
		=&E\left[\ln h_t(\bm\delta)+\frac{|\varepsilon_t(\bm\alpha_0)-(\bm\alpha-\bm\alpha_0)^\prime \bm Y_{1t}h_t(\bm \delta_0)|}{h_t(\bm \delta)}\right] \nonumber\\
		= &E\left[\ln h_t(\bm \delta)+\frac{h_t(\bm \delta_0)}{h_t(\bm \delta)}E\left\{ \left|\eta_t-(\bm\alpha-\bm\alpha_0)^\prime\bm Y_{1t}\right|\Big|\mathcal{F}_{t-1}\right\}\right] \nonumber\\	
		\geq &E\left[\ln h_t(\bm \delta)+\frac{h_t(\bm \delta_0)}{h_t(\bm \delta)}E\left(|\eta_t|\big|\mathcal{F}_{t-1}\right)\right] \nonumber\\
		=&E\left[\ln h_t(\bm \delta)+\frac{h_t(\bm \delta_0)}{h_t(\bm \delta)}\right],			
	\end{align*}
	where the minimum of the inequality is attained if and only if $\bm\alpha=\bm\alpha_0$ almost surely as $\eta_t$ has zero median. 
	Moreover, the function $f(x)=\ln x + a/x$, $a \geq 0$, reaches its minimum at $x=a$. Therefore, $E[\ell_t^E(\bm \theta)]$ reaches its minimum if and only if $h_t(\bm \delta)=h_t(\bm \delta_0)$ almost surely, and hence $\bm \theta=\bm \theta_0$. As a result, $E[\ell_t^E(\bm \theta)]$ is uniformly minimized at $\bm \theta_0$, that is (ii) holds. 
	
	Finally, we show (iii). Let $\bm \theta^*=(\bm\alpha^{*\prime}$, $\bm \delta^{*\prime})^\prime \in \Theta$. For any $\bm \theta \in B_{\eta}(\bm \theta^*)$, we can see that 
	\begin{align*}
		\ell_t^E(\bm \theta)-\ell_t^E(\bm \theta^*)=\ln h_t(\bm \delta)-\ln h_t(\bm \delta^*)+\left[\frac{|\varepsilon_t(\bm\alpha)|}{h_t(\bm \delta)}-\frac{|\varepsilon_t(\bm\alpha^*)|}{h_t(\bm \delta)}\right]+\left[\frac{|\varepsilon_t(\bm\alpha^*)|}{h_t(\bm \delta)}-\frac{|\varepsilon_t(\bm\alpha^*)|}{h_t(\bm \delta^*)}\right].
	\end{align*}
	By Taylor's expansion, we can see that
	\begin{align*}
		\ln h_t(\bm \delta)-\ln h_t(\bm \delta^*)=\frac{(\bm \delta-\bm \delta^*)^\prime}{h_t(\bar{\bm\delta})} (1,|y_{t-1}|,...,|y_{t-p}|)^\prime,
	\end{align*}
	where $\bar{\bm\delta}$ lies between $\bm \delta$ and $\bm \delta^*$. Then, by Assumption \ref{assum1-compactness}, we have 
	\begin{align*}
		&E\left[\sup_{\bm \theta \in B_{\eta}(\bm \theta^*)}\left|\ln h_t(\bm \delta)-\ln h_t(\bm \delta^*)\right|\right]\nonumber\\	
		\leq &\eta E\left[\frac{1+\sum_{i=1}^p|y_{t-i}|}{\underline\omega+\underline \beta \sum_{i=1}^p|y_{t-i}|}\right]\leq \eta\left(\frac{1}{\underline{\omega}}+\frac{p}{\underline{\beta}}\right)\rightarrow 0
		\text{ as } \eta \rightarrow 0. 
	\end{align*}
	Similarly, by Taylor's expansion and \eqref{A.2}, it can be verified that
	\begin{align*}
		&E\left[\sup_{\bm \theta \in B_{\eta}(\bm \theta^*)}\left|\varepsilon_t(\bm\alpha^*)\right|\left|\frac{1}{h_t(\bm \delta)}-\frac{1}{h_t(\bm \delta^*)}\right|\right]\nonumber\\ 
		=& E \left[\sup_{\bm \theta \in B_{\eta}(\bm \theta^*)}|\varepsilon_t(\bm\alpha^*)|\frac{\left|(\bm \delta-\bm \delta^{*})^\prime(1,|y_{t-1}|,...,|y_{t-p}|)^\prime\right|}{h_{t}^{2}( \bar{\bm\delta})}\right]\nonumber\\
		\leq & \eta E\left[\sup_{\bm \theta \in B_{\eta}(\bm \theta^*)}\frac{|\varepsilon_t(\bm\alpha^*)|}{h_{t}( \bar{\bm\delta})}\left(\frac{1+\sum_{i=1}^p|y_{t-i}|}{\underline\omega+\underline \beta \sum_{i=1}^p|y_{t-i}|}\right)\right]\nonumber\\
		\leq & \eta\left(\dfrac{1}{\underline{\beta}}\sum_{i=1}^{p}c_i+\dfrac{\overline{\omega}}{\underline{\omega}}+\frac{1}{\underline{\beta}}\sum_{i=1}^{p}\beta_{i0}\right)\left(\frac{1}{\underline{\omega}}+\frac{p}{\underline{\beta}}\right)
		\rightarrow 0 \text{ as } \eta \rightarrow 0,	
	\end{align*}
	and 
	\begin{align*}
		E\left[\underset{\bm\theta \in B_{\eta} (\bm\theta^{*}) }{\sup}\frac{\left||\varepsilon_t(\bm\alpha)|-|\varepsilon_t(\bm\alpha^*)|\right|}{h_t(\bm \delta)}\right]
		\leq & \eta E\left[\frac{\sum_{i=1}^p|y_{t-i}|}{\underline\omega+\underline \beta \sum_{i=1}^p|y_{t-i}|} \right]
		\leq \frac{\eta p}{\underline{\beta}}\rightarrow 0 \text{ as } \eta\rightarrow 0.
	\end{align*}
	Therefore, we have 
	\[E \left[ \underset{\bm\theta \in B_{\eta} (\bm\theta^{*}) }{\sup}|\ell_{t}^E(\bm\theta)-\ell_{t}^E(\bm\theta^{*})|\right] \to 0  \text{ as } \eta \to 0.\] 
	Hence, (iii) is verified. The proof of this lemma is accomplished. 
\end{proof}
\begin{proof}[Proof of Theorem \ref{thm1EQMLE-consistency}]
	To establish the strong consistency of $\widehat{\bm\theta}_n$, we follow the method in \cite{Huber1967}. Let $V$ be any open neighborhood of $\bm \theta_0\in \Theta$. By Lemma \ref{lemma1} (iii), for any $\bm\theta^*\in V^c = \Theta/V$ and $\varepsilon > 0$, there exists an $\eta_0>0$ such that 
	\begin{align}\label{a3}
		E\left[\inf_{\bm\theta\in B_{\eta_0}(\bm\theta^*)}\ell_{t}^E(\bm\theta)\right]\geq E\left[\ell_{t}^E(\bm\theta^*)\right]-\varepsilon.
	\end{align}
	From Lemma \ref{lemma1} (i), by the Birkhoff ergodic theorem \citep{Fristedt2013}, if $n$ is large enough, then we have
	\begin{align}\label{a4}
		\frac{1}{n-p}\sum_{t=p+1}^{n}\inf_{\bm\theta \in B_{\eta_0}(\bm\theta^*)}\ell_{t}^E(\bm\theta)\geq E\left[ \inf_{\bm\theta \in B_{\eta_0}(\bm\theta^*)}\ell_{t}^E(\bm\theta)\right]-\varepsilon \quad \text{a.s.}.
	\end{align}
	Since $V^c$ is compact, we can choose $\{B_{\eta_0}(\bm\theta_i):\bm\theta_i \in V^c, i=1,2,...,k\}$ to be a finite covering of $V^c$. For each $\bm\theta_i\in V^c$, from Lemma \ref{lemma1} (ii), there exists an $\varepsilon_0 > 0$ such that 
	\begin{align}\label{a5}
		E\left[\inf_{\bm\theta \in B_{\eta_0}(\bm\theta_i)}\ell_{t}^E(\bm\theta)\right] \geq E\left[\ell_{t}^E(\bm\theta_0)\right]+3\varepsilon_0.
	\end{align}
	Thus by \eqref{a3}--\eqref{a5}, taking $\varepsilon = \varepsilon_0$, if $n$ is large enough, we have
	\begin{align}\label{a6}
		\inf_{\bm\theta \in V^c}L_n^E(\bm\theta) = &\min_{1\leq i \leq k}\inf_{\bm\theta \in B_{\eta_0}(\bm\theta_i)}L_n^E(\bm\theta)\nonumber\\
		\geq &\min_{1\leq i \leq k}\frac{1}{n-p}\sum_{t=p+1}^{n}\inf_{\bm\theta \in B_{\eta_0}(\bm\theta_i)}\ell_{t}^E(\bm\theta)\nonumber\\
		\geq &\min_{1\leq i \leq k}E\left[ \inf_{\bm\theta \in B_{\eta_0}(\bm\theta_i)}\ell_{t}^E(\bm\theta)\right]-\varepsilon\nonumber\\
		\geq &E\left[\ell_{t}^E(\bm\theta_0)\right]+2\varepsilon_0.
	\end{align}
	On the other hand, by the Birkhoff ergodic theorem \citep{Fristedt2013}, it follows that
	\begin{align}\label{a7}
		\inf_{\bm\theta\in V}L_n^E(\bm\theta)\leq L_n^E(\bm\theta_0)=\frac{1}{n-p}\sum_{t=p+1}^n \ell_{t}^E(\bm\theta_0)\leq E\left[\ell_{t}^E(\bm\theta_0)\right]+\varepsilon_0  \quad \text{a.s.}.
	\end{align}
	Hence, combining \eqref{a6} and \eqref{a7}, it holds that
	\[\inf_{\bm\theta \in V^c}L_n^E(\bm\theta)\geq E\left[\ell_{t}^E(\bm\theta_0)\right]+2\varepsilon_0 > E\left[\ell_{t}^E(\bm\theta_0)\right]+\varepsilon_0\geq \inf_{\bm\theta\in V}L_n^E(\bm\theta) \quad \text{a.s.}.\]
	Since $\widehat{\bm\theta}_n=\arg\min_{\Theta}L_n^E(\bm\theta)$, if $n$ is large enough, then $\widehat{\bm\theta}_n\in V$ almost surely for  $\forall V$. By the arbitrariness of $V$, it follows that $\widehat{\bm\theta}_n\rightarrow \bm\theta_0$ almost surely. The proof is accomplished.
\end{proof}

\subsection{Proof of Theorem \ref{thm2EQMLE-normality}}
To show Theorem \ref{thm2EQMLE-normality}, we introduce Lemmas \ref{lemma2}--\ref{lemma3} below, where Lemma \ref{lemma2} verifies the stochastic differentiability condition defined by \cite{pollard1985new}, and Lemma \ref{lemma3} provides a representation to obtain the $\sqrt{n}$-consistency and asymptotic normality of $\widehat{\bm\theta}_{n}$. 

\begin{lemma}\label{lemma2}
	If Assumptions \ref{assum1-compactness}, \ref{assum2-stationarity} and \ref{assum3EQMLE-moment and density} hold, then for any sequence of random variables $\bm u$ such that $\bm u=o_p(1)$, it follows that
	\[\Pi_{1n}(\bm u)=o_p(\sqrt{n}\|\bm u\|+n\|\bm u\|^{2}),\]
where $\Pi_{1n}(\bm u)=2\sum_{t=p+1}^n\int_0^{\bm u^\prime \bm K_{1t}}\left\{X_t(s)-E[X_t(s)|\mathcal{F}_{t-1}]\right\}ds$ with $\bm K_{1t} = (\bm Y_{1t}^{'},\bm 0_{p+1}^{\prime})^\prime$ and $X_t(s)=I(\eta_t\leq s)-I(\eta_t\leq 0)$. 
\end{lemma}
\begin{proof}
    Note that $\partial\varepsilon_t(\bm\alpha_0)/\partial\bm \theta=(-\bm Y_{1t}^{\prime}h_t, \bm 0_{p+1}^{\prime})^{\prime}=-h_t\bm K_{1t}$. Then it follows that
    \[\int_0^{\bm u^{'}\bm K_{1t}}X_t(s)ds=\bm u^{'}\bm K_{1t}\int_0^{1}X_t(\bm u^{'}\bm K_{1t}s)ds=-\frac{\bm u^{'}}{h_t}\frac{\partial\varepsilon_t(\bm\alpha_0)}{\partial\bm \theta}M_t(\bm u),\]
    where $M_t(\bm u)=\int_0^{1}X_t(\bm u^{'}\bm K_{1t}s) ds$. 
    It can be verified that
    \begin{align*}
    	|\Pi_{1n}(\bm u)|
    	\leq2\|\bm u\|\sum_{j=1}^{2p+1}\left|\dfrac{1}{h_t}\dfrac{\partial\varepsilon_t(\bm\alpha_0)}{\partial\theta_j}\sum_{t=p+1}^{n}\{M_t(\bm u)-E[M_t(\bm u)|\mathcal{F}_{t-1}]\}\right|,
    \end{align*}
    where $\theta_j$ is the $j$th element of $\bm \theta$ for $1\leq j\leq 2p+1$.   
    Let $m_{t}=h_t^{-1}\partial\varepsilon_t(\bm\alpha_0)/\partial\theta_j$ and $f_t(\bm u)=m_{t}M_t(\bm u)$. Define
    \[D_n(\bm u)=\frac{1}{\sqrt{n}}\sum_{t=p+1}^{n}\{f_t(\bm u)-E[f_t(\bm u)|\mathcal{F}_{t-1}]\}.\]
    Then, to prove Lemma \ref{lemma2}, it suffices to show that, for any $\eta>0$,
    \begin{align}\label{a8}
    \sup_{\|\bm u\|\leq\eta}\frac{|D_n(\bm u)|}{1+\sqrt{n}\|\bm u\|}=o_p(1).
    \end{align}
    Note that $m_t=\max\{m_t,0\}-\max\{-m_t,0\}$. We first prove the case when $m_t\geq0$.

We adopt the method in Lemma 4 of \cite{pollard1985new} to verify \eqref{a8}. Let $\mathfrak{F}=\{f_t(\bm u):\|\bm u\|\leq\eta\}$ be a collection of functions indexed by $\bm u$. We first verify that $\mathfrak{F}$ statisfies the bracketing condition on page 304 of \cite{pollard1985new}. Denote $B_r(\bm \zeta)$ as an open neighborhood of $\bm \zeta$ with radius $r>0$, and define a constant $C_1$ to be selected later. For any fixed $\varepsilon>0$ and $0<\delta\leq\eta$, there exists a sequence of small cubes $\{B_{\varepsilon\delta/C_1}(\bm u_i)\}_{i=1}^{K_\varepsilon}$ to cover $B_{\delta}(\bm 0)$, where $K_{\varepsilon}$ is an integer less than $C_0\varepsilon^{-(2p+1)}$ and $C_0$ is not depending on $\varepsilon$ and $\delta$. 
Denote $V_i(\delta)=B_{\epsilon \delta/C_0}(\bm u_{i})\bigcap B_{\delta}(\bm0)$, and let $U_1(\delta)=V_1(\delta)$ and $U_i(\delta)=V_i(\delta)-\bigcup_{j=1}^{i-1}V_j(\delta)$ for $i\geq 2$. Note that $\{U_i(\delta)\}_{i=1}^{K_\varepsilon}$ is a partition of $B_{\delta}(\bm0)$. 
For each $\bm u_i\in U_i(\delta)$ with $1\leq i\leq K_\varepsilon$, we define the bracketing functions as follows
\begin{align*}
	f_t^{\pm}(\bm u)=m_t\int_0^1X_t\left(\bm u^{'}\bm K_{1t}s\pm\frac{\varepsilon\delta}{C_{1}h_{t}}\left\|\frac{\partial\varepsilon_t(\bm\alpha_0)}{\partial\bm \theta}\right\|\right)ds.
\end{align*}
Since the indicator function is nondecreasing and $m_t\geq0$, for any $\bm u\in U_i(\delta)$, we have
\[f_t^{-}(\bm u_i)\leq f_t(\bm u)\leq f_t^{+}(\bm u_i).\]
Moreover, by Taylor's expansion and the fact that $\partial\varepsilon_t(\bm\alpha_0)/\partial\bm \theta=(-\bm Y_{1t}^{\prime}h_t, \bm 0_{p+1}^{\prime})^{\prime}$, it holds that
\begin{align}\label{a9}
	E\left[f_t^{+}(\bm u_i)-f_t^{-}(\bm u_i)|\mathcal{F}_{t-1}\right]
	&\leq\dfrac{2\varepsilon\delta}{C_{1}h^{2}_{t}}\left\|\frac{\partial\varepsilon_t(\bm\alpha_0)}{\partial\bm \theta}\right\|^{2}\sup_{x}f(x)
	=\frac{\varepsilon \delta\Delta_t}{C_{1}},
\end{align}
where $\Delta_t = 2\left\|\bm Y_{1t}\right\|^{2}\sup_{x}f(x)$. 
Since $\sup_{x}f(x)<\infty$ by Assumption \ref{assum3EQMLE-moment and density}(ii), we choose $C_1=E(\Delta_t)$. Then by iterated-expectation, we have
\[E[f_t^{+}(\bm u_i)-f_t^{-}(\bm u_i)]=E\left\{E\left[f_t^{+}(\bm u_i)-f_t^{-}(\bm u_i)|\mathcal{F}_{t-1}\right]\right\}\leq\varepsilon\delta.\]
Thus, the family $\mathfrak{F}$ satisfies the bracketing condition.

Put $\delta_k=2^{-k}\eta$. Define $B(k)= B_{\delta_k}(\bm 0)$, and $A(k)$ to be the annulus $B(k)\setminus B(k+1)$. Fix $\varepsilon>0$, for each $1\leq i \leq K_{\varepsilon}$, by the bracketing condition, there exists a partition $\{U_i(\delta_k)\}_{i=1}^{K_\varepsilon}$ of $B(k)$.
For $\bm u\in U_i(\delta_k)$, by \eqref{a9} with $\delta=\delta_k$, we have the upper tail 
\begin{align*}
	D_n(\bm u)\leq &\frac{1}{\sqrt{n}}\sum_{t=p+1}^{n}\left\{f_t^{+}(\bm u_i)-E\left[f_t^{-}(\bm u_i)|\mathcal{F}_{t-1}\right]\right\}\nonumber\\
	=&D_n^{+}(\bm u_i)+\frac{1}{\sqrt{n}}\sum_{t=p+1}^{n}E\left[f_t^{+}(\bm u_i)-f_t^{-}(\bm u_i)|\mathcal{F}_{t-1}\right]\nonumber\\
	\leq &D_n^{+}(\bm u_i)+\sqrt{n}\varepsilon\delta_k\left(\frac{1}{nC_1}\sum_{t=p+1}^{n}\Delta_t\right),
\end{align*}
and the lower tail
\begin{align*}
D_n(\bm u)\geq &\frac{1}{\sqrt{n}}\sum_{t=p+1}^{n}\left\{f_t^{-}(\bm u_i)-E\left[f_t^{+}(\bm u_i)|\mathcal{F}_{t-1}\right]\right\}\nonumber\\
=&D_n^{-}(\bm u_i)-\frac{1}{\sqrt{n}}\sum_{t=p+1}^{n}E\left[f_t^{+}(\bm u_i)-f_t^{-}(\bm u_i)|\mathcal{F}_{t-1}\right]\nonumber\\
\geq &D_n^{-}(\bm u_i)-\sqrt{n}\varepsilon\delta_k\left(\frac{1}{nC_1}\sum_{t=p+1}^{n}\Delta_t\right),
\end{align*}
where 
\begin{align*}
D_n^{+}(\bm u_i)&=\frac{1}{\sqrt{n}}\sum_{t=p+1}^{n}\left\{f_t^{+}(\bm u_i)-E\left[f_t^{+}(\bm u_i)|\mathcal{F}_{t-1}\right]\right\} \quad\text{and}\\
D_n^{-}(\bm u_i)&=\frac{1}{\sqrt{n}}\sum_{t=p+1}^{n}\left\{f_t^{-}(\bm u_i)-E\left[f_t^{-}(\bm u_i)|\mathcal{F}_{t-1}\right]\right\}.
\end{align*}
Denote the event 
\[E_n=\left\{\omega:\frac{1}{nC_1}\sum_{t=p+1}^{n}\Delta_t(\omega)<2\right\}.\]
On $E_n$ with $\bm u\in U_i(\delta_k)$, it follows that
\begin{align}\label{a10}
  D_n^{-}(\bm u_i)-2\sqrt{n}\varepsilon\delta_k	\leq D_n(\bm u)\leq D_n^{+}(\bm u_i)+2\sqrt{n}\varepsilon\delta_k.
\end{align}
For $\bm u \in A(k)$, i.e. $\delta_{k+1} \leq \|\bm u\| \leq \delta_{k}$, we have $1+\sqrt{n}\|\bm u\|>\sqrt{n}\delta_{k+1}=\sqrt{n}\delta_k/2$. Thus, by the upper tail in \eqref{a10} and Chebyshev's inequality, it holds that
\begin{align}\label{a12}
	P\left(\sup_{\bm u\in A(k)}\frac{D_n(\bm u)}{1+\sqrt{n}\|\bm u\|}>6\varepsilon,E_n\right)\leq &P\left(\max_{1\leq i\leq K_{\varepsilon}}\sup_{\bm u\in U_i(k)\cap A(k)}D_n(\bm u)>3\sqrt{n}\varepsilon\delta_k,E_n\right)\nonumber\\
	\leq &K_{\varepsilon}\max_{1\leq i\leq K_{\varepsilon}}P(D_n^{+}(\bm u_i)>\sqrt{n}\varepsilon\delta_k)\nonumber\\
	\leq &K_{\varepsilon}\max_{1\leq i\leq K_{\varepsilon}}\frac{E[(D_n^{+}(\bm u_i))^2]}{n\varepsilon^{2}\delta_k^{2}}, 
\end{align}
and by the lower tail in \eqref{a10} and Chebyshev's inequality, we have
\begin{align}\label{a13}
	P\left(\inf_{\bm u\in A(k)}\frac{D_n(\bm u)}{1+\sqrt{n}\|\bm u\|}<-6\varepsilon,E_n\right)
	\leq &P\left(\min_{1\leq i\leq K_{\varepsilon}}\inf_{\bm u\in U_i(k)\cap A(k)}D_n(\bm u)<-3\sqrt{n}\varepsilon\delta_k,E_n\right)\nonumber\\
	\leq &K_{\varepsilon}\max_{1\leq i\leq K_{\varepsilon}}P(D_n^{-}(\bm u_i)<-\sqrt{n}\varepsilon\delta_k)\nonumber\\
	\leq &K_{\varepsilon}\max_{1\leq i\leq K_{\varepsilon}}\frac{E[(D_n^{-}(\bm u_i))^2]}{n\varepsilon^{2}\delta_k^{2}}.
\end{align}
Note that we have $\|\bm K_{1t}\|\leq p/\underline{\beta}$ by Assumption \ref{assum1-compactness}, and it follows that   
\begin{align*}
|\bm u^{'}\bm K_{1t}|\leq \|\bm u\|\|\bm K_{1t}\|\leq \frac{p\delta_{k}}{\underline{\beta}} \quad \text{and} \quad 
|m_{t}|\leq \frac{1}{h_t}\left\|\frac{\partial\varepsilon_t(\bm\alpha_0)}{\partial\bm\theta}\right\|=\|\bm K_{1t}\|\leq \frac{p}{\underline{\beta}}.
\end{align*}
By Taylor's expansion and iterated-expectation, we have
\begin{align*}
	E[f_t^{+}(\bm u_i)^2]=&E\left\{E[f_t^{+}(\bm u_i)^2|\mathcal{F}_{t-1}]\right\}\nonumber\\
	\leq &E\left\{m_{t}^{2}\int_0^1E\left[\left|X_t\left(\bm u^{'}\bm K_{1t}s+\frac{\varepsilon\delta_{k}}{C_{1}h_{t}}\left\|\frac{\partial\varepsilon_t(\bm\alpha_0)}{\partial\bm \theta}\right\|\right) \Big|\mathcal{F}_{t-1}\right|\right]ds\right\}\nonumber\\
	\leq &\frac{p^2}{\underline{\beta}^2}E\left(\sup_{|x|\leq C\delta_{k}}|F(x)-F(0)|\right)
	\leq C\delta_{k}\sup_{x}f(x).
\end{align*}
This together with the fact that $f_t^{+}(\bm u_i)-E\left[f_t^{+}(\bm u_i)|\mathcal{F}_{t-1}\right]$ is a martingale difference sequence,  and $\sup_{x}f(x)<\infty$ by Assumption \ref{assum3EQMLE-moment and density}(ii), implies that
\begin{align}\label{a14}
	E\left[(D_n^{+}(\bm u_i))^2\right]=&\frac{1}{n}\sum_{t=p+1}^nE\left\{f_t^{+}(\bm u_i)-E\left[f_t^{+}(\bm u_i)|\mathcal{F}_{t-1}\right]\right\}^2\nonumber\\
	\leq &\frac{1}{n}\sum_{t=p+1}^nE\left[(f_t^{+}(\bm u_i))^2\right]
	\leq C\delta_{k}\sup_{x}f(x):=\pi_n(\delta_k).
\end{align}
Similarly, we have
\begin{align}\label{a15}
	E\left[(D_n^{-}(\bm u_i))^2\right]\leq \pi_n(\delta_k).
\end{align}
Combining \eqref{a12} and \eqref{a14}, we have
\[P\left(\sup_{\bm u\in A(k)}\frac{D_n(\bm u)}{1+\sqrt{n}\|\bm u\|}>6\varepsilon,E_n\right)\leq \frac{K_{\varepsilon}\pi_n(\delta_k)}{n\varepsilon^2\delta_k^2}.\]
Combining \eqref{a13} and \eqref{a15}, we have
\[P\left(\inf_{\bm u\in A(k)}\frac{D_n(\bm u)}{1+\sqrt{n}\|\bm u\|}<-6\varepsilon,E_n\right)\leq \frac{K_{\varepsilon}\pi_n(\delta_k)}{n\varepsilon^2\delta_k^2}.\]
Thus, we can show that
\begin{align}\label{a16}
	P\left(\sup_{\bm u\in A(k)}\frac{|D_n(\bm u)|}{1+\sqrt{n}\|\bm u\|}>6\varepsilon,E_n\right)\leq \frac{2K_{\varepsilon}\pi_n(\delta_k)}{n\varepsilon^2\delta_k^2}.
\end{align}
Since $\pi_n(\delta_k)\to 0$ as $k\to\infty$, we can choose $k_{\varepsilon}$ so that
	$2K_{\varepsilon}\pi_n(\delta_k)/(\varepsilon^2\eta^2)<\varepsilon$
for $k\geq k_{\varepsilon}$. Let $k_n$ be an integer so that $n^{-1/2}\leq 2^{-k_n}\leq 2n^{-1/2}$. Split $\{\bm u:\|\bm u\|\leq \eta\}$ into two sets $B:=B(k_n+1)$ and $B^{c}:=B(0)-B(k_n+1)=\cup_{k=0}^{k_n}A(k)$. Then by \eqref{a16}, we have
\begin{align}\label{a17}
	P\left(\sup_{\bm u\in B^{c} }\frac{|D_n(\bm u)|}{1+\sqrt{n}\|\bm u\|}>6\varepsilon \right)\leq &\sum_{k=0}^{k_n}P\left(\sup_{\bm u\in A(k)}\frac{|D_n(\bm u)|}{1+\sqrt{n}\|\bm u\|}>6\varepsilon,E_n\right)+P(E_n^{c})\nonumber\\
	\leq &\frac{1}{n}\sum_{k=0}^{k_{\varepsilon}-1}\frac{2K_{\varepsilon}\pi_n(\delta_k)}{\varepsilon^{2}\eta^{2}}2^{2k}+\frac{\varepsilon}{n}\sum_{k=k_{\varepsilon}}^{k_n}2^{2k}+P(E_n^{c})\nonumber\\
	\leq &O\left(\frac{1}{n}\right)+4\varepsilon+P(E_n^{c}).
\end{align}
Since $1+\sqrt{n}\|\bm u\|>1$ and $\sqrt{n}\delta_{k_n+1}<1$, similar to the proof of \eqref{a12} and \eqref{a14}, we have
\begin{align*}
	P\left(\sup_{\bm u\in B}\frac{D_n(\bm u)}{1+\sqrt{n}\|\bm u\|}>3\varepsilon,E_n\right)
	\leq &P\left(\max_{1\leq i\leq K_{\varepsilon}}D_n^{+}(\bm u_i)>\varepsilon,E_n\right)
	\leq \frac{K_{\varepsilon}\pi_n(\delta_{k_n+1})}{\varepsilon^2}.
\end{align*}
We can get the same bound for the lower tail. Therefore, we have
\begin{align}\label{a18}
	P\left(\sup_{\bm u\in B}\frac{|D_n(\bm u)|}{1+\sqrt{n}\|\bm u\|}>3\varepsilon \right)\leq &P\left(\sup_{\bm u\in B}\frac{|D_n(\bm u)|}{1+\sqrt{n}\|\bm u\|}>3\varepsilon,E_n\right)+P(E_n^c)\nonumber\\
	\leq &\frac{2K_{\varepsilon}\pi_n(\delta_{k_n+1})}{\varepsilon^2}+P(E_n^c).
\end{align}
Note that $\pi_n(\delta_{k_n+1})\rightarrow 0$ as $n\rightarrow\infty$. Furthermore, $P(E_n)\rightarrow1$ by the ergodic theorem. Hence, $P(E_n^{c})\rightarrow 0$ as $n\rightarrow\infty$. 
Finally, \eqref{a8} follows by \eqref{a17} and \eqref{a18}. 
We can verify \eqref{a8} similarly for the case of $m_t<0$. This completes the proof of this lemma.
\end{proof}

\begin{lemma}\label{lemma3}
Suppose that Assumptions \ref{assum1-compactness}, \ref{assum2-stationarity} and \ref{assum3EQMLE-moment and density} hold, then for any sequence of random variables $\bm u$ such that $\bm u=o_p(1)$, it follows that
\[(n-p)[L_n(\bm \theta_0+\bm u) - L_n(\bm \theta_0)]=\sqrt{n}\bm u^{'}\bm T_n+(\sqrt{n}\bm u)^{'}\Sigma_2 (\sqrt{n}\bm u)+o_p(\sqrt{n}\|\bm u\|+n\|\bm u\|^{2}),\]
where $\bm T_n=\bm T_{1n}+\bm T_{2n}$ with $\bm T_{1n}=n^{-1/2}\sum_{t=p+1}^n \bm K_{1t}[I(\eta_t<0)-I(\eta_t>0)]$ and $\bm T_{2n}=n^{-1/2}\sum_{t=p+1}^n\bm K_{2t}(1-|\eta_t|)$, and $\Sigma_2 = \diag \left \{f(0)E(\bm Y_{1t}\bm Y_{1t}^\prime),E(\bm Y_{2t}\bm Y_{2t}^\prime)/2\right\}$. 
Moreover, 
\[\bm T_n\rightarrow _{\mathcal{L}} N(0,\Omega_2) \quad as \quad n\rightarrow\infty,\]
where 
\[\Omega_2 = \begin{pmatrix} E(\bm Y_{1t}\bm Y_{1t}^\prime) & \kappa_1 E(\bm Y_{1t}\bm Y_{2t}^\prime) \\ \kappa_1 E(\bm Y_{2t}\bm Y_{1t}^\prime) & \kappa_2 E(\bm Y_{2t}\bm Y_{2t}^\prime) \end{pmatrix} \quad \text{with} \quad \kappa_1=E(\eta_t)\;\;\text{and}\;\; \kappa_2=E(\eta_t^2)-1.\]
\end{lemma}
\begin{proof}
We first re-parameterize the objective function as
\[H_n(\bm u) = (n-p)\left[L_n^E(\bm \theta_0+\bm u) - L_n^E(\bm \theta_0)\right],\]
where $\bm u\in \Lambda\equiv\{\bm u = (\bm u_1^\prime,\bm u_2^\prime)^\prime:\bm u+\bm \theta_0 \in \Theta\}$. Let $\widehat{\bm u}_n=\widehat{\bm \theta}_n-\bm \theta_0$. Then, $\widehat{\bm u}_n$ is the minimizer of $H_n(\bm u)$ in $\Lambda$. Furthermore, we have
\begin{align*}
	H_n(\bm u) = \sum_{t=p+1}^nA_t(\bm u) + \sum_{t=p+1}^nB_t(\bm u) +\sum_{t=p+1}^nC_t(\bm u),
\end{align*}
where
\[A_t(\bm u)=\frac{1}{h_t(\bm \delta_0)}\left[|\varepsilon_t(\bm\alpha_0+\bm u_1)|-|\varepsilon_t(\bm\alpha_0)|\right],\]
\[B_t(\bm u)=\ln h_t(\bm \delta_0+\bm u_2)-\ln h_t(\bm \delta_0)+\frac{|\varepsilon_t(\bm\alpha_0)|}{h_t(\bm \delta_0+\bm u_2)}-\frac{|\varepsilon_t(\bm\alpha_0)|}{h_t(\bm \delta_0)},\]
\[C_t(\bm u)=\left[\frac{1}{h_t(\bm \delta_0+\bm u_2)}-\frac{1}{h_t(\bm \delta_0)}\right]\left[|\varepsilon_t(\bm\alpha_0+\bm u_1)|-|\varepsilon_t(\bm\alpha_0)|\right].\]
Recall that $\bm K_{1t} = (\bm Y_{1t}^{'},\bm 0_{p+1}^{\prime})^\prime$. By Taylor's expansion, we can show that
\begin{align}\label{a19}
\frac{\varepsilon_t(\bm\alpha_0+\bm u_1)}{h_t(\bm \delta_0)}-\frac{\varepsilon_t(\bm\alpha_0)}{h_t(\bm \delta_0)}=-\frac{\bm u_1^\prime(y_{t-1},...,y_{t-p})^\prime}{h_t(\bm \delta_0)}=-\bm u^\prime \bm K_{1t}.
	\end{align}
Let $I(\cdot)$ be the indicator function. For $x\neq0$, using the identity
\begin{align}\label{a20}
	|x-y|- |x| = -y[I(x>0)-I(x<0)]+2\int_0^y[I(x\leq s)-I(x\leq0)]ds,
\end{align}
together with \eqref{a19}, we have
\begin{align}\label{KightEqVarepsilon_t}
|\varepsilon_t(\bm\alpha_0+\bm u_1)|-|\varepsilon_t(\bm\alpha_0)|=&-\bm u^{'}\bm K_{1t}h_t[I(\eta_t>0)-I(\eta_t<0)]\nonumber\\
&+2h_t\int _0^{\bm u^{'}\bm K_{1t}} X_t\left(s\right)ds.
\end{align}
Therefore, it follows that
\begin{align*}
	A_t(\bm u)=\bm u^\prime \bm K_{1t}[I(\eta_t<0)-I(\eta_t>0)]+2\int_0^{\bm u^\prime \bm K_{1t}}X_t(s)ds , 
\end{align*}
where $X_t(s)=I(\eta_t\leq s)-I(\eta_t\leq 0)$. Then we have
\begin{align}\label{a22At}
	\sum_{t=p+1}^nA_t(\bm u)=\sqrt{n}\bm u^\prime \bm T_{1n}+\Pi_{1n}(\bm u)+\Pi_{2n}(\bm u),
\end{align}
where 
\[\bm T_{1n}=\dfrac{1}{\sqrt{n}}\sum_{t=p+1}^n \bm K_{1t}[I(\eta_t<0)-I(\eta_t>0)],\]
\[\Pi_{1n}(\bm u)=2\sum_{t=p+1}^n\int_0^{\bm u^\prime \bm K_{1t}}\left\{X_t(s)-E\left[X_t(s)|\mathcal{F}_{t-1}\right]\right\}ds,\]
\[\Pi_{2n}(\bm u)=2\sum_{t=p+1}^n\int_0^{\bm u^\prime \bm K_{1t}}E\left[X_t(s)|\mathcal{F}_{t-1}\right]ds.\]

For $\Pi_{1n}(\bm u)$ with $\bm u=o_p(1)$, by Lemma 2, it holds that
\begin{align}\label{a23}
	\Pi_{1n}(\bm u)=o_p(\sqrt{n}\|\bm u\|+n\|\bm u\|^{2}).
\end{align}

For $\Pi_{2n}(\bm u)$, by iterated-expectation and Taylor's expansion, we have
\begin{align}\label{a24}
	\Pi_{2n}(\bm u)
	&=2\sum_{t=p+1}^{n}\int_{0}^{\bm u^{'}\bm K_{1t}}E\left[I(\eta_t\leq s)-I(\eta_t\leq 0)|\mathcal{F}_{t-1}\right]ds\nonumber\\
	&=2\sum_{t=p+1}^{n}\int_{0}^{\bm u^{'}\bm K_{1t}}[F(s)-F(0)]ds\nonumber\\
	&=2\sum_{t=p+1}^{n}\int_{0}^{\bm u^{'}\bm K_{1t}}sf(0)ds+2\sum_{t=p+1}^{n}\int_{0}^{\bm u^{'}\bm K_{1t}}s[f(s^*)-f(0)]ds\nonumber\\
	&=(\sqrt{n}\bm u)^{'}[K_{1n}+K_{2n}(\bm u)](\sqrt{n}\bm u),
\end{align}
where $s^*$ lies between $0$ and $s$,
\begin{align*}
	K_{1n}=\frac{f(0)}{
	n} \sum_{t=p+1}^{n}\bm K_{1t}\bm K_{1t}^{'}\text{ and }
	K_{2n}(\bm u)=\frac{2}{n\|\bm u\|^2}\sum_{t=p+1}^{n}\int_{0}^{\bm u^{'}\bm K_{1t}}s[f(s^*)-f(0)]ds.
\end{align*}
By the ergodic theorem, it holds that
\begin{align}\label{a25}
	K_{1n}=f(0)E[\bm K_{1t}\bm K_{1t}^{'}]+o_p(1)=\Sigma_{21}+o_p(1),
\end{align}
where $\Sigma_{21}=\diag\left\{f(0)E(\bm Y_{1t}\bm Y_{1t}^\prime),0_{(p+1)\times(p+1)}\right\}$. 
Furthermore, since $\|\bm K_{1t}\|\leq p/\underline{\beta}$ by Assumption \ref{assum1-compactness} and $\|\bm u\|\leq \eta$, we have $|\bm u^{'}\bm K_{1t}|\leq p\eta/\underline{\beta}$. Then for any $\eta>0$, it holds that
\begin{align*}
	\sup_{\|\bm u\|<\eta}| K_{2n}(\bm u)| 
	\leq &\sup_{\|\bm u\|<\eta}\frac{2}{n\|\bm u\|^2}\sum_{t=p+1}^{n}\int_{0}^{|\bm u^{'}\bm K_{1t}|}s|f(s^{*})-f(0)|ds\nonumber\\
	 \leq &\frac{1}{n\|\bm u\|^{2}}\sum_{t=p+1}^{n}\|\bm u\|^2\|\bm K_{1t}\|^{2}\left[\sup_{0\leq s\leq p\eta/\underline{\beta}}|f(s)-f(0)|\right]\nonumber\\
	 \leq &\frac{p^2}{n\underline{\beta}^2}\sum_{t=p+1}^{n}\left[\sup_{0\leq s\leq p\eta/\underline{\beta}}|f(s)-f(0)|\right].
	  \end{align*}
By the dominated convergence theorem and $\sup_{x}f(x)<\infty$ by Assumption \ref{assum3EQMLE-moment and density}(ii), we have 
\[\lim_{\eta\rightarrow 0}E\left[\sup_{0\leq s\leq p\eta/\underline{\beta}}|f(s)-f(0)|\right]=E\left[\lim_{\eta\rightarrow 0}\sup_{0\leq s\leq p\eta/\underline{\beta}}|f(s)-f(0)|\right]=0.\]
Therefore, by Markov's theorem and the stationarity of $\{y_t\}$ by Assumption \ref{assum2-stationarity}, for $\forall \varepsilon, \delta>0$, there exists $\eta_0=\eta_0(\varepsilon)>0$ such that
\begin{align}\label{a26}
	P\left(\sup_{\|\bm u\|\leq\eta_0}|K_{2n}(\bm u)|>\delta\right)<\frac{\varepsilon}{2},
\end{align}
for all $n\geq 1$. On the other hand, since $\bm u=o_p(1)$, it follows that 
\begin{align}\label{a27}
	P(\|\bm u\|>\eta_0)<\frac{\varepsilon}{2}
\end{align}
as $n$ is large enough. By \eqref{a26} and \eqref{a27}, for $\forall\varepsilon,\delta>0,$ we have
\begin{align*}
	P\left(|K_{2n}(\bm u)|>\delta\right)=&P(|K_{2n}(\bm u)|>\delta,\|\bm u\|\leq\eta_0)+P(|K_{2n}(\bm u)|>\delta,\|\bm u\|>\eta_0)\nonumber\\
	\leq &P\left(\sup_{\|\bm u\|\leq\eta_0}|K_{2n}(\bm u)|>\delta\right)+P(\|\bm u\|>\eta_0)<\varepsilon
\end{align*}
as $n$ is large enough. Hence, it holds that $K_{2n}(\bm u)=o_p(1)$. Furthermore, combining \eqref{a24} and \eqref{a25}, we can show that 
\begin{align}\label{a28}
	\Pi_{2n}(\bm u)=(\sqrt{n}\bm u)^\prime\Sigma_{21}(\sqrt{n}\bm u)+o_p(n\|\bm u\|^{2}).
\end{align}
As a result, combining \eqref{a22At}, \eqref{a23} and \eqref{a28}, we have
\begin{align}\label{SumAt}
\sum_{t=p+1}^nA_t(\bm u)=\sqrt{n}\bm u^\prime \bm T_{1n} + (\sqrt{n}\bm u)^\prime\Sigma_{21}(\sqrt{n}\bm u) + o_p(\sqrt{n}\|\bm u\|+n\|\bm u\|^{2}).
\end{align}

We next consider $B_t(\bm u)$. Note that $\partial h_t(\bm\delta_0)/\partial \bm\delta= h_t\bm Y_{2t}$ and $\varepsilon_t(\bm\alpha_0)=h_t\eta_t$. 
By Taylor's expansion, we can show that
 \begin{align*}
	B_t(\bm u)&=\ln h_t(\bm \delta_0+\bm u_2)+\frac{|\varepsilon_t(\bm\alpha_0)|}{h_t(\bm \delta_0+\bm u_2)}-\left(\ln h_t(\bm \delta_0)+\frac{|\varepsilon_t(\bm\alpha_0)|}{h_t(\bm \delta_0)}\right)\nonumber\\
	&=\bm u_2^\prime \bm Y_{2t}-|\eta_t|\bm u_2^\prime \bm Y_{2t}-\frac{\bm u_2^\prime\frac{\partial h_t( \bar{\bm\delta})}{\partial \bm \delta}\frac{\partial h_t( \bar{\bm\delta})}{\partial \bm \delta^\prime}\bm u_2}{2h_t^2( \bar{\bm\delta})}+\frac{|\varepsilon_t(\bm\alpha_0)|\bm u_2^\prime\frac{\partial h_t( \bar{\bm\delta})}{\partial \bm \delta}\frac{\partial h_t( \bar{\bm\delta})}{\partial \bm \delta^\prime}\bm u_2}{h_t^{3}( \bar{\bm\delta})}\nonumber\\
	&=\bm u^\prime\bm K_{2t}(1-|\eta_t|)+\bm u^\prime\left(\frac{|\varepsilon_t
	(\bm\alpha_0)|}{h_t( \bar{\bm\delta})}-\frac{1}{2}\right)\frac{1}{h_t^{2}( \bar{\bm\delta})}\frac{\partial h_t( \bar{\bm\delta})}{\partial \bm \theta}\frac{\partial h_t( \bar{\bm\delta})}{\partial \bm \theta^\prime} \bm u,
\end{align*}
where $\bar{\bm\delta}$ lies between $\bm \delta_0$ and $\bm \delta_0 +\bm u_2$, and $\bm K_{2t}=(\bm 0_{p}^{\prime},\bm Y_{2t}^\prime)^\prime$. Then, we have
\begin{align}\label{Bt}
	\sum_{t=p+1}^nB_t(\bm u)=\sqrt{n}\bm u^\prime \bm T_{2n}+\Pi_{3n}(\bm u),
\end{align}
where
\[\bm T_{2n}=\dfrac{1}{\sqrt{n}}\sum_{t=p+1}^n\bm K_{2t}(1-|\eta_t|)\text{ and } 
   \Pi_{3n}(\bm u)=\bm u^\prime\sum_{t=p+1}^{n}\left(\frac{|\varepsilon_t
	(\bm\alpha_0)|}{h_t( \bar{\bm\delta})}-\frac{1}{2}\right)\frac{1}{h_t^{2}( \bar{\bm\delta})}\frac{\partial h_t(\bar{\bm\delta})}{\partial \bm \theta}\frac{\partial h_t(\bar{\bm\delta})}{\partial \bm \theta^\prime}\bm u.\]
Rewrite $\Pi_{3n}(\bm u)$ as $\Pi_{3n}(\bm u)=(\sqrt{n}\bm u)^{'}\left[n^{-1}\sum_{t=p+1}^{n}J_t(\bar{\bm\delta})\right](\sqrt{n}\bm u)$, where 
\[J_t(\bar{\bm\delta})=\left(\frac{|\varepsilon_t
	(\bm\alpha_0)|}{h_t(\bar{\bm\delta})}-\frac{1}{2}\right)\frac{1}{h_t^{2}(\bar{\bm\delta})}\frac{\partial h_t(\bar{\bm\delta})}{\partial \bm \theta}\frac{\partial h_t(\bar{\bm\delta})}{\partial \bm \theta^{'}}.\]
By the ergodic theorem, it is easy to see that
\[\frac{1}{n}\sum_{t=p+1}^{n}J_t(\bar{\bm\delta})=E[J_t(\bar{\bm\delta})]+o_p(1).\]
Since $\bar{\bm\delta}\rightarrow\bm \delta_0$ almost surely, by the dominated convergence theorem, we have
\[\lim_{n\rightarrow\infty}E[J_t(\bar{\bm\delta})]=E[J_t(\bm \theta_0)]=\Sigma_{22},\]
where $\Sigma_{22}=\diag\left\{0_{p\times p},E(\bm Y_{2t}\bm Y_{2t}^\prime)/2\right\}$. Thus, for $\bm u=o_p(1)$, we can get that
\begin{align}\label{a29}
	\Pi_{3n}(\bm u)=(\sqrt{n}\bm u)^\prime\Sigma_2 (\sqrt{n}\bm u)+o_p(n\|\bm u\|^{2}).
\end{align}
Therefore, \eqref{Bt} and \eqref{a29} imply that
\begin{align}\label{SumBt}
\sum_{t=p+1}^nB_t(\bm u)=\sqrt{n}\bm u^\prime \bm T_{2n} + (\sqrt{n}\bm u)^\prime\Sigma_{22} (\sqrt{n}\bm u)+o_p(n\|\bm u\|^{2}).
\end{align}

Finally, we consider $C_t(\bm u)$. By Taylor's expansion, we have
\[\frac{1}{h_t(\bm \delta_0+\bm u_2)}-\frac{1}{h_t(\bm \delta_0)}=-\frac{\bm u^{'}}{h_t^{2}(\bm \delta_0)}\frac{\partial h_t(\bar{\bm\delta})}{\partial\bm \theta},\]
where $\bar{\bm\delta}$ lies between $\bm \delta_0$ and $\bm \delta_0+\bm u_2$. This together with \eqref{KightEqVarepsilon_t}, implies that
\begin{align*}
	\sum_{t=p+1}^{n}C_t(\bm u)=(\sqrt{n}\bm u)^{'}[K_{3n}(\bm u)+K_{4n}(\bm u)](\sqrt{n}\bm u),
\end{align*}
where
\[K_{3n}(\bm u)=\frac{1}{n}\sum_{t=p+1}^{n}\frac{1}{h_t}\frac{\partial h_t(\bar{\bm\delta})}{\partial\bm \theta}\bm K_{1t}^{\prime}[I(\eta_t<0)-I(\eta_t>0)],\]
\[K_{4n}(\bm u)=-\frac{1}{n}\sum_{t=p+1}^{n}\frac{2}{h_t}\frac{\partial h_t(\bar{\bm\delta})}{\partial\bm \theta}\bm K_{1t}^{\prime}\int _0^1 X_t(\bm u^{'}\bm K_{1t}s)ds.\]
Since $\eta_t$ has zero median by Assumption \ref{assum3EQMLE-moment and density}(i), then by the ergodic theorem, we have
\[K_{3n}(\bm u)=E\left\{\frac{1}{h_t}\frac{\partial h_t(\bar{\bm\delta})}{\partial\bm \theta}\bm K_{1t}^{\prime}[I(\eta_t<0)-I(\eta_t>0)]\right\}+o_p(1)=o_p(1).\]
Since $\bm u=o_p(1)$, similar to the proof of $K_{2n}(\bm u)$, we can show that $K_{4n}(\bm u)=o_p(1)$. Therefore, it follows that
\begin{align}\label{SumCt}
\sum_{t=p+1}^{n}C_t(\bm u)=o_p(n\|\bm u\|^{2})
\end{align}
Combining \eqref{SumAt}, \eqref{SumBt} and \eqref{SumCt}, we have
\[(n-p)[L_n(\bm \theta_0+\bm u) - L_n(\bm \theta_0)]=\sqrt{n}\bm u^{'}\bm T_n+(\sqrt{n}\bm u)^{'}\Sigma_2 (\sqrt{n}\bm u)+o_p(\sqrt{n}\|\bm u\|+n\|\bm u\|^{2}),\]
where $\bm T_n=\bm T_{1n}+\bm T_{2n}$ and $\Sigma_2 = \Sigma_{21}+\Sigma_{22} = \diag \left \{f(0)E(\bm Y_{1t}\bm Y_{1t}^\prime),E(\bm Y_{2t}\bm Y_{2t}^\prime)/2\right\}$. 

Moreover, let $\bm G_t=\left(\bm Y_{1t}^\prime[I(\eta_t<0)-I(\eta_t>0)],\bm Y_{2t}^\prime(1-|\eta_t|)\right)^\prime$, we have that $\bm T_n=\bm T_{1n}+\bm T_{2n}=\sum_{t=p+1}^{n}\bm G_t$.
If Assumptions \ref{assum1-compactness}, \ref{assum2-stationarity} and \ref{assum3EQMLE-moment and density} hold, by the Central Limit Theorem, we have 
\begin{align}\label{a30}
\bm T_n\rightarrow _{\mathcal{L}} N(0,\Omega_2) \quad as \quad n\rightarrow\infty,
\end{align}
where $\Omega_2$ is defined as in Lemma \ref{lemma3}. We accomplish the proof of this lemma.
\end{proof}

\begin{proof}[Proof of Theorem \ref{thm2EQMLE-normality}]
	We have $\widehat{\bm u}_n=\widehat{\bm \theta}_n-\bm \theta_0=o_p(1)$ by Theorem \ref{thm1EQMLE-consistency}. Furthermore, by Lemma \ref{lemma3}, we have
\begin{align}\label{a33}
	H_n(\widehat{\bm u}_n)&=\sqrt{n}\widehat{\bm u}_n^{'}\bm T_n+(\sqrt{n}\widehat{\bm u}_n)^{'}\Sigma_2 (\sqrt{n}\widehat{\bm u}_n)+o_p(\sqrt{n}\|\widehat{\bm u}_n\|+n\|\widehat{\bm u}_n\|^{2}) \\
	&\geq -\sqrt{n}\|\widehat{\bm u}_n\|\left[\left\|\bm T_n\right\|+o_p(1)\right]+n\|\widehat{\bm u}_n\|^{2}[\lambda_{\min}+o_p(1)], \nonumber 
\end{align}
where $\lambda_{\min}>0$ is the minimum eigenvalue of $\Sigma_2$. 
Note that $H_n(\widehat{\bm u}_n)=(n-p)[L_n^E(\widehat{\bm \theta}_n)-L_n^E(\bm \theta_0)]\leq0$. Then it follows that
\begin{align}\label{a34}
	n\|\widehat{\bm u}_n\|\leq[\lambda_{\min}+o_p(1)]^{-1}\left[\left\|\bm T_n\right\|+o_p(1)\right]=O_p(1).
\end{align}
This together with Theorem \ref{thm1EQMLE-consistency}, verifies the $\sqrt{n}$-consistency. Hence, Statement (i) holds. 

Next, let $\sqrt{n}\bm u_n^{*}=-\Sigma_2^{-1}\bm T_n/2$, then, by \eqref{a30} from Lemma \ref{lemma3}, we have
\[\sqrt{n}\bm u_n^{*}\rightarrow_{\mathcal{L}} N\left(0,\frac{1}{4}\Sigma_2^{-1}\Omega_2\Sigma_2^{-1}\right)\text{ as } n\rightarrow \infty .\]
As a result, it is sufficient to show that $\sqrt{n}\widehat{\bm u}_{n}-\sqrt{n}\bm u_n^{*}=o_p(1)$. By \eqref{a33} and \eqref{a34}, we have
\begin{align*}
	H_n(\widehat{\bm u}_n)=&(\sqrt{n}\widehat{\bm u}_n)^{'}\bm T_n+(\sqrt{n}\widehat{\bm u}_n)^{'}\Sigma_2 (\sqrt{n}\widehat{\bm u}_n)+o_p(1) \\
    =&-2(\sqrt{n}\widehat{\bm u}_n)^{'}\Sigma_2(\sqrt{n}\bm u_n^{*})+(\sqrt{n}\widehat{\bm u}_n)^{'}\Sigma_2 (\sqrt{n}\widehat{\bm u}_n)+o_p(1).
\end{align*}
Note that \eqref{a33} still holds when $\widehat{\bm u}_n$ is replaced by $\bm u_n^*$. Thus,
\begin{align*}
	H_n(\bm u_n^*)=&(\sqrt{n}\bm u_n^*)^{'}\bm T_n+(\sqrt{n}\bm u_n^*)^{'}\Sigma_2 (\sqrt{n}\bm u_n^*)+o_p(1)\nonumber\\
    =&-(\sqrt{n}\bm u_n^*)^{'}\Sigma_2(\sqrt{n}\bm u_n^{*})+o_p(1).
\end{align*}
By the previous two equations, it follows that
\begin{align}\label{a35}
	H_n(\widehat{\bm u}_n)-H_n(\bm u_n^*)=&(\sqrt{n}\widehat{\bm u}_n-\sqrt{n}\bm u_n^*)^{'}\Sigma_2(\sqrt{n}\widehat{\bm u}_n-\sqrt{n}\bm u_n^*)+o_p(1)\nonumber\\
	\geq&\lambda_{\min}\|\sqrt{n}\widehat{\bm u}_n-\sqrt{n}\bm u_n^*\|^{2}+o_p(1).
\end{align}
Since $H_n(\widehat{\bm u}_n)-H_n(\bm u_n^{*})=(n-p)\left[L_n^E(\bm \theta_0+\widehat{\bm u}_n)-L_n^E(\bm \theta_0+\bm u_n^{*})\right]\leq 0$ almost surely, by \eqref{a35} we have $\|\sqrt{n}\widehat{\bm u}_{n}-\sqrt{n}\bm u_n^{*}\|=o_p(1)$. Therefore, we have
\[\sqrt{n}\widehat{\bm u}_n\rightarrow_{\mathcal{L}} N\left(0,\frac{1}{4}\Sigma_2^{-1}\Omega_2\Sigma_2^{-1}\right)\text{ as }n\rightarrow \infty.\]
Therefore, Statement (ii) holds. The proof is accomplished. 
\end{proof}

\subsection{Proof of Remark \ref{remark1}}

\begin{proof}
	Recall that the error function is $\eta_t(\bm \theta)=\varepsilon_t(\bm\alpha)/h_t(\bm \delta)$, where $\varepsilon_t(\bm\alpha) = y_t-\sum_{i=1}^{p}\alpha_{i}y_{t-i}$ and $h_t(\bm\delta) = \omega  +\sum_{i=1}^{p}\beta_{i}|y_{t-i}|$, and the residuals are defined as  $\widehat{\eta}_t=\eta_t(\widehat{\bm \theta}_n)=\varepsilon_t(\widehat{\bm\alpha}_n)/h_t(\widehat{\bm \delta}_n)$. 
	Note that $\eta_t=\eta_t(\bm \theta_0)$ and $\widehat{f}(0)=(nb_n)^{-1}\sum_{t=p+1}^{n}K(\widehat{\eta}_{t}/b_n)$. 	 
	Since $|K(x)-K(y)|\leq L|x-y|$ for some $L>0$, by Taylor's expansion, it holds that
	\begin{align}\label{a36}
		\left|\widehat{f}(0)-\frac{1}{nb_{n}}\sum_{t=p+1}^{n}K\left(\frac{\eta_t}{b_n}\right)\right|&=\frac{1}{nb_{n}}\sum_{t=p+1}^{n}\left|K\left(\frac{\widehat{\eta}_t}{b_n}\right)-K\left(\frac{\eta_t}{b_n}\right)\right|\nonumber\\
		&\leq\frac{L}{nb_{n}^{2}}\sum_{t=p+1}^{n}|\widehat{\eta}_t-\eta_t|\nonumber\\
		&\leq\frac{L\|\widehat{\bm \theta}_n-\bm \theta_0\|}{nb_{n}^{2}}\sum_{t=p+1}^{n}\left\|\frac{\partial\eta_t(\bm \theta^{\ast})}{\partial\bm\theta} \right\|,
	\end{align}
	where $\bm \theta^{\ast}$ lies between $\bm \theta_0$ and $\widehat{\bm \theta}_n$. It can be verified that 
	\begin{align}\label{Partial_eta_t}
		\frac{\partial\eta_t(\bm \theta)}{\partial\bm \theta}=\frac{1}{h_t(\bm \delta)}\frac{\partial\varepsilon_t(\bm\alpha)}{\partial\bm \theta}-\frac{\varepsilon_t(\bm\alpha)}{h_t^{2}(\bm \delta)}\frac{\partial h_t(\bm \delta)}{\partial\bm \theta}=-\frac{h_t}{h_t(\bm \delta)}\bm K_{1t}-\frac{h_t\varepsilon_t(\bm\alpha)}{h_t^{2}(\bm \delta)}\bm K_{2t},
	\end{align}
	where $\bm K_{1t} = (\bm Y_{1t}^{'},\bm 0_{p+1}^{\prime})^\prime$ and $\bm K_{2t}=(\bm 0_{p}^{\prime},\bm Y_{2t}^\prime)^\prime$. 
	Since $\|\bm K_{1t}\|\leq p/\underline{\beta}$ by \eqref{K1t_bound}, we can show that $\|\bm K_{2t}\|\leq \underline{\omega}^{-1} +p/\underline{\beta}$, and $\sup_{\bm \theta \in \Theta}h_t/h_t(\bm \delta)\leq \omega_0/\underline{\omega}+\sum_{i=1}^{p}\beta_{i0}/\underline{\beta}$ by Assumption \ref{assum1-compactness}. 
	These together with \eqref{A.2}, imply that
	\begin{align*}
		E\left[\sup_{\bm \theta \in \Theta}\left\|\dfrac{\partial\eta_t(\bm \theta)}{\partial\bm \theta}\right\|\right]
		&\leq E\left[\sup_{\bm \theta \in \Theta}\dfrac{h_t}{h_t(\bm \delta)}\|\bm K_{1t}\|\right] + E\left[\sup_{\bm \theta \in \Theta}\dfrac{h_t}{h_t(\bm \delta)}\dfrac{|\varepsilon_t(\bm\alpha)|}{h_t(\bm \delta)}\|\bm K_{2t}\|\right] <\infty.
	\end{align*}
	Then by Theorem 3.1 in \cite{Ling_McAleer2003} and the Dominated Convergence Theorem, we have
	\begin{align}\label{a37}
		\frac{1}{n}\sum_{t=p+1}^{n}\left\|\frac{\partial\eta_t(\bm \theta^{\ast})}{\partial\bm \theta} \right\|=E\left\|\frac{\partial\eta_t(\bm \theta^{\ast})}{\partial\bm \theta} \right\|+o_p(1)=E\left\|\frac{\partial\eta_t(\bm \theta_0)}{\partial\bm \theta} \right\|+o_p(1)=O_p(1).
	\end{align}
	Since $\sqrt{n}(\widehat{\bm \theta}_n-\bm \theta_0)=O_p(1)$ by Theorem \ref{thm2EQMLE-normality} and $nb_n^{4}\rightarrow\infty$ as $n\rightarrow\infty$, by \eqref{a36} and \eqref{a37}, we have
	\begin{align}\label{a38}
		\left|\widehat{f}(0)-\frac{1}{nb_{n}}\sum_{t=p+1}^{n}K\left(\frac{\eta_t}{b_n}\right)\right|\leq O_p\left(\frac{1}{\sqrt{n}b_n^{2}}\right)=o_p(1).
	\end{align}
	Moreover, since $\sup_{x}f(x)<\infty$ by Assumption \ref{assum3EQMLE-moment and density}(ii) and $\int_{-\infty}^{\infty}K(x)dx=1$, it holds that
	\[E\left[\frac{1}{b_n}K\left(\frac{\eta_t}{b_n}\right)\right]=\int_{-\infty}^{\infty}K(x)f(b_{n}x)dx<\infty.\]
	Then, by Theorem 3.1 in \cite{Ling_McAleer2003}, it follows that
	\begin{align}\label{a39}
		\frac{1}{nb_{n}}\sum_{t=p+1}^{n}K\left(\frac{\eta_t}{b_n}\right)=E\left[\frac{1}{b_n}K\left(\frac{\eta_t}{b_n}\right)\right]+o_p(1).
	\end{align}
	Furthermore, since $\int_{-\infty}^{\infty}|x|K(x)dx < \infty$, $\sup_x|f^\prime(x)|<\infty$ and $b_n\rightarrow 0$ as $n\rightarrow \infty$, it can be verified that
	\begin{align}\label{a40}
		\left|E\left[\frac{1}{b_n}K\left(\frac{\eta_t}{b_n}\right)\right]-f(0)\right|&=\left|\int_{-\infty}^{\infty}K(x)[f(b_{n}x)-f(0)]dx\right|\nonumber\\
		&\leq b_n\sup_{x}|f^{'}(x)|\int_{-\infty}^{\infty}|x|K(x)dx\rightarrow 0 \;\; \text{as} \;\; n\rightarrow \infty.
	\end{align}
	By \eqref{a38}--\eqref{a40}, we have $\widehat{f}(0)=f(0)+o_p(1)$. Finally, by a similar argument as for \eqref{a37}, we can show that $(n-p)^{-1}\sum_{t=p+1}^{n}\bm Y_{1t}\bm Y_{1t}^{'}=E(\bm Y_{1t}\bm Y_{1t}^{'})+o_p(1)$, $(n-p)^{-1}\sum_{t=p+1}^{n}\bm Y_{2t}\bm Y_{2t}^{'}=E(\bm Y_{2t}\bm Y_{2t}^{'})+o_p(1)$, $(n-p)^{-1}\sum_{t=p+1}^{n}\bm Y_{1t}\bm Y_{2t}^{'}=E(\bm Y_{1t}\bm Y_{2t}^{'})+o_p(1)$, 
	$\widehat{\kappa}_1=(n-p)^{-1}\sum_{t=p+1}^{n}\widehat{\eta}_t=E(\eta_t)+o_p(1)$ and $\widehat{\kappa}_2=(n-p)^{-1}\sum_{t=p+1}^{n}\widehat{\eta}_t^2-1=E(\eta_t^{2})-1+o_p(1)$. 
	As a result, $\widehat{\Sigma}_{n}\rightarrow_p\Sigma$ and $\widehat{\Omega}_{n}\rightarrow_p\Omega$ as $n\rightarrow\infty$. 
	The proof is accomplished.
\end{proof}

\subsection{Proof of Theorem \ref{thm3BIC}}
\begin{proof}	
	In the following proof, notations $\Theta^p$, $\bm\theta_{0}^{p}$, $\widehat{\bm\theta}_{n}^{p}$, $\bm T_n^{p}$, and $\Sigma^{p}$ are employed to emphasize their dependence on the order $p$. 
	Let $p_0$ be the true order of model \eqref{LDAR}. It is sufficient to show that, for any $p\neq p_0$, 
	\begin{equation}\label{goal of E-BIC}
		\lim_{n\rightarrow \infty} P(\text{BIC}^{E}(p)-\text{BIC}^{E}(p_0)>0)=1.
	\end{equation}
	We first consider the case that the model is overfitted, i.e. $p > p_0$. 
	Let $\bm\theta_{0}^{p}=\argmin_{\bm\theta\in\Theta^{p}}E\ell_t^E(\bm\theta)$ be the unique minimizer, and denote $\mathring{\bm\theta}_{0}^{p}=(\bm\alpha_{0}^{p_0\prime},\bm 0_{p-p_0}^{\prime},\bm\delta_{0}^{p_0\prime},\bm 0_{p-p_0}^{\prime})^{\prime}$, where $\bm 0_{m}$ is the $m \times 1$ vector of zeros. 
	Let $\mathring{\Theta}^{p}$ be the parameter space corresponding to $\Theta^{p}$ with $\underline{\beta}=0$, which allows the scale parameters $\beta_i$'s to be zero. Note that $E\ell_t^E(\bm\theta_{0}^{p})=\min_{\bm\theta\in\Theta^{p}}E\ell_t^E(\bm\theta)>\min_{\bm\theta\in\mathring{\Theta}^{p}}E\ell_t^E(\bm\theta)=E\ell_t^E(\mathring{\bm\theta}_{0}^{p})$, where the inequality holds since $\Theta^{p}\subset\mathring{\Theta}^{p}$ and $\mathring{\bm\theta}_{0}^{p}\notin\Theta^{p}$. This together with $E\ell_t^E(\mathring{\bm\theta}_{0}^{p}) =E\ell_t^E(\bm\theta_{0}^{p_0})$, implies that $E\ell_t^E(\bm\theta_{0}^{p}) =E\ell_t^E(\bm\theta_{0}^{p_0})+c$ for some $c>0$. 
	Then by ergodic theorem, we have
	\[ 
	L_n^E(\bm\theta_{0}^{p}) =	L_n^E(\bm\theta_{0}^{p_0})+c+o_p(1).\]
	Denote $\widehat{\bm u}_n^{p_0}=\widehat{\bm \theta}_n^{p_0}-\bm \theta_0^{p_0}$. 
	Similar to the proof of Theorems \ref{thm1EQMLE-consistency} and \ref{thm2EQMLE-normality}, we can show that $\widehat{\bm\theta}^{p_0}_{n} \to \bm\theta_0^{p_0}$ almost surely and $\sqrt{n}\widehat{\bm u}_n^{p_0}=O_p(1)$ as $n\rightarrow \infty$. 	
	Then it follows that 
	\begin{align}\label{Ln distance2}
		&(n-p_{\max})[L_n^E(\widehat{\bm\theta}_{n}^{p_0})-L_n^E(\bm\theta_{0}^{p_0})] \nonumber\\
		=&(\sqrt{n}\widehat{\bm u}_n^{p_0})^{'}\bm T_n^{p_0}+(\sqrt{n}\widehat{\bm u}_n^{p_0})^{'}\Sigma^{p_0} (\sqrt{n}\widehat{\bm u}_n^{p_0})+o_p(1)= O_p(1), 
	\end{align}
	Recall that $\bm\theta_{0}^{p}=\argmin_{\bm\theta\in\Theta^{p}}E\ell_t^E(\bm\theta)$ is the unique minimizer. Similar to the proof of Theorems \ref{thm1EQMLE-consistency}--\ref{thm2EQMLE-normality} and \eqref{Ln distance2}, it can be verified that $\widehat{\bm\theta}^{p}_{n} \to \bm\theta_0^{p}$ almost surely, $\sqrt{n}\widehat{\bm u}_n^{p}=O_p(1)$, and $(n-p_{\max})[L_n^E(\widehat{\bm\theta}_{n}^{p})-L_n^E(\bm\theta_{0}^{p})]=O_p(1)$. Therefore, we have
	\begin{align*}
		&L_n^E(\widehat{\bm\theta}_{n}^{p})-L_n^E(\widehat{\bm\theta}_{n}^{p_0}) \\
		=&[L_n^E(\widehat{\bm\theta}_{n}^{p})-L_n^E(\bm\theta_{0}^{p})]-[L_n^E(\widehat{\bm\theta}_{n}^{p_0})-L_n^E(\bm\theta_{0}^{p_0})]+[L_n^E(\bm\theta_{0}^{p})-L_n^E(\bm\theta_{0}^{p_0})]=O_p\left(\dfrac{1}{n}\right)+c+o_p(1).
	\end{align*}
	Hence, it follows that
	\begin{align}\label{BIC2_overfitted}
		&\text{BIC}^{E}(p)-\text{BIC}^{E}(p_0) \nonumber\\
		=& 2(n-p_{\max})[L_n^E(\widehat{\bm \theta}_{n}^{p})-L_n^E(\widehat{\bm \theta}_{n}^{p_0})] +[(2p+1)\ln (n-p_{\max})-(2p_0+1)\ln (n-p_{\max})] \nonumber\\
		=& O_p(1) +2(n-p_{\max})(c+o_p(1)) + 2(p-p_0)\ln(n-p_{\max}) \to \infty \;\;\text{as} \;\; n\to \infty.
	\end{align}
	
	We next consider the case that the model is underfitted, i.e. $p < p_0$. 
	Let $\bm\theta_{0,E}^p = \arg\min_{\bm\theta \in \Theta^p} E[\ell_t^E(\bm\theta^p)]$. 
	Similar to the proof of Theorems \ref{thm1EQMLE-consistency}--\ref{thm2EQMLE-normality} and \eqref{Ln distance2}, we can verify that  $\sqrt{n}(\widehat{\bm\theta}_{n}^{p}-\bm\theta_{0}^{p})=O_p(1)$ and  
	\[(n-p_{\max})[L_n^E(\widehat{\bm\theta}_{n}^{p})-L_n^E(\bm\theta_{0}^{p})] = O_p(1).\]
	Note that $L_n^E(\cdot)$ is the negative likelihood. Since the model with order $p$ corresponds to a smaller model than the true model, we have $E[\ell_t^E(\bm\theta_{0}^{p})]= E[\ell_{t}^E(\bm \theta_0^{p_0})]+\varepsilon$ for some positive constant $\varepsilon$. By ergodic theorem, we have $L_n^E(\bm\theta_{0}^p)= E[\ell_{t}^E(\bm \theta_{0}^p)]+o_p(1)$. Thus it holds that
	\begin{equation*}
		L_n^E(\bm\theta_{0}^{p})-L_n^E(\bm\theta_{0}^{p_0})=E[\ell_{t}^E(\bm\theta_{0}^{p})] - E[\ell_t^E(\bm\theta_0^{p_0})] + o_p(1)=\varepsilon + o_p(1).
	\end{equation*}
	Therefore, we have
	\begin{align*}
		L_n^E(\widehat{\bm\theta}_{n}^{p})-L_n^E(\widehat{\bm\theta}_{n}^{p_0})
		=&[L_n^E(\widehat{\bm\theta}_{n}^{p})-L_n^E(\bm\theta_{0}^{p})]-[L_n^G(\widehat{\bm\theta}_{n}^{p_0})-L_n^E(\bm\theta_{0}^{p_0})] \nonumber\\ &+ [L_n^E(\bm\theta_{0}^{p})-L_n^E(\bm\theta_{0}^{p_0})]
		=O_p\left(\dfrac{1}{n}\right)+\varepsilon + o_p(1).
	\end{align*}
	This together with $(2p+1)\ln (n-p_{\max})-(2p_0+1)\ln (n-p_{\max})=O(\ln n)$, implies that
	\begin{align}\label{BIC2_underfitted}
		&\text{BIC}^{E}(p)-\text{BIC}^{E}(p_0) \nonumber\\
		=& 2(n-p_{\max})[L_n^E(\widehat{\bm \theta}_{n}^{p})-L_n^E(\widehat{\bm \theta}_{n}^{p_0})] +\left[(2p+1)\ln (n-p_{\max})-(2p_0+1)\ln (n-p_{\max})\right] \nonumber\\
		= & 2(n-p_{\max})\varepsilon + o_p(n-p_{\max}) + O_p(1) + O(\ln n) \to \infty \;\;\text{as} \;\; n\to \infty.
	\end{align}
	Hence, combining \eqref{BIC2_overfitted} and \eqref{BIC2_underfitted} implies that \eqref{goal of E-BIC} holds. 
	The proof is accomplished. 
\end{proof}

\subsection{Proof of Remark 2}
\begin{proof}
	In the following proof, notations $\Theta^p$, $\bm\theta_{0}^p$, $\widetilde{\bm\theta}_{n}^{p}$ and $\Sigma_1^{p}$ are employed to emphasize their dependence on the order $p$. 
	Let $p_0$ be the true order of model \eqref{LDAR}, and $\bm\theta_{0}^{p}=\argmin_{\bm\theta\in\Theta^{p}}E\ell_t^G(\bm\theta)$ be the unique minimizer. It is sufficient to show that, for any $p\neq p_0$, 
	\begin{equation}\label{goal of G-BIC}
		\lim_{n\rightarrow \infty} P(\text{BIC}^{G}(p)-\text{BIC}^{G}(p_0)>0)=1.
	\end{equation}
	
	We first consider the case that the model is overfitted, i.e. $p > p_0$. 
	Similar to the proof of Theorem \ref{thm3BIC} for the overfitted case, we can show that $E\ell_t^G(\bm\theta_{0}^{p}) =E\ell_t^G(\bm\theta_{0}^{p_0})+c$ for some $c>0$. Then by ergodic theorem, it follows that
	\[ 
	L_n^G(\bm\theta_{0}^{p}) =	L_n^G(\bm\theta_{0}^{p_0})+c+o_p(1).\]
	Denote $\widetilde{\bm u}_n^{p_0}=\widetilde{\bm \theta}_n^{p_0}-\bm \theta_0^{p_0}$. 	
	Similar to the proof of Theorem \ref{thm1G-QMLE-asymptotics}, we can verify that $\widetilde{\bm\theta}^{p_0}_{n} \to_p \bm\theta_0^{p_0}$, $\sqrt{n}\partial L_n^G(\bm\theta_{0}^{p_0})/\partial \bm \theta = O_p(1)$ and $\sqrt{n}\widetilde{\bm u}_n^{p_0}=O_p(1)$ as $n\rightarrow \infty$. 
	By Taylor's expansion and Slutsky's theorem, it can be shown that
	\begin{align}\label{Ln distance1}
		&(n-p_{\max})[L_n^G(\widetilde{\bm\theta}_{n}^{p_0})-L_n^G(\bm\theta_{0}^{p_0})] \nonumber \\
		=& (\sqrt{n}\widetilde{\bm u}_n^{p_0})^{\prime}\dfrac{n-p_{\max}}{\sqrt{n}}\dfrac{\partial L_n^G(\bm\theta_{0}^{p_0})}{\partial \bm \theta}-(\sqrt{n}\widetilde{\bm u}_n^{p_0})^{\prime}\Sigma_1^{p_0}\sqrt{n}\widetilde{\bm u}_n^{p_0}+o_p(1) = O_p(1).
	\end{align}
	Similar to the proof of Theorem \ref{thm1G-QMLE-asymptotics} and \eqref{Ln distance1}, it can be verified that $\sqrt{n}\widetilde{\bm u}_n^p=O_p(1)$ and  $(n-p_{\max})[L_n^G(\widetilde{\bm\theta}_{n}^{p})-L_n^G(\bm\theta_{0}^{p})]=O_p(1)$. Therefore, we have
	\begin{align*}
		&L_n^G(\widetilde{\bm\theta}_{n}^{p})-L_n^G(\widetilde{\bm\theta}_{n}^{p_0}) \\
		=&[L_n^G(\widetilde{\bm\theta}_{n}^{p})-L_n^G(\bm\theta_{0}^{p})]-[L_n^G(\widetilde{\bm\theta}_{n}^{p_0})-L_n^G(\bm\theta_{0}^{p_0})]+[L_n^G(\bm\theta_{0}^{p})-L_n^G(\bm\theta_{0}^{p_0})]=O_p\left(\dfrac{1}{n}\right)+c+o_p(1).
	\end{align*}
	Hence, we have
	\begin{align}\label{BIC1_overfitted}
		&\text{BIC}^{G}(p)-\text{BIC}^{G}(p_0) \nonumber\\
		=& 2(n-p_{\max})[L_n^G(\widetilde{\bm \theta}_{n}^{p})-L_n^G(\widetilde{\bm \theta}_{n}^{p_0})] +[(2p+1)\ln (n-p_{\max})-(2p_0+1)\ln (n-p_{\max})] \nonumber\\
		=& O_p(1) +2(n-p_{\max})(c+o_p(1)) + 2(p-p_0)\ln(n-p_{\max}) \to \infty \;\;\text{as} \;\; n\to \infty.
	\end{align} 
	
	We next consider the case that the model is underfitted, i.e. $p < p_0$. 
	Let $\bm\theta_{0,G}^p = \arg\min_{\bm\theta \in \Theta^p} E[\ell_t^G(\bm\theta^p)]$ and $\bm\theta_{0,E}^p = \arg\min_{\bm\theta \in \Theta^p} E[\ell_t^E(\bm\theta^p)]$. 	
	Similar to the proof of Theorem \ref{thm1G-QMLE-asymptotics} and \eqref{Ln distance1}, we can verify that  $\sqrt{n}(\widetilde{\bm\theta}_{n}^{p}-\bm\theta_{0}^{p})=O_p(1)$ and  
	\[(n-p_{\max})[L_n^G(\widetilde{\bm\theta}_{n}^{p})-L_n^G(\bm\theta_{0}^{p})] = O_p(1).\]
	Note that $L_n^G(\cdot)$ is the negative likelihood. Since the model with order $p$ corresponds to a smaller model than the true model, we have $E[\ell_t^G(\bm\theta_{0}^{p})]\geq E[\ell_{t}^G(\bm \theta_0^{p_0})]+\varepsilon$ for some positive constant $\varepsilon$. By ergodic theorem, we have $L_n^G(\bm\theta_{0,G}^p)= E[\ell_{t}^G(\bm \theta_{0,G}^p)]+o_p(1)$. Thus it holds that
	\begin{equation*}
		L_n^G(\bm\theta_{0}^{p})-L_n^G(\bm\theta_{0}^{p_0})= E[\ell_{t}^G(\bm\theta_{0}^{p})] - E[\ell_t^G(\bm\theta_0^{p_0})] + o_p(1)=\varepsilon + o_p(1).
	\end{equation*}
	Therefore, we have
	\begin{align*}
		L_n^G(\widetilde{\bm\theta}_{n}^{p})-L_n^G(\widetilde{\bm\theta}_{n}^{p_0})
		=&[L_n^G(\widetilde{\bm\theta}_{n}^{p})-L_n^G(\bm\theta_{0}^{p})]-[L_n^G(\widetilde{\bm\theta}_{n}^{p_0})-L_n^G(\bm\theta_{0}^{p_0})] \nonumber\\ &+ [L_n^G(\bm\theta_{0}^{p})-L_n^G(\bm\theta_{0}^{p_0})]
		=O_p\left(\dfrac{1}{n}\right)+\varepsilon + o_p(1).
	\end{align*}
	This together with $(2p+1)\ln (n-p_{\max})-(2p_0+1)\ln (n-p_{\max})=O(\ln n)$, implies that
	\begin{align}\label{BIC1_underfitted}
		&\text{BIC}^{G}(p)-\text{BIC}^{G}(p_0) \nonumber\\
		=& 2(n-p_{\max})[L_n^G(\widetilde{\bm \theta}_{n}^{p})-L_n^G(\widetilde{\bm \theta}_{n}^{p_0})] +\left[(2p+1)\ln (n-p_{\max})-(2p_0+1)\ln (n-p_{\max})\right] \nonumber\\
		= & 2(n-p_{\max})\varepsilon + o_p(n-p_{\max}) + O_p(1) + O(\ln n) \to \infty \;\;\text{as} \;\; n\to \infty.
	\end{align}
	Hence, combining \eqref{BIC1_overfitted} and \eqref{BIC1_underfitted} implies that \eqref{goal of G-BIC} holds. 
	The proof is accomplished.
\end{proof}

\subsection{Proof of Theorem \ref{thm1G-QMLE-asymptotics}}
To show Theorem \ref{thm1G-QMLE-asymptotics}, we introduce the following Lemmas \ref{boundedness}--\ref{MCLT}. 
\begin{lemma}\label{boundedness}
	If Assumptions \ref{assum1-compactness}--\ref{assum2-stationarity} hold, then it holds that
	\[
	\text{ (i) } E\sup _{\bm \theta \in \Theta}\left|\ell^G_{t}(\bm \theta)\right|<\infty; \quad \text{ (ii) } E\sup_{\bm \theta \in \Theta} \left\| \dfrac{\partial \ell^G_{t}(\bm \theta)}{\partial \bm \theta } \right\|<\infty; \quad \text{ (iii) } E\sup _{\bm \theta \in \Theta}\left\|\dfrac{\partial^{2} \ell^G_{t}(\bm \theta)}{\partial \bm \theta \partial \bm \theta^{\prime}}\right\|<\infty.
	\]
\end{lemma}
\begin{proof}
	Recall that $\bm \theta = (\bm \alpha^\prime,\bm \delta^\prime)^{'}$ and $\ell^G_{t}(\bm \theta)=\ln h_t(\bm \delta)+0.5\varepsilon^{2}_t(\bm \alpha)/h_t^{2}(\bm \delta)$, where $\varepsilon_t(\bm\alpha)=y_t-\sum_{i=1}^p \alpha_{i}y_{t-i}$ and $h_t(\bm \delta)=\omega+\sum_{i=1}^{p}\beta_{i}|y_{t-i}|$. 
	It can be derived that
	\begin{align*}
		&\dfrac{\partial \ell^G_t(\bm \theta)}{\partial \alpha_i}=-\dfrac{\varepsilon_t(\bm \alpha)y_{t-i}}{h_t^2(\bm \delta)},\; 
		\dfrac{\partial \ell^G_t(\bm \theta)}{\partial \omega}=\dfrac{1}{h_t(\bm \delta)}\left(1-\dfrac{\varepsilon^2_t(\bm \alpha)}{h_t^2(\bm \delta)}\right), \;
		\dfrac{\partial \ell^G_t(\bm \theta)}{\partial \beta_i}=\dfrac{|y_{t-i}|}{h_t(\bm \delta)}\left(1-\dfrac{\varepsilon^2_t(\bm \alpha)}{h_t^2(\bm \delta)}\right),\\
		&\dfrac{\partial^2 \ell_{t}^G(\bm \theta)}{\partial \alpha_i \partial \alpha_j}=\dfrac{y_{t-i}y_{t-j}}{h_t^2(\bm \delta)},\;
		\dfrac{\partial^2\ell_{t}^G(\bm \theta)}{\partial \alpha_i \partial \omega}=\dfrac{2\varepsilon_t(\bm \alpha)y_{t-i}}{h_t^3(\bm \delta)},\;
		\dfrac{\partial^2\ell_{t}^G(\bm \theta)}{\partial \alpha_i \partial \beta_j}=\dfrac{2\varepsilon_t(\bm \alpha)y_{t-i}|y_{t-j}|}{h_t^3(\bm \delta)},\\
		&\dfrac{\partial^2 \ell^G_t(\bm \theta)}{\partial^2 \omega}=-\dfrac{1}{h_t^2(\bm\delta)}\left(1-\dfrac{3\varepsilon^2_t(\bm \alpha)}{h_t^2(\bm \delta)}\right), \;
		\dfrac{\partial^2 \ell^G_t(\bm \theta)}{\partial \omega \partial \beta_j}=-\dfrac{|y_{t-j}|}{h_t^2(\bm\delta)}\left(1-\dfrac{3\varepsilon^2_t(\bm \alpha)}{h_t^2(\bm \delta)}\right)	\; \text{and} \\
		&\dfrac{\partial^2 \ell^G_t(\bm \theta)}{\partial \beta_i \partial \beta_j}=-\dfrac{|y_{t-i}y_{t-j}|}{h_t^2(\bm\delta)}\left(1-\dfrac{3\varepsilon^2_t(\bm \alpha)}{h_t^2(\bm \delta)}\right).
	\end{align*}
	
	We first show (i). By Assumption \ref{assum2-stationarity}, there exists a constant $0< \kappa \leq 1$ such that $E(|y_t|^\kappa) < \infty$. Let $\overline{\omega}^\star = \max\left\lbrace 1,\overline \omega\right\rbrace$. By Jensen’s inequality, we have
	\begin{align*}
		E \ln \left(\overline{\omega}^{\star}+\overline{\beta}\sum_{i=1}^{p} |y_{t-i}|\right) &=\frac{1}{\kappa} E \ln \left(\overline{\omega}^{\star}+\overline{\beta}\sum_{i=1}^{p}  |y_{t-i}|\right)^{\kappa} \\ 
		&\leq \frac{1}{\kappa} E \ln \left(\overline{\omega}^{\star \kappa}+\overline{\beta}^\kappa\sum_{i=1}^{p} |y_{t-i}|^\kappa\right) \\ 
		&\leq \frac{1}{\kappa} \ln \left(\overline{\omega}^{\star \kappa}+\overline{\beta}^\kappa\sum_{i=1}^{p} E(|y_{t-i}|^\kappa)\right) <\infty,
	\end{align*}
	where the following elementary relation is used: $(\sum_{i=1}^{p}\alpha_i)^s \leq \sum_{i=1}^{p}\alpha_i^s$ for all $\alpha_i > 0$ and $s \in (0,1]$. 
	This together with Assumption \ref{assum2-stationarity}, implies that
	\begin{align}\label{finite 1}
		E \sup _{\bm\theta \in \Theta}\left|\ln h_t(\bm\delta)\right|  
		\leq &E \sup _{\bm\theta \in \Theta}\left[I\left(\omega+\overline{\beta}\sum_{i=1}^{p}\left|y_{t-i}\right| \geq 1\right) \ln \left(\omega+\overline{\beta}\sum_{i=1}^{p}\left|y_{t-i}\right|\right)\right]  \nonumber\\
		&+E \sup _{\bm\theta \in \Theta}\left[-I\left(\omega+\overline{\beta}\sum_{i=1}^{p}\left|y_{t-i}\right| \leq 1\right) \ln \left(\omega+\overline{\beta}\sum_{i=1}^{p}\left|y_{t-i}\right|\right)\right]  \nonumber\\
		\leq & E \ln \left(\overline{\omega}^{\star}+\overline{\beta}\sum_{i=1}^{p}  |y_{t-i}|\right)-I\{\underline{\omega}<1\} \ln \underline{\omega}<  \infty.
	\end{align}
	By Assumption \ref{assum3G-QMLE-moment and density}(i), we have $E(\eta_t)=0$ and $E(\eta_t^2)=1$. 
	Moreover, note that $\eta_t$ is independent of $\mathcal{F}_{t-1}$, and  
	$y_{t}-\sum_{i=1}^{p} \alpha_{i} y_{t-i}=\sum_{i=1}^{p}(\alpha_{i0}-\alpha_{i})y_{t-i}+\eta_t\left(\omega_0+\sum_{i=1}^{p}\beta_{i0}|y_{t-i}|\right)$, 
	then by Assumption \ref{assum1-compactness} and $c_r$ inequality, it can be verified that
	\begin{align}\label{finite 2}
		&E\sup _{\bm \theta \in \Theta}\left[\frac{\left(y_{t}-\sum_{i=1}^{p} \alpha_{i} y_{t-i}\right)^{2}}{(\omega+\sum_{i=1}^{p} \beta_{i}|y_{t-i}|)^2}\right] \nonumber\\
		\leq &E\sup _{\bm \theta \in \Theta}\left[\left(\sum_{i=1}^{p}\dfrac{\left(\alpha_{i0}-\alpha_{i}\right) y_{t-i}}{\underline{\omega}+\underline{\beta}\sum_{i=1}^{p} |y_{t-i}|}\right)^{2}\right]
		+ E \left[\left(\dfrac{\overline{\omega}+\bar{\beta}\sum_{i=1}^{p} |y_{t-i}|}{\underline{\omega}+\underline{\beta}\sum_{i=1}^{p} |y_{t-i}|}\right)^2\right]  \nonumber\\
		\leq & E\sup _{\bm \theta \in \Theta}\left[\left(\sum_{i=1}^{p}\dfrac{\left(\alpha_{i0}-\alpha_{i}\right) y_{t-i}}{\underline{\beta} |y_{t-i}|}\right)^{2}\right] \nonumber\\
		&+
		2E\left[\left(\dfrac{\overline{\omega}}{\underline{\omega}+\underline{\beta}\sum_{i=1}^{p} |y_{t-i}|}\right)^2\right]
		+2E\left[ \left(\sum_{i=1}^{p}\dfrac{\bar{\beta}|y_{t-i}|}{\underline{\omega}+\underline{\beta}\sum_{i=1}^{p} |y_{t-i}|}\right)^2\right] \nonumber\\
		\leq & \dfrac{p}{\underline{\beta}^2}\sum_{i=1}^{p}\sup _{\bm \theta \in \Theta}(\alpha_{i0}-\alpha_{i})^2 +\dfrac{2\overline{\omega}^2}{\underline{\omega}^2}+\dfrac{2p^2\bar{\beta}^2}{\underline{\beta}^2}<\infty. 
	\end{align}
	By \eqref{finite 1}, \eqref{finite 2} and the triangle inequality, we have
	\[E\sup_{\bm \theta \in \Theta}\left|\ell^G_{t}(\bm \theta)\right|
	\leq E\sup_{\bm\theta \in \Theta}\left|\ln h_t(\bm\delta)\right|
	+ \dfrac{1}{2}E\sup_{\bm \theta \in \Theta}\dfrac{\varepsilon^2_{t}(\bm\alpha)}{h^2_t(\bm\delta)}<\infty.\]
	Thus, (i) is verified. 
	Similarly, we can show that (ii) and (iii) hold. 
\end{proof}

\begin{lemma}\label{three o_p}
	If Assumptions \ref{assum1-compactness} and \ref{assum2-stationarity} hold, then 
	\begin{equation*}
		\begin{aligned}
			&\text { (i) } \sup_{\bm \theta \in \Theta}\left|L^G_n(\bm \theta)-E \left[\ell^G_{t}(\bm \theta)\right]\right|=o_{p}(1); \\ 
			&\text{ (ii) } \sup_{\bm \theta \in \Theta}\left\|\dfrac{\partial L^G_n(\bm \theta)}{\partial \bm \theta}-E\left[\dfrac{\partial \ell^G_{t}(\bm \theta)}{\partial \bm \theta}\right]\right\|=o_{p}(1); \\ 
			&\text { (iii) } \sup _{\bm \theta \in \Theta}\left\|\frac{\partial^{2} L^G_{n}(\bm \theta)}{\partial \bm \theta \partial \bm \theta^{\prime}}-E\left[\frac{\partial^{2} \ell^G_{t}(\bm \theta)}{\partial \bm \theta \partial \bm \theta^{\prime}}\right]\right\|=o_{p}(1).\\
		\end{aligned}
	\end{equation*}
\end{lemma}

\begin{proof}
	These follow from Lemma \ref{boundedness} and Theorem 3.1 in \cite{Ling_McAleer2003}.
\end{proof}

\begin{lemma}\label{maxmum at true parm}
	If Assumptions \ref{assum1-compactness}--\ref{assum3G-QMLE-moment and density}(i) hold, then $E[\ell^G_t(\bm \theta)]$ has a unique minimum at $\bm \theta_0$. 
\end{lemma}
\begin{proof}
	We first prove that	
	\begin{equation}
		\bm c_1 = \bm 0_p\quad \text{if} \quad \bm c^{\prime}_1 \bm Y_{1t} = 0\ a.s. \quad \text{and}\quad  \bm c_2 = \bm 0_{p+1}\quad \text{if} \quad \bm c^{\prime}_2 \bm Y_{2t} = 0\ a.s.,
		\label{constant Vector}
	\end{equation}
	where $\bm c_1$ and $\bm c_2$ are $p\times 1$ and $(p+1)\times 1$ constant vectors, respectively.
	If $\bm c^{\prime}_1 \bm Y_{1t} = 0\ a.s.$ and $\bm c_1=(c_{1},\cdots,c_{p})^{\prime} \neq \bm 0$, without loss of generality, we can assume $c_{1}=1$, thus $y_t = -\sum_{i=2}^{p}c_{i}y_{t-i+1}\ a.s.$. Recall that $\eta_{t}=\varepsilon_t(\bm \alpha_0)/h_t(\bm \delta_0)$ with $\varepsilon_t(\bm \alpha_0), h_t(\bm \delta_0)\in \mathcal{F}_{t-1}$, and $\eta_t$ is independent of $\mathcal{F}_{t-1}$, we have
	\begin{equation}\label{contradiction}
		E(\eta^2_t)=E(\eta_t)E\left(\dfrac{-\sum_{i=2}^{p}c_{1i}y_{t-i+1}-\varepsilon_t(\bm \alpha_0)}{h_t(\bm\delta_0)}\right)= 0,
	\end{equation}
	which is a contradiction with $E(\eta^2_t)=1$ by Assumption \ref{assum3G-QMLE-moment and density}(i), thus $\bm c_1=\bm 0_p$. Similarly, we can show that $\bm c_2 = \bm 0_{p+1}$.
	
	Note that $\varepsilon_t(\bm \alpha)=\varepsilon_t(\bm \alpha_0)-(\bm \alpha-\bm \alpha_0)^\prime \bm Y_{1t}h_t(\bm \delta_0)$ and $\varepsilon_{t}(\bm \alpha_{0})=\eta_t h_t(\bm \delta_0)$. As for \eqref{finite 2}, then by the law of iterated expectations and Assumption \ref{assum3G-QMLE-moment and density}(i), we can show that
	\begin{align}
		E[\ell^G_{t}(\bm \theta)]
		=&E\left\{\ln h_t(\bm \delta)+\dfrac{\left[\varepsilon_t(\bm \alpha_0)-(\bm \alpha-\bm \alpha_0)^\prime \bm Y_{1t}h_t(\bm \delta_0)\right]^2}{2h^2_t(\bm\delta)}\right\}\notag \\
		=&E\left[\ln h_t(\bm \delta)+\dfrac{1}{2}\left(\dfrac{h_t(\bm \delta_{0})}{h_t(\bm \delta)}\right)^2\right]-\dfrac{1}{2} E\left[\dfrac{\left(\bm \alpha-\bm \alpha_{0}\right)^{\prime} \bm Y_{1t}h_t(\bm \delta_0)}{h_t(\bm \delta)}\right]^2.
		\label{E_ell}
	\end{align}
	The second term in \eqref{E_ell} reaches its minimum at zero, and this happens if and only if $(\bm \alpha-\bm \alpha_0)^{\prime}\bm Y_{1t} = 0\ a.s.$, which holds if and only if $\bm \alpha = \bm \alpha_{0}$ by \eqref{constant Vector}. 
	For the first term in \eqref{E_ell}, denote $f(x) = -\ln(x) - 0.5a^2/x^2,$ where $x = h_t(\bm \delta)$ and $a = h_t(\bm \delta_0)$. We can prove that $f(x)$ reaches its minimum at $x = a$, i.e. $h_t(\bm \delta) = h_t(\bm \delta_{0})$, which holds if and only if $\bm \delta = \bm \delta_{0}$ by \eqref{constant Vector}. 	
	Therefore, $E[\ell^G_t(\bm \theta)]$ is uniquely minimized at $\bm \theta_0$.
\end{proof}
\begin{lemma} \label{MCLT}
	Suppose Assumptions \ref{assum1-compactness}--\ref{assum3G-QMLE-moment and density} hold, then
	\item[(i)] $\Omega_1$ and $\Sigma_1$ are finite and positive definite;
	\item[(ii)] $\sqrt{n}\partial L_{n}^G\left(\bm \theta_{0}\right)/\partial \bm \theta \to_{\mathcal{L}} N(0,\Omega_1)$ as $n\to\infty$.
\end{lemma}
\begin{proof}
	We first show (i). Recall that $\kappa_3=E\eta_t^3$, $\kappa_4=E\eta_t^4-1$, 
	\[
	\Sigma_1 = \diag \left \{E(\bm Y_{1t}\bm Y_{1t}^{\prime}),2E(\bm Y_{2t}\bm Y_{2t}^{\prime})\right \} \; \text{and} \;
	\Omega_1 = \begin{pmatrix} E(\bm Y_{1t}\bm Y_{1t}^{\prime}) & \kappa_3 E(\bm Y_{1t}\bm Y_{2t}^{\prime}) \\ \kappa_3 E(\bm Y_{2t}\bm Y_{1t}^{\prime}) & \kappa_4 E(\bm Y_{2t}\bm Y_{2t}^{\prime}) \end{pmatrix}.
	\]
	By Assumptions \ref{assum1-compactness}--\ref{assum2-stationarity}, for some constant $C$, we have
	\[
	||E(\bm Y_{1t} \bm Y_{1t}^{\prime})||<C, \; 
	||E(\bm Y_{1t} \bm Y_{2t}^{\prime})||<C \; \text{and} \; 
	||E(\bm Y_{2t} \bm Y_{2t}^\prime)||<C.
	\]
	Thus, $\Sigma_1$ is finite. Since $E(\eta_t^4)<\infty$ by Assumption \ref{assum3G-QMLE-moment and density}(ii), we have $\kappa_1,\kappa_2<\infty$, then $\Omega_1$ is also finite. 
	
	Let $\bm x = (\bm x^{\prime}_1,\bm x^{\prime}_2)^{\prime}$, where $\bm x_1\in\mathbb{R}^p$ and $\bm x_2\in\mathbb{R}^{p+1}$ are arbitrary non-zero constant vectors. It follows that
	\begin{align}
		\bm x^{\prime}\Omega_1 \bm x 
		&= E\left\lbrace(\bm x^{\prime}_1\bm Y_{1t})^2+\kappa_4(\bm x^{\prime}_2\bm Y_{2t})^2+2\kappa_3\bm x^{\prime}_2\bm Y_{2t}\bm Y^{\prime}_{1t}\bm x_1\right\rbrace \notag\\
		&= E\left\lbrace(\bm x^{\prime}_1\bm Y_{1t}+\kappa_3\bm x^{\prime}_2\bm Y_{2t})^2+(\kappa_4-\kappa_3^2)(\bm x^{\prime}_2\bm Y_{2t})^2\right\rbrace. \label{positive Omega}
	\end{align}
	By Cauchy-Schwarz inequality and Assumption \ref{assum3G-QMLE-moment and density}(i), we have $\kappa_3^2=\left[\cov(\eta_t,\eta^2_t)\right]^2\leq \var(\eta_t)\var(\eta^2_t)=\kappa_4$, and the equality holds when $P(\eta_t^2-c\eta_t=1)=1$ for any $c\in \mathbb{R}$, which is equivalent to $\det(D)=0$. 
	Since $D$ is positive definite, we have $\kappa_4-\kappa_3^2>0$ and thus $\bm x^{\prime}\Omega_1 \bm x >0$, i.e. $\Omega_1$ is positive definite. Moreover, by \eqref{constant Vector}, it can be verified that
	\begin{align*}
		\bm x^{\prime}\Sigma_1 \bm x 
		&= E\left\lbrace(\bm x^{\prime}_1\bm Y_{1t})^2+2(\bm x^{\prime}_2\bm Y_{2t})^2\right\rbrace>0.
	\end{align*}
	As a result, $\Sigma_1$ is positive definite. Hence, (i) holds. 
	
	Note that $\Omega_1= E\left[\partial \ell_{t}\left(\bm \theta_{0}\right)/\partial \bm \theta \partial \ell_{t}\left(\bm \theta_{0}\right)/\partial \bm \theta^{\prime}\right]$. By the Martingale Central Limit Theorem and the Cram\'{e}r-Wold device, we can show that (ii) holds. 
\end{proof}

\begin{proof}[Proof of Theorem \ref{thm1G-QMLE-asymptotics}]
	By Lemma \ref{three o_p}(i) and Lemma \ref{maxmum at true parm}, we have established all the conditions for consistency in Theorem 4.1.1 in \cite{Amemiya1985}, and hence $\widetilde{\bm\theta}_n \rightarrow_{p} \bm\theta_{0}$ as $n\to\infty$. 
	
	By Lemma \ref{three o_p}(iii), for any $\bm\theta=\bm\theta_0+o_p(1)$, we have $\partial^2L_n^G(\bm\theta)/\partial\bm\theta\partial\bm\theta^{\prime}=\Sigma_1+o_p(1)$.  
	By Taylor's expansion and the consistency of $\widetilde{\bm\theta}_n$, then we have
	\begin{equation}\label{QMLE representation}
		\sqrt{n}(\widetilde{\bm\theta}_n-\bm\theta_0) =-\Sigma_1^{-1}\dfrac{1}{\sqrt{n}}\sum_{t=p+1}^{n}\dfrac{\partial \ell^G_t(\bm \theta_0)}{\partial \bm \theta} +o_p(1).
	\end{equation}
	This together with Lemma \ref{MCLT}, we have established all the conditions of Theorem 4.1.3 in \cite{Amemiya1985}, and hence the asymptotic normality follows. 
\end{proof}

\subsection{Proofs of Theorem \ref{thmACF} and Remark \ref{remark-PT-of-GQMLE}}

Recall that the error function is defined as $\eta_{t}(\bm\theta)=\varepsilon_t(\bm \alpha)/h(\bm \delta)$, $\kappa_1=E(\eta_t)$, $\tau_1=E[\sgn(\eta_t)]$, $\tau_2=E(|\eta_t|)$,  $\sigma_1^{2}=\var(\eta_t)$ and $\sigma_2^2=\var(|\eta_t|)$. 
Note that $\eta_{t}=\eta_{t}(\bm \theta_{0})$, $\widetilde{\eta}_{t}=\eta_{t}(\widetilde{\bm\theta}_n)$ and $\widehat{\eta}_{t}=\eta_{t}(\widehat{\bm\theta}_n)$. 
Let $\breve{\bm\theta}_n\in\Theta$ be a $\sqrt{n}-$estimator of $\bm\theta_0$ and have the asymptotic property $\sqrt{n}(\breve{\bm\theta}_n-\bm\theta_0)=-D\sum_{t=1}^n\bm G_t+o_p(1)$, and the residual ACF and absolute residual ACF of model \eqref{LDAR} fitted by $\breve{\bm \theta}_{n}$ can be defined as
\begin{align*}
	&\breve{\rho}_{k}=\dfrac{\sum_{t=p+k+1}^{n}(\breve{\eta}_{t}-\breve{\eta}_1)(\breve{\eta}_{t-k}-\breve{\eta}_1)}{\sum_{t=p+1}^{n}(\breve{\eta}_{t}-\breve{\eta}_1)^{2}} \;\text{and}\; \breve{\gamma}_{k}=\dfrac{\sum_{t=p+k+1}^{n}(|\breve{\eta}_{t}|-\breve{\eta}_2)(|\breve{\eta}_{t-k}|-\breve{\eta}_2)}{\sum_{t=p+1}^{n}(|\breve{\eta}_{t}|-\breve{\eta}_2)^{2}},
\end{align*}
where $\breve{\eta}_1=(n-p)^{-1}\sum_{t=p+1}^{n}\breve{\eta}_t$ and $\breve{\eta}_2=(n-p)^{-1}\sum_{t=p+1}^{n}|\breve{\eta}_t|$. 

We introduce the following lemma to show Theorem \ref{thmACF} and Remark \ref{remark-PT-of-GQMLE} in a unified framework.

\begin{lemma}\label{lemma-PT}
	Denote $\breve{\bm\rho}=(\breve{\rho}_1,\ldots,\breve{\rho}_M)^{\prime}$ and $\breve{\bm\gamma}=(\breve{\gamma}_1,\ldots,\breve{\gamma}_M)^{\prime}$. 
	For any $\breve{\bm\theta}_n\in\Theta$, if the model \eqref{LDAR} is correctly specified and $\sqrt{n}(\breve{\bm\theta}_n-\bm\theta_0)=-D\sum_{t=1}^n\bm G_t+o_p(1)$ holds, then
	\begin{equation*}
		\sqrt{n}(\breve{\bm\rho}^{\prime},\breve{\bm\gamma}^{\prime})^{\prime}\rightarrow_{\mathcal{L}} N\left(0, V G V^{\prime}\right) \;\; \text{as} \;\; n\rightarrow \infty,
	\end{equation*}
	where $G=E(\bm v_{t} \bm v_{t}^{\prime})$ with
	\begin{align*}
		\bm v_t=&\left[(\eta_t-\kappa_1)(\eta_{t-1}-\kappa_1)/\sigma_1^2,\ldots,(\eta_{t}-\kappa_1)(\eta_{t-M}-\kappa_1)/\sigma_1^2,\right.\\
		&\left.(|\eta_{t}|-\tau_2)(|\eta_{t-1}|-\tau_2)/{\sigma}_2^{2},\ldots,(|\eta_{t}|-\tau_2)(|\eta_{t-M}|-\tau_2)/{\sigma}_2^{2},-\bm G_t^\prime D^\prime\right],
	\end{align*}
	and $V=\left(\begin{array}{ccc}
		I_{M} & 0 & U_{\rho}/\sigma_1^2 \\
		0 & I_{M} & U_{\gamma}/{\sigma}_2^{2}
	\end{array}\right)$, $U_{\rho}=(\bm U_{\rho 1}^{\prime},\ldots,\bm U^{\prime}_{\rho M})^\prime$, $U_{\gamma}=(\bm U^{\prime}_{\gamma 1},\ldots,\bm U^{\prime}_{\gamma M})^{\prime}$, and
	$\bm U_{\rho k}=-\left(E\left[ (\eta_{t-k}-\kappa_1) \bm {Y}_{1t}^{\prime}\right], \kappa_1 E\left[(\eta_{t-k}-\kappa_1)\bm Y_{2t}^\prime\right]\right)$, $\bm U_{\gamma k}=-\left(\tau_1E[(|\eta_{t-k}|-\tau_2) \bm{Y}_{1t}^{\prime}], \tau_2E[(|\eta_{t-k}|-\tau_2) \bm{Y}_{2t}^{\prime}]\right)$ for $1\leq k\leq M$.
\end{lemma}

\begin{proof}[Proof of Lemma \ref{lemma-PT}]
	When model \eqref{LDAR} is correctly specified, by the ergodic theorem and the dominated convergence theorem, it can be shown that, 
	as $n\to\infty$, 
	\begin{align*}
		\breve{\eta}_1&=\dfrac{1}{n-p}\sum_{t=p+1}^{n}\breve{\eta}_t \to_p \kappa_1 \quad \text {and} \quad 
		\dfrac{1}{n}\sum_{t=p+1}^{n}\left(\breve{\eta}_{t}-\breve{\eta}_1\right)^{2} \to_p \sigma_1^2, \\
		\breve{\eta}_2&=\dfrac{1}{n-p}\sum_{t=p+1}^{n}|\breve{\eta}_t| \to_p \tau_2 \quad \text {and} \quad 
		\dfrac{1}{n}\sum_{t=p+1}^{n}(|\breve{\eta}_{t}|-\breve{\eta}_2)^2 \to_p \sigma_2^2,
	\end{align*}
	Therefore, it follows that 
	\begin{equation}\label{relationship between hat_acf and tilde_acf}
		\sqrt{n}\left(\breve{\bm \rho}^{\prime}, \breve{\bm \gamma}^{\prime}\right)^{\prime}=\sqrt{n}\left(\check{\bm \rho}^{\prime}, \check{\bm\gamma}^{\prime}\right)^{\prime}+o_{p}(1),
	\end{equation}
	where $\check{\bm\rho}=(\check{\rho}_{1},\ldots,\check{\rho}_{M})^\prime$ and $\check{\bm \gamma}=(\check{\gamma}_{1},\ldots,\check{\gamma}_{M})^\prime$ with $\check{\rho}_{k}=(n\sigma_1^2)^{-1}\sum_{t=p+k+1}^{n}(\breve{\eta}_{t}-\kappa_1)(\breve{\eta}_{t-k}-\kappa_1)$ and $\check{\gamma}_{k}=(n\sigma^2_2)^{-1}\sum_{t=p+k+1}^{n}(|\breve{\eta}_{t}|-\tau_2)(|\breve{\eta}_{t-k}|-\tau_2)$. 
	It can be shown that
	\begin{align*}
		\sigma^2_1\sqrt{n}\check{\rho}_{k}
		=& \dfrac{1}{\sqrt{n}}\sum_{t=p+k+1}^{n}(\eta_{t}-\kappa_1)(\eta_{t-k}-\kappa_1)+ \dfrac{1}{\sqrt{n}}\sum_{t=p+k+1}^{n}A_{1nt}\nonumber\\
		&+\dfrac{1}{\sqrt{n}}\sum_{t=p+k+1}^{n}A_{2nt} + \dfrac{1}{\sqrt{n}}\sum_{t=p+k+1}^{n}A_{3nt},  \\
		\sigma^2_2\sqrt{n}\check{\gamma}_{k}
		=& \dfrac{1}{\sqrt{n}}\sum_{t=p+k+1}^{n}(|\eta_t|-\tau_2)(|\eta_{t-k}|-\tau_2) + \dfrac{1}{\sqrt{n}}\sum_{t=p+k+1}^{n}B_{1nt}\nonumber\\
		&+\dfrac{1}{\sqrt{n}}\sum_{t=p+k+1}^{n}B_{2nt} + \dfrac{1}{\sqrt{n}}\sum_{t=p+k+1}^{n}B_{3nt},
	\end{align*}
	where
	\begin{align*}
		A_{1nt}&=\left(\breve{\eta}_{t}-\eta_{t}\right)\left(\eta_{t-k}-\kappa_1\right), \quad\;
		B_{1nt}=\left(|\breve{\eta}_{t}|-|\eta_t|\right)\left(|\eta_{t-k}|-\tau_2\right); \\
		A_{2nt}&=\left(\eta_{t}-\kappa_1\right)\left(\breve{\eta}_{t-k}-\eta_{t-k}\right),\; B_{2nt}=\left(|\eta_{t}|-1\right)(|\breve{\eta}_{t-k}|-|\eta_{t-k}|);\\
		A_{3nt}&=\left(\breve{\eta}_{t}-\eta_{t}\right)(\breve{\eta}_{t-k}-\eta_{t-k}), \;\; B_{3nt}=\left(|\breve{\eta}_{t}|-|\eta_{t}|\right)(|\breve{\eta}_{t-k}|-|\eta_{t-k}|).
	\end{align*}
	Since $\sqrt{n}(\breve{\bm\theta}_{n} - \bm\theta_{0})=O_p(1)$,
	moreover, note that $\partial\eta_t(\bm\theta_0)/\partial\bm\theta=(-\bm Y_{1t}^\prime,-\eta_{t}\bm Y_{2t}^\prime)^\prime$ by \eqref{Partial_eta_t}. Then by Taylor's expansion and the ergodic theorem, we have 
	\begin{align}\label{A1n}
		\dfrac{1}{\sqrt{n}}\sum_{t=p+k+1}^{n}A_{1nt}=&\dfrac{1}{n}\sum_{t=p+k+1}^{n}(\eta_{t-k}-\kappa_1)\dfrac{\partial\eta_t(\bm \theta_0)}{\partial \bm \theta^\prime}\sqrt{n}(\breve{\bm \theta}_n-\bm \theta_0)+o_p(1) \nonumber\\
		=& -\dfrac{1}{n}\sum_{t=p+k+1}^{n}(\eta_{t-k}-\kappa_1)\bm Y_{1t}^\prime \sqrt{n}(\breve{\bm\alpha}_n-\bm\alpha_{0}) \nonumber\\
		&-\dfrac{1}{n}\sum_{t=p+k+1}^{n}(\eta_{t-k}-\kappa_1)\eta_t\bm Y_{2t}^\prime\sqrt{n}(\breve{\bm\delta}_n-\bm\delta_{0})+o_p(1) \nonumber\\
		=&\bm U_{\rho k}\sqrt{n}(\breve{\bm \theta}_n-\bm \theta_0)+o_p(1),
	\end{align}
	where $\bm U_{\rho k}=-\left(E\left[ (\eta_{t-k}-\kappa_1) \bm {Y}_{1t}^{\prime}\right], \kappa_1 E\left[(\eta_{t-k}-\kappa_1)\bm Y_{2t}^\prime\right]\right)$. 
	Similarly, we can show that
	\begin{align}\label{A2nA3n}
		\dfrac{1}{\sqrt{n}}\sum_{t=p+k+1}^{n}A_{2nt}&=\dfrac{1}{n}\sum_{t=p+k+1}^{n}(\eta_{t}-\kappa_1) \sqrt{n}\left[\eta_{t-k}(\breve{\bm \theta}_{n})-\eta_{t-k}(\bm\theta_0)\right] = o_{p}(1), \; \text{and} \nonumber\\
		\dfrac{1}{\sqrt{n}}\sum_{t=p+k+1}^{n}A_{3nt}&=\dfrac{1}{n} \sum_{t=p+k+1}^{n}(\breve{\eta}_{t}-\eta_{t})\sqrt{n}\left[\eta_{t-k}(\breve{\bm \theta}_{n})-\eta_{t-k}(\bm\theta_0)\right] =o_{p}(1). 
	\end{align}
	As a result, we have
	\begin{equation}\label{rhok rep}
		\sqrt{n}\check{\rho}_{k} = \dfrac{1}{\sqrt{n}}\sum_{t=p+k+1}^{n}\dfrac{(\eta_{t}-\kappa_1)(\eta_{t-k}-\kappa_1)}{\sigma_1^2} + \dfrac{\bm U_{\rho k}}{\sigma_1^2}\sqrt{n}(\breve{\bm \theta}_n-\bm \theta_0)+o_p(1).
	\end{equation}
	We next consider $\check{\gamma}_{k}$. By the identity \eqref{a20}, we have
	\[|\breve{\eta}_{t}|-|\eta_t|=(\breve{\eta}_{t}-\eta_t)\sgn(\eta_t)+2\int_{0}^{-(\breve{\eta}_{t}-\eta_t)}[I(\eta_t\leq s)-I(\eta_t\leq 0)]ds.\]  
	where $\sgn(\eta_t)=I(\eta_t>0)-I(\eta_t<0)$. 
	Then similar to the proof of \eqref{A1n}, we can show that 
	\[
	\dfrac{1}{\sqrt{n}}\sum_{t=p+k+1}^{n}B_{1nt}=\bm U_{\gamma k}\sqrt{n}(\widehat{\bm \theta}_n-\bm \theta_{0})+o_p(1),
	\]
	where $\bm U_{\gamma k}=-\left(\tau_1E[(|\eta_{t-k}|-\tau_2) \bm{Y}_{1t}^{\prime}], \tau_2E[(|\eta_{t-k}|-\tau_2) \bm{Y}_{2t}^{\prime}]\right)$.  
	Similar to the proof of \eqref{A2nA3n}, we can verify that
	$n^{-1/2}\sum_{t=p+k+1}^{n}B_{2nt}=o_p(1)$ and $n^{-1/2}\sum_{t=p+k+1}^{n}B_{3nt}=o_p(1)$. 
	Therefore, we have
	\begin{equation}\label{gammak rep}
		\sqrt{n}\check{\gamma}_{k} = \dfrac{1}{\sqrt{n}}\sum_{t=p+k+1}^{n}\dfrac{(|\eta_t|-\tau_2)(|\eta_{t-k}|-\tau_2)}{\sigma^2_2} + \dfrac{\bm U_{\gamma k}}{\sigma^2_2}\sqrt{n}(\breve{\bm \theta}_n-\bm \theta_0)+o_p(1).
	\end{equation} 
	Combining \eqref{relationship between hat_acf and tilde_acf}, \eqref{rhok rep} and \eqref{gammak rep}, by Theorem \ref{thm2EQMLE-normality}, we can obtain that
	\begin{equation}\label{relationship between hat_acf and true_acf}
		\sqrt{n}(\breve{\bm \rho}^\prime,\breve{\bm \gamma}^\prime)^\prime=V\dfrac{1}{\sqrt{n}}\sum_{t=p+k+1}^{n}\bm v_t+o_p(1).
	\end{equation}
		Then by the martingale central limit theorem and the Cram\'{e}r--Wold device, we have
		\begin{equation*}
			\sqrt{n}(\breve{\bm\rho}^{\prime},\breve{\bm\gamma}^{\prime})^{\prime}\rightarrow_{\mathcal{L}} N\left(0, V G V^{\prime}\right) \;\; \text{as} \;\; n\rightarrow \infty,
		\end{equation*}
		where $G=E(\bm v_{t} \bm v_{t}^{\prime})$. 
		Hence, the proof of this theorem is accomplished. 
	\end{proof}
	
	\begin{proof}[Proof of Theorem \ref{thmACF}]
		Since $\sqrt{n}(\widehat{\bm\theta}_{n} - \bm\theta_{0})=-\Sigma^{-1}/2\sum_{t=1}^n\bm G_t$ with $\bm G_t=(\bm Y_{1t}^\prime[I(\eta_t\leq 0)-I(\eta_t>0)],\bm Y_{2t}^\prime (1-|\eta_t|))^\prime$ by Theorem \ref{thm2EQMLE-normality} and $\tau_1=0$, $\tau_2=1$ by Assumption \ref{assum3EQMLE-moment and density}(i), we have established all the conditions for Lemma \ref{lemma-PT}, and hence Theorem \ref{thmACF} is established.
	\end{proof}

	\begin{proof}[Proof of Remark \ref{remark-PT-of-GQMLE}]
		Since $\sqrt{n}(\widetilde{\bm\theta}_{n} - \bm\theta_{0})=-\Sigma_1^{-1}\dfrac{1}{\sqrt{n}}\sum_{t=p+1}^{n}\dfrac{\partial \ell^G_t(\bm \theta_0)}{\partial \bm \theta} +o_p(1)$ by \eqref{QMLE representation} and $\kappa_1=0$, $\sigma_1^2=1$ by Assumption \ref{assum3G-QMLE-moment and density}(i), we have established all the conditions for Lemma \ref{lemma-PT}, and hence Remark \ref{remark-PT-of-GQMLE} is verified.
	\end{proof}

\newpage
\bibliography{QMELEforLDAR}

\begin{thebibliography}{}

\bibitem[Amemiya, 1985]{Amemiya1985}
Amemiya, T. (1985).
\newblock {\em Advanced econometrics}.
\newblock Harvard University Press.

\bibitem[Aue and Horv\'{a}th, 2011]{Aue2011}
Aue, A. and Horv\'{a}th, L. (2011).
\newblock Quasi-likelihood estimation in stationary and nonstationary
  autoregressive models with random coefficients.
\newblock {\em Statistica Sinica}, 21:973--999.

\bibitem[Bollerslev, 1986]{Bollerslev1986}
Bollerslev, T. (1986).
\newblock Generalized autoregressive conditional heteroscedasticity.
\newblock {\em Journal of Econometrics}, 31:307--327.

\bibitem[Box et~al., 2015]{box2015_TSA}
Box, G.~E., Jenkins, G.~M., Reinsel, G.~C., and Ljung, G.~M. (2015).
\newblock {\em Time series analysis: forecasting and control}.
\newblock John Wiley \& Sons.

\bibitem[Box et~al., 2008]{Box_Jenkins_Reinsel2008}
Box, G. E.~P., Jenkins, G.~M., and Reinsel, G.~C. (2008).
\newblock {\em Time Series Analysis, Forecasting and Control}.
\newblock Wiley, New York, 4th edition.

\bibitem[Christoffersen, 1998]{christoffersen1998evaluating}
Christoffersen, P.~F. (1998).
\newblock Evaluating interval forecasts.
\newblock {\em International economic review}, 39:841--862.

\bibitem[Cryer and Chan, 2008]{Cryer_Chan2008TSA}
Cryer, J.~D. and Chan, K.-S. (2008).
\newblock {\em Time series analysis: with applications in R}.
\newblock Springer Science \& Business Media.

\bibitem[Engle, 1982]{Engle1982}
Engle, R.~F. (1982).
\newblock {Autoregressive conditional heteroscedasticity with estimates of the
  variance of United Kingdom inflation}.
\newblock {\em Econometrica}, 50:987--1007.

\bibitem[Engle and Manganelli, 2004]{Engle2004}
Engle, R.~F. and Manganelli, S. (2004).
\newblock \textsc{CAViaR}: conditional autoregressive value at risk by
  regression quantiles.
\newblock {\em Journal of Business \& Economic Statistics}, 22:367--381.

\bibitem[Francq and Zakoian, 2004]{Francq_Zakoian2004}
Francq, C. and Zakoian, J.-M. (2004).
\newblock Maximum likelihood estimation of pure \textsc{GARCH} and
  \textsc{ARMA}-\textsc{GARCH} processes.
\newblock {\em Bernoulli}, 10:605--637.

\bibitem[Francq and Zakoian, 2019]{Francq2019}
Francq, C. and Zakoian, J.~M. (2019).
\newblock {\em Garch models: structure, statistical inference and financial
  applications, 2nd edition}.
\newblock Wiley.

\bibitem[Fristedt and Gray, 2013]{Fristedt2013}
Fristedt, B.~E. and Gray, L.~F. (2013).
\newblock {\em A modern approach to probability theory}.
\newblock Springer Science \& Business Media.

\bibitem[Huber, 1967]{Huber1967}
Huber, P. (1967).
\newblock The behavior of maximum likelihood estimates under nonstandard
  conditions.
\newblock {\em Proceedings of the Fifth Berkeley Symposium on Mathematical
  Statistics and Probability}, 1:221--233.

\bibitem[Jiang et~al., 2020]{Jiang2020}
Jiang, F., Li, D., and Zhu, K. (2020).
\newblock Non-standard inference for augmented double autoregressive models
  with null volatility coefficients.
\newblock {\em Journal of Econometrics}, 215:165--183.

\bibitem[Kupiec, 1995]{Kupiec1995}
Kupiec, P. (1995).
\newblock Techniques for verifying the accuracy of risk measurement models.
\newblock {\em Journal of Derivatives}, 3:73--84.

\bibitem[Li et~al., 2015]{Li_Ling_Zakoian2015MTDAR}
Li, D., Ling, S., and Zako{\"\i}an, J.-M. (2015).
\newblock Asymptotic inference in multiple-threshold double autoregressive
  models.
\newblock {\em Journal of Econometrics}, 189:415--427.

\bibitem[Li et~al., 2016]{Li2016_TDAR}
Li, D., Ling, S., and Zhang, R. (2016).
\newblock On a threshold double autoregressive model.
\newblock {\em Journal of Business and Economic Statistics}, 34:68--80.

\bibitem[Li and Li, 2005]{Li_Li2005}
Li, G. and Li, W.~K. (2005).
\newblock {Diagnostic checking for time series models with conditional
  heteroscedasticity estimated by the least absolute deviation approach}.
\newblock {\em Biometrika}, 92:691--701.

\bibitem[Li et~al., 2017]{Li_Zhu_Liu_Li2017}
Li, G., Zhu, Q., Liu, Z., and Li, W.~K. (2017).
\newblock On mixture double autoregressive time series models.
\newblock {\em Journal of Business and Economic Statistics}, 35:306--317.

\bibitem[Li, 2004]{Li2004}
Li, W.~K. (2004).
\newblock {\em Diagnostic Checks in Time Series}.
\newblock Chapman and Hall, Boca Raton.

\bibitem[Li et~al., 2002]{Li_Ling_McAleer2002}
Li, W.~K., Ling, S., and McAleer, M. (2002).
\newblock {Recent theoretical results for time series models with
  \textsc{GARCH} errors}.
\newblock {\em Journal of Economic Surveys}, 16:245--269.

\bibitem[Li and Mak, 1994]{Li1994squared}
Li, W.~K. and Mak, T. (1994).
\newblock On the squared residual autocorrelations in non-linear time series
  with conditional heteroskedasticity.
\newblock {\em Journal of Time Series Analysis}, 15:627--636.

\bibitem[Ling, 2004]{Ling2004}
Ling, S. (2004).
\newblock Estimation and testing stationarity for double-autoregressive models.
\newblock {\em Journal of the Royal Statistical Society, Series B}, 66:63--78.

\bibitem[Ling, 2007a]{Ling2007}
Ling, S. (2007a).
\newblock A double \textsc{AR}(p) model: structure and estimation.
\newblock {\em Statistica Sinica}, 17:161--175.

\bibitem[Ling, 2007b]{Ling2007Self}
Ling, S. (2007b).
\newblock Self-weighted and local quasi-maximum likelihood estimators for
  \textsc{ARMA-GARCH/IGARCH} models.
\newblock {\em Journal of Econometrics}, 140:849--873.

\bibitem[Ling and McAleer, 2003]{Ling_McAleer2003}
Ling, S. and McAleer, M. (2003).
\newblock Asymptotic theory for a vector \textsc{ARMA--GARCH} model.
\newblock {\em Econometric Theory}, 19:280--310.

\bibitem[Ljung and Box, 1978]{Ljung_Box1978}
Ljung, G.~M. and Box, G. E.~P. (1978).
\newblock On a measure of lack of fit in time series models.
\newblock {\em Biometrika}, 65:297--303.

\bibitem[Luo et~al., 2021]{Luo_Zhu2021}
Luo, D., Zhu, K., Gong, H., and Li, D. (2021).
\newblock Testing error distribution by kernelized stein discrepancy in
  multivariate time series models.
\newblock {\em Journal of Business \& Economic Statistics}, 0:1--15.

\bibitem[Machado, 1993]{machado1993robust}
Machado, J.~A. (1993).
\newblock Robust model selection and \textsc{M}-estimation.
\newblock {\em Econometric Theory}, 9:478--493.

\bibitem[Machado, 1990]{machado1990model}
Machado, J. A.~F. (1990).
\newblock Model selection: Consistency and robustness properties of the schwarz
  information criterion for generalized \textsc{M}-estimation.
\newblock Ph.D. thesis, University of Illinois, Urbana-Champaign.

\bibitem[Pollard, 1985]{pollard1985new}
Pollard, D. (1985).
\newblock New ways to prove central limit theorems.
\newblock {\em Econometric Theory}, 1:295--313.

\bibitem[Poskitt and Tremayne, 1983]{Poskitt_Tremayne1983BICinTS}
Poskitt, D.~S. and Tremayne, A.~R. (1983).
\newblock On the posterior odds of time series models.
\newblock {\em Biometrika}, 70:157--162.

\bibitem[Schwarz, 1978]{Schwarz1978}
Schwarz, G. (1978).
\newblock Estimaing the dimentions of a model.
\newblock {\em Annals of Statistics}, 6:461--464.

\bibitem[Serfling, 2009]{Serfling2009}
Serfling, R.~J. (2009).
\newblock {\em Approximation theorems of mathematical statistics}, volume 162.
\newblock John Wiley \& Sons.

\bibitem[Tan and Zhu, 2022]{Tan_Zhu2022}
Tan, S. and Zhu, Q. (2022).
\newblock Asymmetric linear double autoregression.
\newblock {\em Journal of Time Series Analysis}, 43:371--388.

\bibitem[Taylor, 2008]{Taylor2008}
Taylor, S.~J. (2008).
\newblock {\em Modelling financial time series}.
\newblock World Scientific, New York.

\bibitem[Tsay, 2005]{tsay2005AoFTS}
Tsay, R.~S. (2005).
\newblock {\em Analysis of financial time series}.
\newblock John wiley \& sons.

\bibitem[Urquhart, 2016]{Urquhart_2016}
Urquhart, A. (2016).
\newblock The inefficiency of bitcoin.
\newblock {\em Economics Letters}, 148:80--82.

\bibitem[Weiss, 1986]{weiss1986asymptotic}
Weiss, A.~A. (1986).
\newblock Asymptotic theory for \textsc{ARCH} models: Estimation and testing.
\newblock {\em Econometric Theory}, 2:107--131.

\bibitem[Wong and Ling, 2005]{Wong2005mixed_portmanteau}
Wong, H. and Ling, S. (2005).
\newblock Mixed portmanteau tests for time-series models.
\newblock {\em Journal of Time Series Analysis}, 26:569--579.

\bibitem[Xiao and Koenker, 2009]{Xiao_Koenker2009}
Xiao, Z. and Koenker, R. (2009).
\newblock Conditional quantile estimation for generalized autoregressive
  conditional heteroscedasticity models.
\newblock {\em Journal of the American Statistical Association},
  104:1696--1712.

\bibitem[Zhu and Ling, 2011]{Zhu_Ling2011}
Zhu, K. and Ling, S. (2011).
\newblock Global self-weighted and local quasi-maximum exponential likelihood
  estimators for \textsc{ARMA}-\textsc{GARCH/IGARCH} models.
\newblock {\em The Annals of Statistics}, 39:2131--2163.

\bibitem[Zhu and Ling, 2013]{zhu2013quasi-maximum}
Zhu, K. and Ling, S. (2013).
\newblock Quasi-maximum exponential likelihood estimators for a double ar(p)
  model.
\newblock {\em Statistica Sinica}, 23:251--270.

\bibitem[Zhu et~al., 2018]{Zhu2018_LDAR}
Zhu, Q., Zheng, Y., and Li, G. (2018).
\newblock Linear double autoregression.
\newblock {\em Journal of Econometrics}, 207:162--174.

\end{thebibliography}
\clearpage
\begin{table}
	\caption{\label{tableE-QMLE} Biases ($\times 10$), ESDs, and ASDs of the E-QMLE $\widehat{\bm \theta}_n$ and G-QMLE $\widetilde{\bm \theta}_n$ when the innovations follow the normal, Laplace, or Student's $t_{3}$ distribution.}
	\begin{center}
		
		\begin{tabular}{crrrrrrrrrrrr}
			\hline\hline
			&&\multicolumn{3}{c}{Normal}&&\multicolumn{3}{c}{Laplace}&&\multicolumn{3}{c}{$t_3$}\\
			\cline{3-5}\cline{7-9}\cline{11-13}
			&$n$&\multicolumn{1}{c}{Bias}&\multicolumn{1}{c}{ESD}&\multicolumn{1}{c}{ASD}&&\multicolumn{1}{c}{Bias}&\multicolumn{1}{c}{ESD}&\multicolumn{1}{c}{ASD}&&\multicolumn{1}{c}{Bias}&\multicolumn{1}{c}{ESD}&\multicolumn{1}{c}{ASD}\\
			\hline
			&&\multicolumn{11}{c}{E-QMLE}\\
			\hline
			$\alpha$  &	$500$	&	-0.002	&   0.065	&	0.069	
			&&	-0.013	&   0.044	&	0.051	
			&&	-0.039	&	0.052	&	0.056\\
			&	$1000$	&	 0.006	&	0.047	&	0.048	
			&&	-0.009	&	0.031	&	0.036	
			&&	-0.017	&   0.037	&   0.039\\
			$\omega$  &	$500$	&	 0.075	&	0.071	&	0.072	
			&&	 0.063	&	0.088	&	0.088	
			&&	 0.018	&	0.102	&	0.100\\
			&	$1000$	&	 0.034	&	0.050	&	0.051	
			&&	 0.036	&	0.061	&	0.062	
			&&	 0.010	&   0.073	&	0.072\\
			$\beta$   &	$500$	&	-0.059	&	0.045	&	0.045	
			&&  -0.073  &	0.056	&	0.056	
			&&  -0.027  &	0.070	&	0.066\\
			&	$1000$	&	-0.021	&	0.032	&	0.032	
			&&  -0.036  &	0.039	&	0.039	
			&&  -0.018	&	0.049	&	0.047\\
			\hline
			&&\multicolumn{11}{c}{G-QMLE}\\
			\hline
			$\alpha$  &	$500$	&	 0.000	&   0.050	&	0.051	
			&&	-0.001	&   0.055	&	0.054	
			&&	-0.002	&	0.061	&	0.058\\
			&	$1000$	&	 0.000	&	0.036	&	0.036	
			&&	 0.000	&	0.039	&	0.038	
			&&	 0.000	&   0.042	&   0.042\\
			$\omega$  &	$500$	&	 0.048	&	0.064	&	0.063	
			&&	 0.037	&	0.091	&	0.089	
			&&	-0.279	&	0.201	&	0.150\\
			&	$1000$	&	 0.026	&	0.044	&	0.045	
			&&	 0.035	&	0.064	&	0.064	
			&&	-0.169	&   0.162	&	0.125\\
			$\beta$   &	$500$	&	-0.060	&	0.052	&	0.051	
			&&  -0.106  &	0.082	&	0.081	
			&&   0.074  &	0.250	&	0.157\\
			&	$1000$	&	-0.025	&	0.037	&	0.036	
			&&  -0.065  &	0.057	&	0.058	
			&&   0.022	&	0.177	&	0.129\\
			\hline
		\end{tabular}
	\end{center}
\end{table}

\begin{figure}[htp]
	\centering
	\includegraphics[width=5.5in]{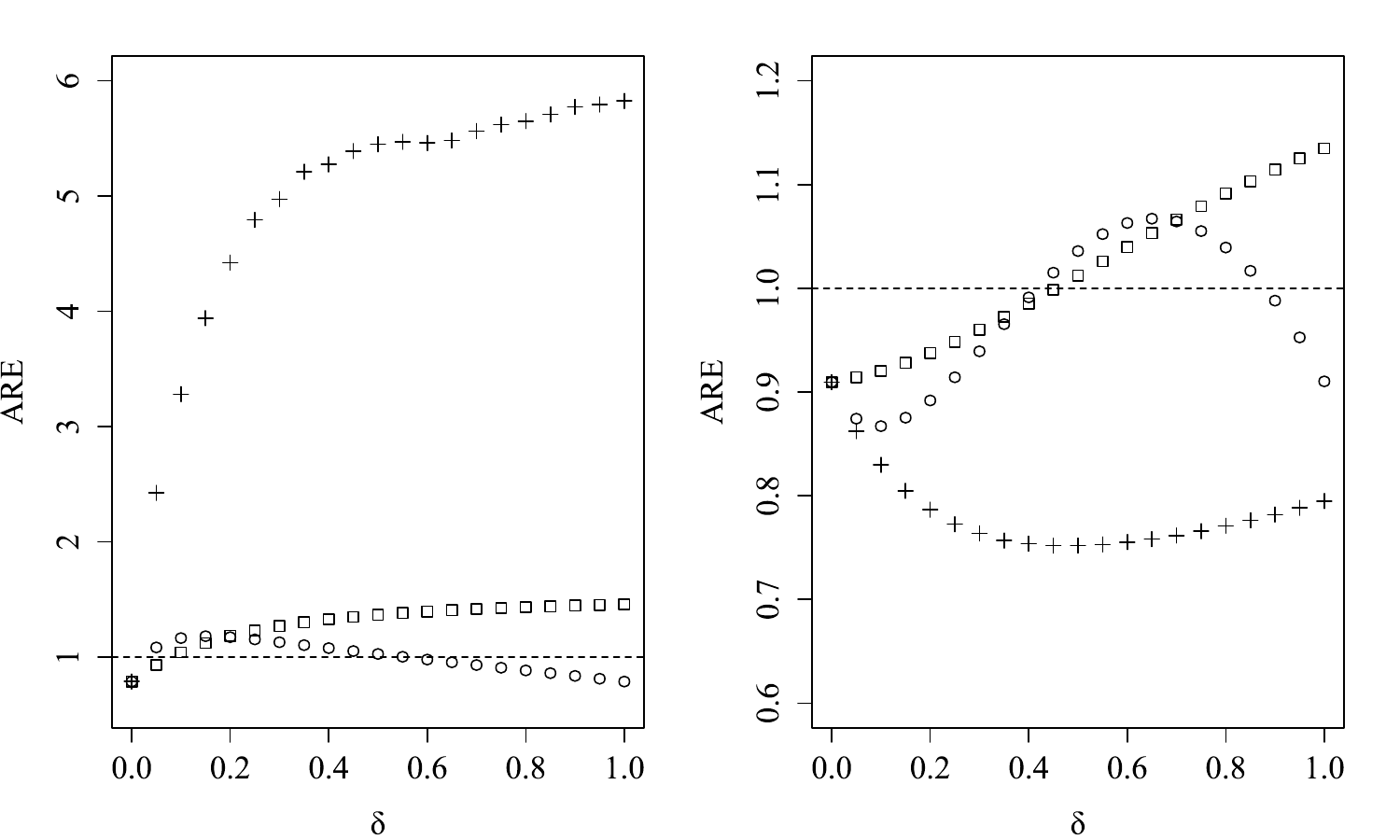}
	\caption{\label{fig_efficiency} 
		The $\text{ARE}(\widehat{\bm\theta}_n,\widetilde{\bm\theta}_n)$ (left panel) and $\text{ARE}(\widehat{\bm\theta}_n^{\star},\check{\bm\theta}^{\star opt}_n)$ (right panel) for $\delta=k/20~(k=0,1,\ldots,20)$, where $m(x)$ is the pdf of the $N(0,6)$ ($\circ$), standard Laplace ($\square$), or $t_3$ ($+$) distribution.}
\end{figure}

\begin{table}
	\caption{\label{tableBIC} Percentages of underfitted, correctly selected, and overfitted models by the BIC and BIC$^G$ when the innovations follow the normal, Laplace, or Student's $t_{3}$ distribution.}
	\begin{center}
		\begin{tabular}{ccccccccccccc}
			\hline\hline
			&&\multicolumn{3}{c}{Normal}&&\multicolumn{3}{c}{Laplace}&&\multicolumn{3}{c}{$t_{3}$}\\
			\cline{3-5}\cline{7-9}\cline{11-13} 
			& $n$ & Under & Exact & Over && Under & Exact & Over && Under & Exact & Over\\\hline
			BIC    &$300$   &40.7 & 58.9 & 0.4 && 12.8 & 87.1 & 0.1 &  & 15.3 & 84.4 & 0.3\\
			&$500$   & 15.1 & 84.5 & 0.4 && 2.1 & 97.4 & 0.5 &  & 3.0 & 96.1 & 0.9\\
			&$1000$   & 1.0 & 99.0 & 0.0 && 0.0 & 100.0 & 0.0 &  & 0.0 & 99.8 & 0.2\\
			BIC$^G$   &$300$   &  33.5  &  66.5  &  0.0 &&  32.2 &  66.6  &  1.2 
			&&  32.7 &  58.2  &  9.1\\
			&$500$   &  10.1  &  89.8  &  0.1 &&  9.8  &  88.4  &  1.8 
			&&  11.8 &  78.1  &  10.1\\
			&$1000$  &  0.0  &  100.0  &  0.0 &&  0.3  &  98.7  &  1.0  
			&&  0.4  &  83.9  &  15.7\\ 			 
			\hline
		\end{tabular}
	\end{center}
\end{table}

\begin{table}
	\caption{\label{tablechecking} Rejection rates of the tests $Q(6)$ and $Q^G(6)$ at the 5\% significance level, where the innovations follow the normal, Laplace, or Student's $t_{3}$ distribution.}	
	\begin{center}	
		\begin{tabular}{lcccccccccc}
			\hline\hline
			&&&\multicolumn{2}{c}{Normal}&&\multicolumn{2}{c}{Laplace}&&\multicolumn{2}{c}{$t_3$}\\\cline{4-5}\cline{7-8}\cline{10-11}
			&$c_1$ & $c_2$ &500&1000&&500&1000&&500&1000\\
			\hline			
			Q	
			&0.0 & 0.0 &0.041 & 0.042 &  & 0.049 & 0.051 &  & 0.049 & 0.048\\
			&0.1 & 0.0 &0.194 & 0.443 &  & 0.216 & 0.462 &  & 0.207 & 0.469\\
			&0.3 & 0.0 &0.996 & 1.000 &  & 0.995 & 1.000 &  & 0.977 & 0.998\\
			&0.0 & 0.1 &0.123 & 0.340 &  & 0.099 & 0.228 &  & 0.076 & 0.173\\
			&0.0 & 0.3 &0.898 & 1.000 &  & 0.633 & 0.975 &  & 0.487 & 0.874\\
			Q$^G$   
			&0.0 & 0.0 &0.058 & 0.047 &  & 0.076 & 0.053 &  & 0.155 & 0.133\\
			&0.1 & 0.0 &0.228 & 0.449 &  & 0.252 & 0.485 &  & 0.340 & 0.532\\
			&0.3 & 0.0 &0.999 & 1.000 &  & 0.999 & 1.000 &  & 0.991 & 1.000\\
			&0.0 & 0.1 &0.106 & 0.194 &  & 0.091 & 0.132 &  & 0.151 & 0.167\\
			&0.0 & 0.3 &0.733 & 0.993 &  & 0.366 & 0.837 &  & 0.338 & 0.630\\
			\hline
		\end{tabular}
	\end{center}
\end{table}			

\begin{figure}[htp]
	\centering
	\includegraphics[width=5.5in]{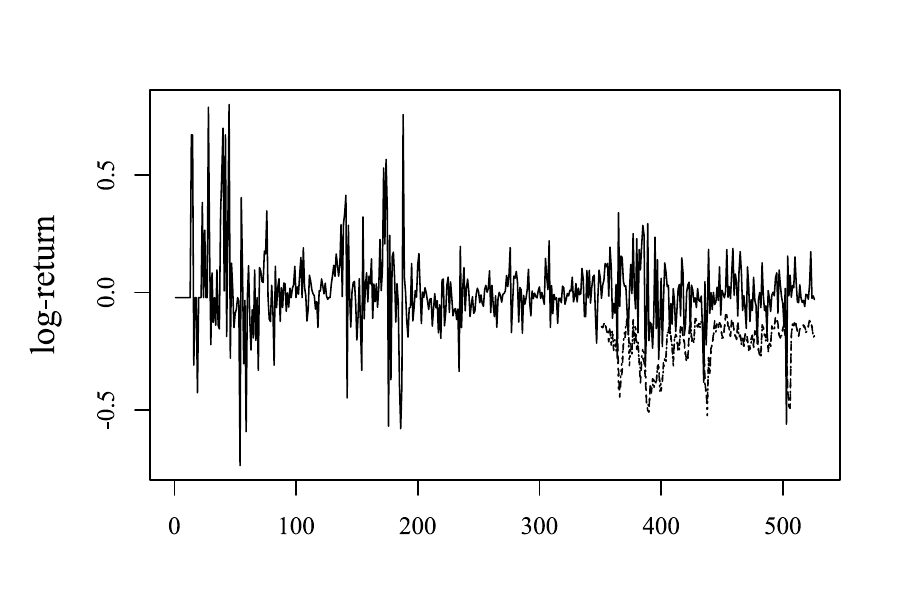}
	\caption{\label{pred}  Time plot for centered weekly log returns in percentage (black line) of BTC from July 18, 2010, to August 16, 2020, with one-week negative VaR forecasts at the level of 5\% based on the G-QMLE (dotted line) and the E-QMLE (dashed line) from March 26, 2017, to August 16, 2020.}
\end{figure}

\begin{figure}[htp]
	\centering
	\includegraphics[width=5.5in]{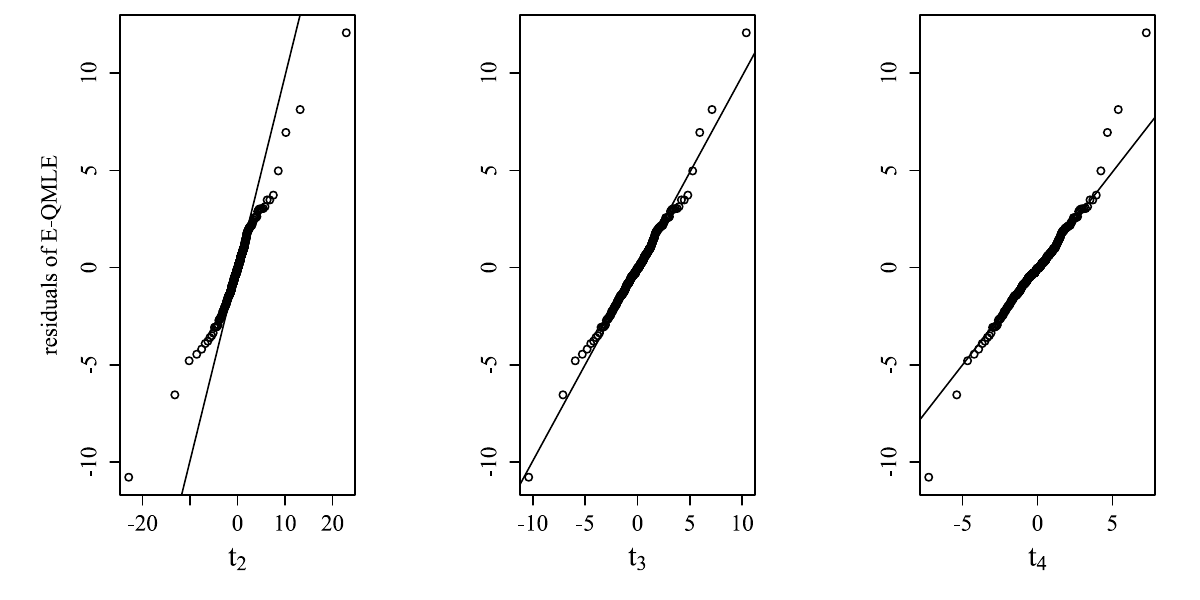}
	\caption{\label{QQplots-E-QMLE}  QQ plots of the residuals $\{\widehat{\eta}_t\}$ against the Student's $t_{2}$ (left panel), $t_{3}$ (middle panel), and $t_{4}$ (right panel) distributions.}
\end{figure}

\begin{figure}[htp]
	\centering
	\includegraphics[width=5.5in]{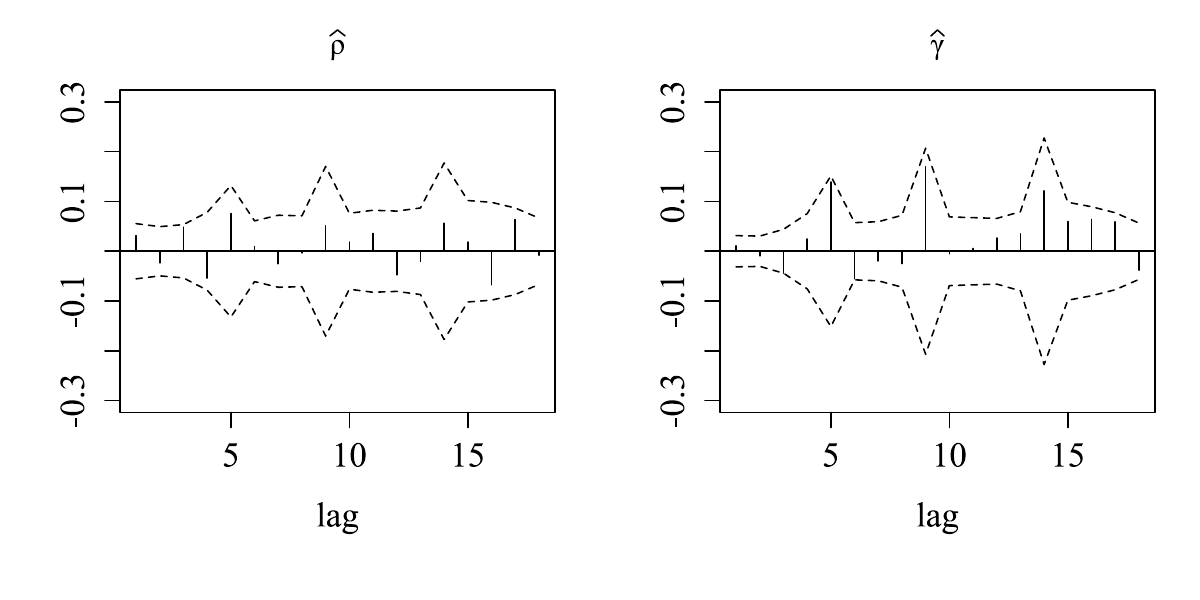} 
	\caption{\label{acf-E-QMLE}  Residual ACF plots for $\widehat{\rho}_l$ (left panel) and $\widehat{\gamma}_l$ (right panel), where the dashed lines are the corresponding 95\% pointwise confidence intervals.}
\end{figure}

\begin{table}
	\caption{\label{forcast} Empirical coverage rates (\%) and $p$-values of three VaR backtests  at the 5\%, 10\%, 90\%, and 95\% conditional quantiles. M1, M2, and M3 represent the linear DAR model fitted using the E-QMLE, G-QMLE, and DWQRE, respectively; M4 and M5 represent the DAR model fitted using the E-QMLE and G-QMLE, respectively; and M6 and M7 represent the AR-GARCH model fitted using the E-QMLE and G-QMLE, respectively. The ECRs closest to the nominal level are marked in bold.}
	\begin{center}
		\resizebox{\textwidth}{!}{
		\begin{tabular}{crrrrrrrrrrrrrrrrrrr}
			\hline\hline 
			& \multicolumn{4}{c}{$\tau=5\%$} && \multicolumn{4}{c}{$\tau=10\%$} && \multicolumn{4}{c}{$\tau=90\%$} && \multicolumn{4}{c}{$\tau=95\%$}\\
			\cline{2-5}\cline{7-10}\cline{12-15}\cline{17-20}
			&\multicolumn{1}{c}{ECR}&\multicolumn{1}{c}{UC}&\multicolumn{1}{c}{CC}&\multicolumn{1}{c}{DQ}&&\multicolumn{1}{c}{ECR}&\multicolumn{1}{c}{UC}&\multicolumn{1}{c}{CC}&\multicolumn{1}{c}{DQ}&&\multicolumn{1}{c}{ECR}&\multicolumn{1}{c}{UC}&\multicolumn{1}{c}{CC}&\multicolumn{1}{c}{DQ}&&\multicolumn{1}{c}{ECR}&\multicolumn{1}{c}{UC}&\multicolumn{1}{c}{CC}&\multicolumn{1}{c}{DQ}\\
			\hline			    
			M1  & 5.68 & 0.68 & 0.79   & 0.13
			&& \textbf{10.23} & 0.92    & 0.99    & 0.48
			&& 88.07 & 0.41    & 0.34    & 0.19
			&& \textbf{94.89} & 0.95    & 0.61     & 0.37\\
			M2  & 6.25  & 0.46   & 0.71    & 0.22
			&& 10.80 & 0.73    & 0.94    & 0.31
			&& 88.64 & 0.55    & 0.48    & 0.29
			&& 93.75 & 0.46    & 0.37     & 0.09\\ 
			M3  & 6.25 & 0.46    & 0.71      & 0.07       
			&& 9.09 & 0.68    & 0.82      & 0.45
			&& \textbf{89.77} &0.92   & 0.99      & 0.06
			&& 94.32 &0.68   & 0.50      & 0.07\\	
			M4    & 5.68 &0.68    &0.50       &0.71
			&& 11.93 &0.41  & 0.67      & 0.54
			&& 86.93 &0.20   & 0.33      & 0.67
			&& 94.32 &0.68   & 0.50     & 0.73\\
			M5    & 6.25 &0.46    &0.37       &0.48
			&& 9.65 &0.88    & 0.83      & 0.48
			&& 87.50&0.29    & 0.23      & 0.56 
			&& 93.75&0.46    & 0.37      & 0.50\\			
			M6    & 6.82 & 0.29    &0.56       &0.14
			&& 9.66  &0.88   & 0.83      & 0.50
			&& 89.20 &0.73   & 0.09      & 0.34
			&& 95.45 &0.79   & 0.66      & 0.71\\
			M7    & \textbf{5.11} &0.95    &0.76       &0.26
			&& 9.66 &0.87    & 0.83      & 0.54
			&& \textbf{89.77}&0.92    & 0.75      & 0.87
			&& 96.02&0.52    & 0.61      & 0.75\\
			\hline
		\end{tabular}}
	\end{center}
\end{table}

\begin{table}[htp]
	\caption{\label{table-QMLE-supp} Biases ($\times 10$), ESDs, and ASDs of the E-QMLE $\widehat{\bm \theta}_n$ and G-QMLE $\widetilde{\bm \theta}_n$ when the innovations follow the $t_2$, $t_5$, or $st_{3,-1.5}$ distribution.}
	\begin{center}
		\begin{tabular}{crrrrrrrrrrrr}
			\hline\hline
			&&\multicolumn{3}{c}{$t_2$}&&\multicolumn{3}{c}{$t_5$}&&\multicolumn{3}{c}{$st_{3,-1.5}$}\\
			\cline{3-5}\cline{7-9}\cline{11-13}
			&$n$&\multicolumn{1}{c}{Bias}&\multicolumn{1}{c}{ESD}&\multicolumn{1}{c}{ASD}&&\multicolumn{1}{c}{Bias}&\multicolumn{1}{c}{ESD}&\multicolumn{1}{c}{ASD}&&\multicolumn{1}{c}{Bias}&\multicolumn{1}{c}{ESD}&\multicolumn{1}{c}{ASD}\\
			\hline
			&&\multicolumn{11}{c}{E-QMLE}\\
			\hline
			$\alpha$	&
			$500$	&-0.034&0.045&0.047&&-0.040&0.059&0.061&&-0.004&0.053&0.056\\
			&	$1000$	&-0.022&0.031&0.033&&-0.011&0.042&0.043&&-0.004&0.037&0.039\\
			$\omega$	&	
			$500$		&0.112&0.138&0.144&&0.066&0.078&0.083&&-0.072&0.090&0.094\\
			&	$1000$	&0.057&0.098&0.108&&0.031&0.055&0.059&&-0.075&0.062&0.068\\
			$\beta$	
			&	$500$	&-0.096&0.115&0.100&&-0.060&0.053&0.053&&-0.043&0.070&0.068\\
			&	$1000$	&-0.051&0.079&0.076&&-0.027&0.038&0.038&&-0.038&0.050&0.049\\
			\hline
			&&\multicolumn{11}{c}{G-QMLE}\\
			\hline
			$\alpha$	&	
			$500$	&-0.056&0.073&0.067&&-0.037&0.054&0.054&&-0.008&0.061&0.058\\
			&	$1000$	&-0.023&0.053&0.049&&-0.012&0.039&0.038&&-0.005&0.043&0.042\\
			$\omega$	&
			$500$	&-0.365&0.249&0.222&&0.036&0.091&0.095&&-0.016&0.181&0.157\\
			&	$1000$	&-0.312&0.213&0.202&&0.020&0.068&0.071&&-0.008&0.145&0.135\\
			$\beta$	&
			$500$	&0.428&0.581&0.298&&-0.072&0.090&0.085&&0.000&0.239&0.165\\
			&	$1000$	&0.479&0.569&0.284&&-0.036&0.067&0.063&&0.002&0.217&0.145\\
			\hline
		\end{tabular}
	\end{center}
\end{table}

\begin{table}[htp]
	\caption{\label{tableBIC-supp} Percentages of underfitted, correctly selected, and overfitted models by the BIC and BIC$^G$ when the innovations follow the $t_2$, $t_5$, or $st_{3,-1.5}$ distribution.}
	\begin{center}
		\begin{tabular}{ccccccccccccc}
			\hline\hline
			&&\multicolumn{3}{c}{$t_2$}&&\multicolumn{3}{c}{$t_5$}&&\multicolumn{3}{c}{$st_{3,-1.5}$}\\
			\cline{3-5}\cline{7-9}\cline{11-13} 
			&$n$&Under&Exact&Over&&Under&Exact&Over&&Under&Exact&Over\\\hline
			BIC	
			&$300$ &7.2&89.7&3.1&&30.9&68.9&0.2&&14.1&85.0&0.9\\
			&$500$ &0.5&95.8&3.7&&7.2&92.7&0.1&&1.5&98.1&0.4\\
			&$1000$&0.0&97.0&3.0&&0.2&99.8&0.0&&0.0&99.6&0.4\\
			BIC$^G$
			&$300$ &29.2&48.6&22.2&&36.6&61.2&2.2&&36.6&54.4&9.0\\
			&$500$ &10.7&59.5&29.8&&11.2&87.5&1.3&&14.6&73.4&12.0\\
			&$1000$&1.3&65.3&33.4&&0.4&98.0&1.6&&0.7&83.2&16.1\\ 		 
			\hline
		\end{tabular}
	\end{center}
\end{table}

\begin{table}[htp]
	\caption{\label{tablechecking-supp} Rejection rates of the tests $Q(6)$ and $Q^G(6)$ at the 5\% significance level, when the innovations follow the $t_2$, $t_5$, or $st_{3,-1.5}$ distribution.}	
	\begin{center}	
		\begin{tabular}{lcccccccccc}
			\hline\hline
			&&&\multicolumn{2}{c}{$t_2$}&&\multicolumn{2}{c}{$t_5$}&&\multicolumn{2}{c}{$st_{3,-1.5}$}\\\cline{4-5}\cline{7-8}\cline{10-11}
			&$c_1$ & $c_2$ &500&1000&&500&1000&&500&1000\\
			\hline			
			Q	
			&0.0&0.0&0.077&0.090&&0.057&0.064&&0.056&0.066\\
			&0.1&0.0&0.253&0.478&&0.203&0.475&&0.278&0.534\\
			&0.3&0.0&0.910&0.985&&0.996&1.000&&0.991&1.000\\
			&0.0&0.1&0.090&0.152&&0.096&0.227&&0.088&0.164\\
			&0.0&0.3&0.284&0.593&&0.734&0.978&&0.409&0.792\\
			Q$^G$
			&0.0&0.0&0.151&0.168&&0.048&0.062&&0.062&0.081\\
			&0.1&0.0&0.367&0.607&&0.182&0.460&&0.298&0.553\\
			&0.3&0.0&0.937&0.994&&0.998&1.000&&0.994&1.000\\
			&0.0&0.1&0.143&0.177&&0.064&0.126&&0.077&0.132\\
			&0.0&0.3&0.231&0.367&&0.420&0.871&&0.270&0.582\\
			\hline
		\end{tabular}
	\end{center}
\end{table}		

\begin{figure}[htp]
	\centering
	\includegraphics[angle=0,width=5.5in]{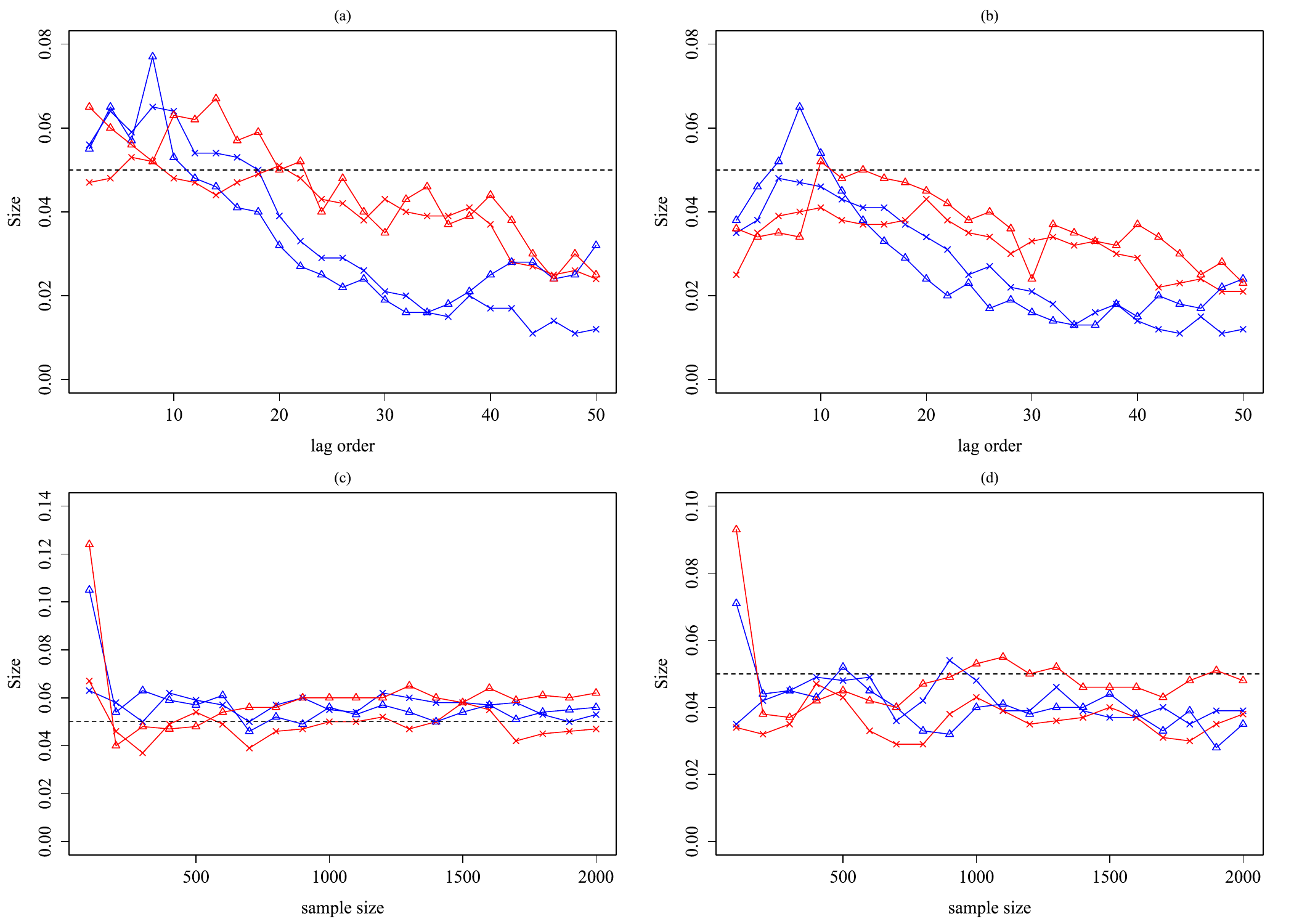}
	\caption{\label{fig_size} The empirical size of $Q(M)$ (a) and $Q^G(M)$ (b) with respect to the lag order $M$ at the significance level 0.05, where the line with color blue or red presents $N=500$ or 2000, the symbol $\times$ or $\triangle$ represents $\eta_t$ being a standard normal or Laplace random variable; 
		the empirical size of $Q(M)$ (c) and $Q^G(M)$ (d) with respect to the sample size $n$ at the significance level 0.05, where the line with color blue or red presents $M=6$ or 12, the symbol $\times$ or $\triangle$ represents $\eta_t$ being a standard normal or Laplace random variable.}
\end{figure}
\begin{figure}[htp]
	\centering
	\includegraphics[angle=0,width=5.5in]{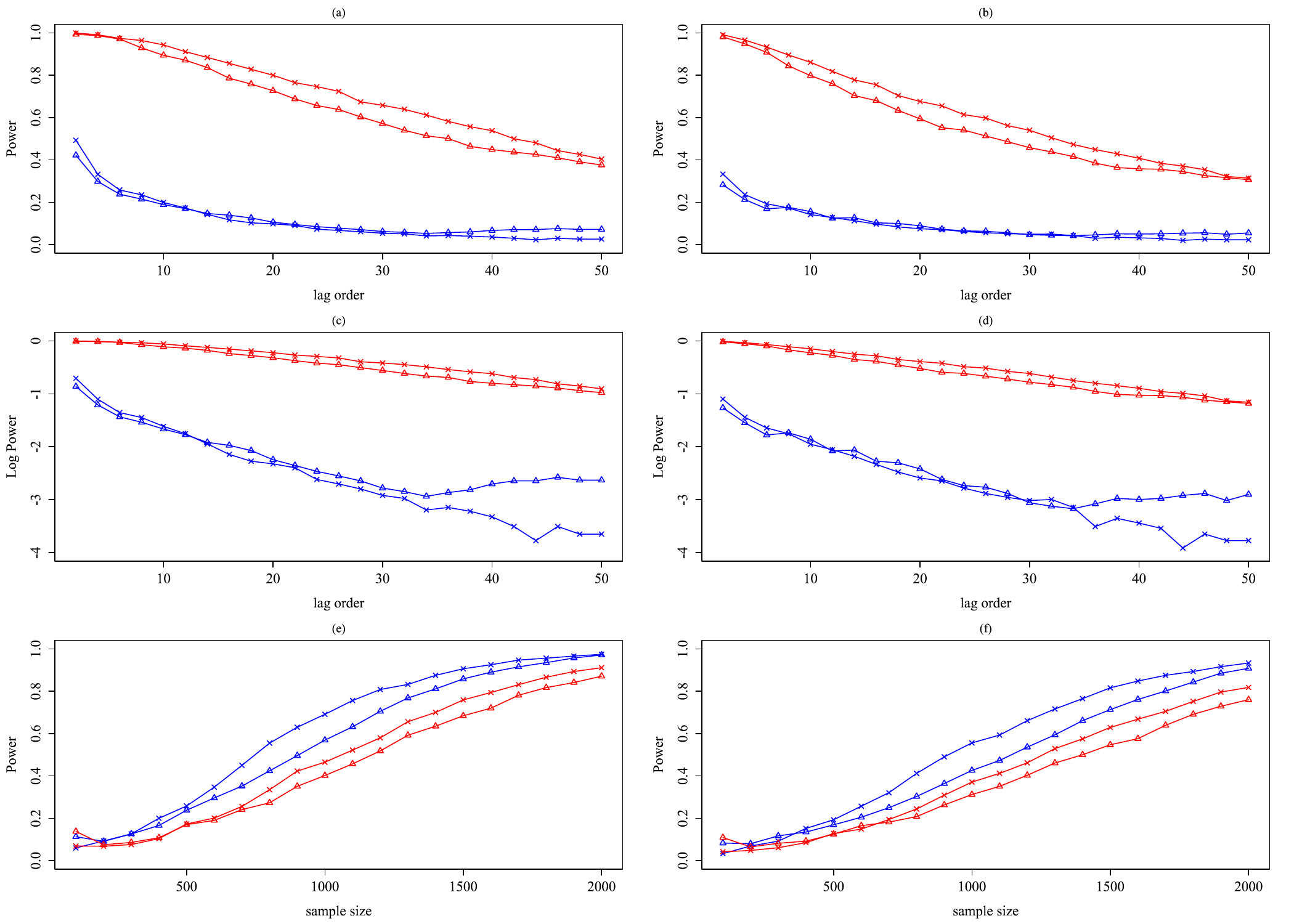}
	\caption{\label{fig_power} The empirical power of $Q(M)$ (a) and $Q^G(M)$ (b) with respect to the lag order $M$ at the significance level 0.05, where the line with color blue or red presents $N=500$ or 2000, the symbol $\times$ or $\triangle$ represents $\eta_t$ being a standard normal or Laplace random variable; (c) and (d) correspond to the logarithmic power version of (a) and (b), respectively; the empirical power of $Q(M)$ (e) and $Q^G(M)$ (f) with respect to the sample size $n$ at the significance level 0.05, where the line with color blue or red presents $M=6$ or 12, the symbol $\times$ or $\triangle$ represents $\eta_t$ being a standard normal or Laplace random variable.}
\end{figure}

\begin{table}[htp]
	\caption{\label{tablechecking-M} Rejection rates of the tests $Q(M)$ and $Q^G(M)$ at the 5\% significance level with $M$=12 or 18, where the innovations follow the normal, Laplace, or Student's $t_{3}$ distribution.}	
	\begin{center}	
		\begin{tabular}{lcccccccccc}
			\hline\hline
			&&&\multicolumn{2}{c}{Normal}&&\multicolumn{2}{c}{Laplace}&&\multicolumn{2}{c}{$t_3$}\\\cline{4-5}\cline{7-8}\cline{10-11}
			&$c_1$ & $c_2$ &500&1000&&500&1000&&500&1000\\
			\hline
			&&&\multicolumn{8}{c}{$M=12$}\\
			\hline		
			Q	
			&0.0 & 0.0 &0.031 & 0.038 &  & 0.042 & 0.055 &  & 0.048 & 0.043\\
			&0.1 & 0.0 &0.140 & 0.334 &  & 0.154 & 0.337 &  & 0.144 & 0.361\\
			&0.3 & 0.0 &0.964 & 1.000 &  & 0.973 & 1.000 &  & 0.946 & 0.999\\
			&0.0 & 0.1 &0.092 & 0.219 &  & 0.073 & 0.129 &  & 0.062 & 0.125\\
			&0.0 & 0.3 &0.712 & 0.997 &  & 0.426 & 0.900 &  & 0.301 & 0.742\\
			Q$^G$   
			&0.0 & 0.0 &0.041 & 0.032 &  & 0.064 & 0.060 &  & 0.130 & 0.106\\
			&0.1 & 0.0 &0.143 & 0.299 &  & 0.179 & 0.383 &  & 0.275 & 0.422\\
			&0.3 & 0.0 &0.981 & 1.000 &  & 0.986 & 1.000 &  & 0.951 & 1.000\\
			&0.0 & 0.1 &0.078 & 0.119 &  & 0.082 & 0.105 &  & 0.130 & 0.143\\
			&0.0 & 0.3 &0.506 & 0.946 &  & 0.230 & 0.626 &  & 0.241 & 0.440\\
			\hline
			&&&\multicolumn{8}{c}{$M=18$}\\
			\hline
			Q
			&0.0 & 0.0 &0.028 & 0.040 &  & 0.031 & 0.043 &  & 0.041 & 0.047\\
			&0.1 & 0.0 &0.096 & 0.244 &  & 0.117 & 0.254 &  & 0.113 & 0.283\\
			&0.3 & 0.0 &0.914 & 1.000 &  & 0.907 & 1.000 &  & 0.855 & 0.999\\
			&0.0 & 0.1 &0.060 & 0.155 &  & 0.063 & 0.108 &  & 0.057 & 0.120\\
			&0.0 & 0.3 &0.510 & 0.988 &  & 0.268 & 0.793 &  & 0.232 & 0.634\\
			Q$^G$  
			&0.0 & 0.0 &0.043 & 0.029 &  & 0.055 & 0.043 &  & 0.124 & 0.091\\
			&0.1 & 0.0 &0.097 & 0.211 &  & 0.136 & 0.282 &  & 0.213 & 0.350\\
			&0.3 & 0.0 &0.924 & 1.000 &  & 0.941 & 1.000 &  & 0.891 & 0.999\\
			&0.0 & 0.1 &0.044 & 0.084 &  & 0.062 & 0.092 &  & 0.129 & 0.121\\
			&0.0 & 0.3 &0.321 & 0.874 &  & 0.158 & 0.493 &  & 0.185 & 0.339\\
			\hline
		\end{tabular}
	\end{center}
\end{table}	

\begin{figure}[htp]
	\centering
	\includegraphics[width=5.5in]{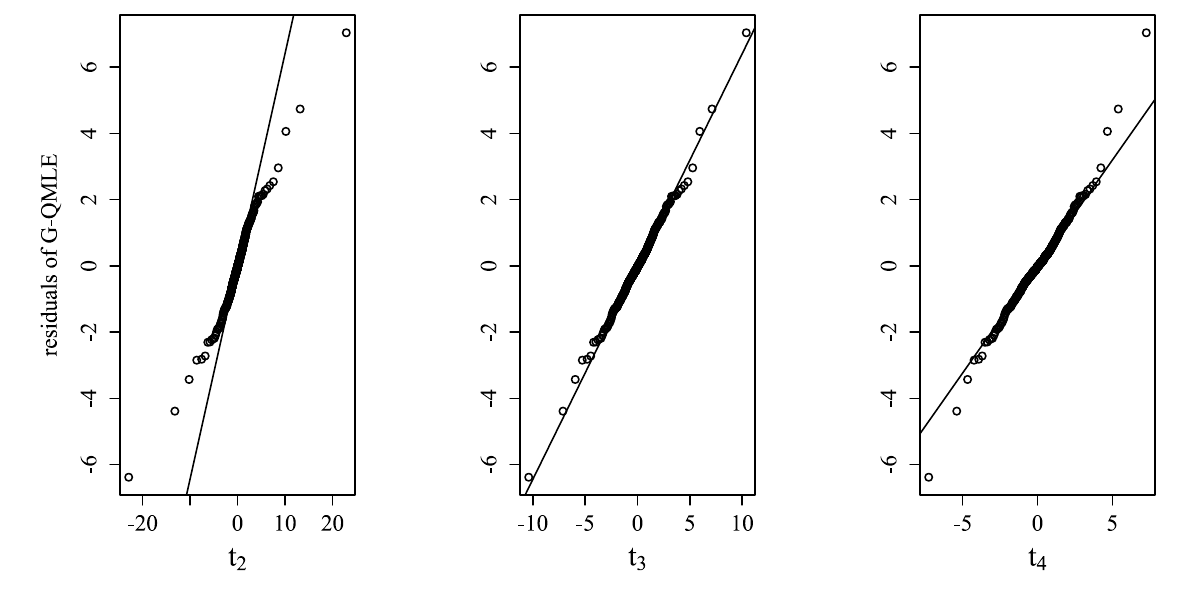}
	\caption{\label{QQplots-G-QMLE}  QQ plots of the residuals $\{\widetilde{\eta}_t\}$ against the Student's $t_{2}$ (left panel), $t_{3}$ (middle panel), and $t_{4}$ (right panel) distributions.}
\end{figure}

\begin{figure}[htp]
	\centering
	\includegraphics[width=5.5in]{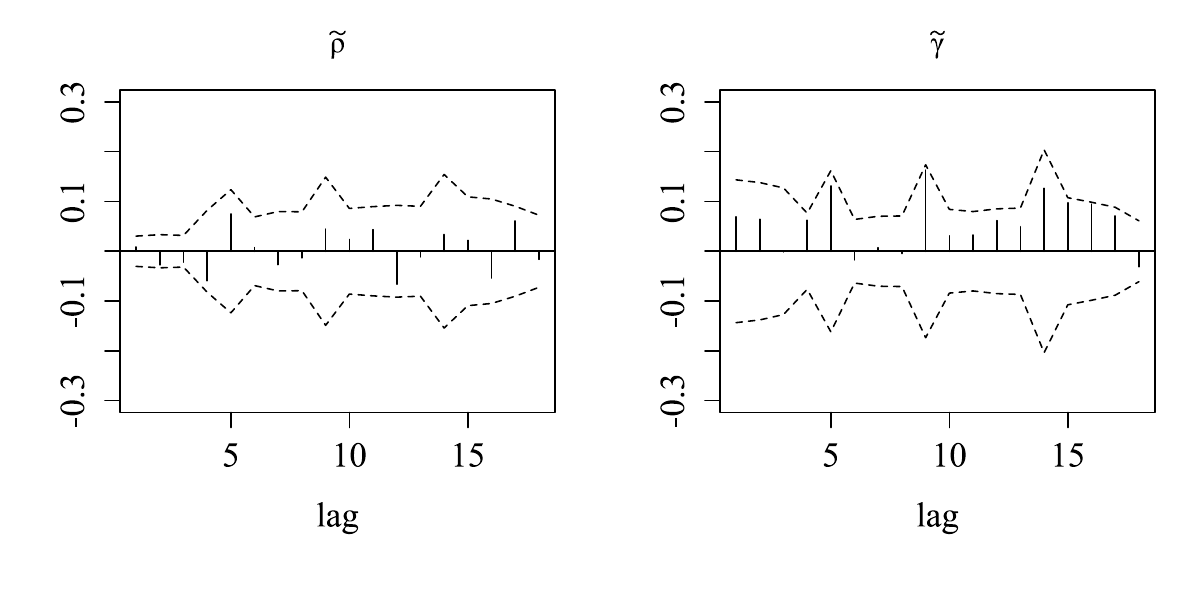}
	\caption{\label{acf-G-QMLE}  Residual ACF plots for $\widetilde{\rho}_l$ (left panel) and $\widetilde{\gamma}_l$ (right panel), where the dashed lines are the corresponding 95\% pointwise confidence intervals.}
\end{figure}

\end{document}